\documentclass[a4paper,twoside,10pt,fleqn]{article}
\usepackage{helvet}

\usepackage[latin1]{inputenc}
\usepackage{amsfonts}
\usepackage{amsmath}
\usepackage{amssymb,marvosym,units}
\usepackage{amsthm}
\usepackage[cmtip,arrow, matrix, curve]{xy}
\usepackage{pb-diagram,pb-xy}
\usepackage{graphicx}
\usepackage{tensor}

\usepackage{hyperref}
\usepackage{zref, xkeyval, ifpdf, ifthen, calc, marginnote, pdfcomment}

\usepackage{lscape}
\usepackage{longtable}
\usepackage{bbm}

\hypersetup{
    pagebackref=true,
    bookmarks=true,        
    unicode=false,         
    pdftoolbar=true,        
    pdfmenubar=true,        
    pdffitwindow=true,     
    pdfstartview={FitH},    
    pdftitle={The localised holonomy-flux cross-product $^*$-algebra for LQG},    
    pdfauthor={Diana Kaminski},     
    pdfsubject={},   
    pdfkeywords={Loop Quantum Gravity, KMS - Theory, localised holonomy-flux cross-product $^*$-algebra}, 
    pdfnewwindow=true,     
    colorlinks=false,       
    linkcolor=black,        
    pdfborder= 0 0 0,
}

\usepackage{lscape}

\title{Algebras of Quantum Variables for Loop Quantum Gravity\\[5pt]
\textbf{V. The localised holonomy-flux cross-product $^*$-algebra}}
\author{Diana Kaminski\\[3pt]
kaminski@math.uni-paderborn.de\\ 
\small{Europe - Germany}}
\date{August 19, 2011}

\newcommand{\Ab}{\begin{large}\bar{\mathcal{A}}\end{large}}

\newcommand{\ha}{\mathfrak{a}}
\newcommand{\Alg}{\begin{large}\mathfrak{A}\end{large}}

\newcommand{\Aut}{\begin{large}\mathfrak{Aut}\end{large}}

\newcommand{\Con}{ \mathfrak{C}}
\newcommand{\CB}{\mathbb{C}}
\newcommand{\CD}{\mathcal{C}}
\newcommand{\Cinf}{\mathbf{C}}

\newcommand{\DD}{\mathcal{D}}
\newcommand{\Df}{\mathfrak{D}}
\newcommand{\E}{\mathcal{E}}
\newcommand{\Ep}{\mathbbm{E}}

\newcommand{\Goid}{\mathcal{G}}

\newcommand{\GG}{G}

\newcommand{\HS}{\mathcal{H}}

\newcommand{\la}{\langle}
\newcommand{\LD}{\mathcal{L}}
\newcommand{\LAb}{\bar{a}}

\newcommand{\MM}{\textbf{M}}

\newcommand{\N}{\mathbb{N}}

\newcommand{\op}{\mathfrak{o}}

\newcommand{\PD}{\mathcal{P}}

\newcommand{\Ss}{\mathcal{S}}

\newcommand{\ra}{\rangle}

\newcommand{\OD}{\mathcal{O}}

\newcommand{\QD}{\mathcal{Q}}
\newcommand{\R}{\mathbb{R}}

\newcommand{\SimGroup}{\mathfrak{G}}

\newcommand{\ZD}{\mathcal{Z}}
\newcommand{\ZzD}{z}

\newcommand{\ho}{\mathfrak{h}}

\newcommand{\go}{\mathfrak{g}}
\newcommand{\gop}{\mathbbm{g}}

\DeclareMathOperator{\dif}{d}
\DeclareMathOperator{\diff}{surf}

\DeclareMathOperator{\disc}{d}
\DeclareMathOperator{\Diff}{Diff}

\DeclareMathOperator{\Exp}{Exp}

\DeclareMathOperator{\Hol}{Hol}
\DeclareMathOperator{\Hom}{Hom}

\DeclareMathOperator{\id}{id}

\DeclareMathOperator{\loc}{loc}

\DeclareMathOperator{\Map}{Map}

\DeclareMathOperator{\Rep}{Rep}

\DeclareMathOperator{\tr}{tr}

\newcommand{\limG}{\lim_{\Gamma_i\rightarrow \Gamma_\infty}}

\newcommand{\gp}{{\gamma^\prime}}
\newcommand{\gpi}{{\gamma^\prime_i}}
\newcommand{\gpj}{{\gamma^\prime_j}}
\newcommand{\gpk}{{\gamma^\prime_K}}

\newcommand{\gppi}{{\gamma^{\prime\prime}_i}}
\newcommand{\gppj}{{\gamma^{\prime\prime}_j}}
\newcommand{\gppje}{{\gamma^{\prime\prime}_{j+1}}}
\newcommand{\gppl}{{\gamma^{\prime\prime}_l}}
\newcommand{\gpe}{{\gamma^\prime_1}}
\newcommand{\gpz}{{\gamma^\prime_2}}
\newcommand{\gpm}{{\gamma^\prime_M}}
\newcommand{\gppe}{{\gamma^{\prime\prime}_1}}

\newcommand{\gppm}{{\gamma^{\prime\prime}_M}}

\newcommand{\gpppn}{{\gamma^{\prime\prime\prime}_N}}
\newcommand{\tg}{{\tilde\gamma}}

\newcommand{\Gp}{{\Gamma^\prime}}
\newcommand{\Gpj}{{\Gamma^\prime_j}}
\newcommand{\Gpi}{{\Gamma^\prime_i}}
\newcommand{\Gpp}{\Gamma^{\prime\prime}}
\newcommand{\Gppp}{\Gamma^{\prime\prime\prime}}

\newcommand{\idf}{\mathbbm{1}}

\newcommand{\bra}{[}
\newcommand{\ket}{]}

\newcommand{\beq}{\begin{equation}\begin{aligned}}
\newcommand{\beqs}{\begin{equation*}\begin{aligned}}
\newcommand{\be}{\begin{flalign}}
\newcommand{\bes}{\begin{equation*}}
\newcommand{\eq}{\end{aligned}\end{equation}}
\newcommand{\eqs}{\end{aligned}\end{equation*}}
\newcommand{\ee}{\end{flalign}}
\newcommand{\ees}{\end{equation}}

\newcommand{\limN}{\lim_{N\rightarrow \infty}}

\newcommand{\limi}{\underset{i\rightarrow\infty}{\underrightarrow{\lim}}}

\newtheorem{theo}{Theorem }[section]
\newtheorem{lem}[theo]{Lemma}
\newtheorem{rem}[theo]{Remark}
\newtheorem{prop}[theo]{Proposition}
\newtheorem{cor}[theo]{Corollary}

\newtheorem{defi}[theo]{Definition}

\newenvironment{proofs}[1][Proof ]{\noindent\textbf{#1}: }{\ \begin{flushright}
                                                                         \rule{0.5em}{0.5em}
                                                                        \end{flushright}}

\newcounter{exa}[section]
 \newenvironment{exa}{\refstepcounter{exa}
  \textbf{Example} \thesection.\arabic{exa}: }{ {\begin{flushright}
                                                                         \rule{0.2em}{0.2em}
                                                                        \end{flushright}}}

\newcounter{problem}[subsection]
 \newenvironment{problem}{\refstepcounter{problem}
  \textbf{Problem} \thesection.\arabic{problem}: }{{\begin{flushright}
                                                                         \rule{0.2em}{0.2em}
                                                                        \end{flushright}}}
\newcommand{\GGi}{\xymatrix{
  \Goid_1  \ar@<-2pt>[r] \ar@<2pt>[r] &  \Goid^0_1    \\
}}
\newcommand{\GGii}{\xymatrix{
  \Goid_2  \ar@<-2pt>[r] \ar@<2pt>[r] &  \Goid^0_2    \\
}}
\newcommand{\GGm}{\xymatrix{
  \Goid  \ar@<-1pt>[r]^{s} \ar@<1pt>[r]_{t} &  \Goid^0    \\
}}
\newcommand{\GGim}{\xymatrix{
  \Goid_1  \ar@<-1pt>[r]^{s_1} \ar@<1pt>[r]_{t_1} &  \Goid^0_1    \\
}}
\newcommand{\GGiim}{\xymatrix{
  \Goid_2  \ar@<-1pt>[r]^{s_2} \ar@<1pt>[r]_{t_2} &  \Goid^0_2    \\
}}

\newcommand{\PGm}{\xymatrix{
  \PD  \ar@<-1pt>[r]^{s} \ar@<1pt>[r]_{t} &  \Sigma    \\
}}

\newcommand{\PGs}{\PD\Sigma\rightrightarrows\Sigma}
\newcommand{\PGoS}{\PD\rightrightarrows\Sigma}
\newcommand{\fPGm}{\xymatrix{
  \PD_\Gamma  \ar@<-1pt>[r]^{s} \ar@<1pt>[r]_{t} &  V_\Gamma    \\
}}
\newcommand{\PGsm}{\xymatrix{
  \PD\Sigma \ar@<-1pt>[r]^{s_{\PD\Sigma}} \ar@<1pt>[r]_{t_{\PD\Sigma}} &  \Sigma   \\
}}
\newcommand{\fPGms}{\xymatrix{
  \PD^s_\Gamma  \ar@<-1pt>[r]^{s} \ar@<1pt>[r]_{t} &  V_\Gamma    \\
}}

\newcommand{\fPG}{\PD_\Gamma\Sigma \rightrightarrows V_\Gamma
}

\newcommand{\fPSGm}{\xymatrix{
  \PD_\Gamma\Sigma  \ar@<-1pt>[r]^{s} \ar@<1pt>[r]_{t} &  V_\Gamma    \\
}}
\newcommand{\fPSG}{\xymatrix{
  \PD_\Gamma\Sigma  \ar@<-2pt>[r] \ar@<2pt>[r] &  V_\Gamma    \\
}}
\newcommand{\fgHGm}{\xymatrix{
  H(\Gamma)  \ar@<-1pt>[r]^/0.3em/{\hat s_H} \ar@<1pt>[r]_/0.3em/{\hat t_H} &  V_\Gamma    \\
}}

\newcommand{\fHGm}{\xymatrix{
  H_\Gamma  \ar@<-1pt>[r]^/0.3em/{\hat s_H} \ar@<1pt>[r]_/0.3em/{\hat t_H} &  V_\Gamma    \\
}}

\newcommand{\fGGm}{\xymatrix{
  \G^G_\Gamma  \ar@<-1pt>[r]^/0.3em/{s_P} \ar@<1pt>[r]_/0.3em/{t_P} &  V_\Gamma    \\
}}
\newcommand{\fGHm}{\xymatrix{
  \G^H_\Gamma  \ar@<-1pt>[r]^/0.3em/{s_P} \ar@<1pt>[r]_/0.3em/{t_P} &  V_\Gamma    \\
}}

\newcounter{count}
\begin{document}
\maketitle
\begin{abstract}\noindent In the project \textit{AQV} \cite{Kaminski1} the issue of quantum constraints, KMS-states and algebras of quantum configuration and momentum variables in Loop Quantum Gravity has been argued. There a physical algebra has been required to contain complete observables and the quantum constraints, or at least the quantum constraints are affilliated with this algebra. In this context a first conjecture for a physical algebra is presented in this article. A new $^*$-algebra for LGQ, which is called the localised holonomy-flux cross-product $^*$-algebra, is studied. A suggestion for a physical $^*$-algebra, which contains the localised holonomy-flux cross-product $^*$-algebra, a modified quantum Hamilton constraint, a localised quantum diffeomorphism constraint and even a modified quantum Master constraint, is given. 
\end{abstract}

\thispagestyle{plain}
\pdfbookmark[0]{\contentsname}{toc}
\tableofcontents

\section{Introduction}

In \cite{Kaminski0} an extended overview about the project \textit{Algebras of Quantum Variables in LQG} has been presented. This article will now focus on the issue of quantum constraints and KMS-theory in LQG. KMS-theory has been studied for von Neumann algebras in \cite[Section 6.5]{KaminskiPHD}. There the \textit{holonomy-flux von Neumann algebra} has been invented. The construction is based on the commutative Weyl $C^*$-algebra for surfaces, which has been introduced in \cite[Sect.: 3.3]{Kaminski1} or \cite[Sect.: 6.3]{KaminskiPHD}, and a certain representation of this $C^*$-algebra on the limit Hilbert space $\HS_\infty$. The \textit{holonomy von Neumann algebra} is given by the anaytic holonomy $C^*$-algebra and a representation on the Hilbert space $\HS_\infty$. But the concept of KMS-states is not naturally available for these von Neumann algebras associated to the Weyl $C^*$-algebra for surfaces or the  analytic holonomy $C^*$-algebra. One can show that even in both cases the von Neumann algebras are not in standard form.
KMS-theory is also available in the context of $C^*$-algebras.
It has been shown in \cite[Sect.: 6.5]{KaminskiPHD} that there are no KMS-states of the Weyl $C^*$-algebra for surfaces or the analytic holonomy $C^*$-algebra with respect to an automorphism group defined by a Lie algebra-valued quantum flux operator. 
Consequently for the algebras derived in \cite[Chapt.: 6,7, Sect.: 8.1-8.3]{KaminskiPHD}, \cite{Kaminski1,Kaminski2,Kaminski3} a KMS-theory is not easy to achieve. Hence new ideas are interesting to study.

Indeed the analysis of algebras in \cite{Kaminski1,Kaminski2,Kaminski3,KaminskiPHD} has shown that new algebras can be defined. In this article such a new $^*$-algebra is developed in section \ref{subsec restrdiffeo}. This algebra is called the \textit{localised holonomy-flux cross-product $^*$-algebra}. This algebra is a certain cross-product algebra and is derived similarly to the holonomy-flux cross-product $^*$-algebra in \cite[Sect.: 8.2]{KaminskiPHD}, \cite{Kaminski3}. 

The new idea is to locate the algebraic objects in particular the algebras. The new construction of the new quantum algebras is also available for non-localised algebraic objects and can be easily extracted from the development. Remark that if matter fields come into the play, the author suggests to locate the algebras of quantum variables and the matter algebras simultanously. The geometric objects like quantum flux operators, which are defined in section \ref{sec fluxdef}, are always localised on surfaces. Moreover, the holonomies are localised along paths or graphs. The bisections are maps from a certain set of vertices in a fixed manifold $\Sigma$ to paths that start at a given vertex in the set of vertices. Consequently, also bisections are somehow localised objects. The definition of these new objects is presented in section \ref{subsubsec bisections}. Furthermore, the important fact is that paths and surfaces intersect each other in vertices. Hence a discretised surface set $\breve S_{\disc}$, which contains only fixed sets of vertices, is proposed in section \ref{subsec fingraphpathgroup}. The localisation of quantum algebras means that algebra elements have distinguishing properties if they depend or do not depend on this discretised surface set. For the implementation of this idea the configuration space is divided in section \ref{subsec graphhol} into two main parts. One of them is constructed from holonomies along paths that start or end at some given surface and is called the \textit{localised part of the configuration space}. The other part is build from holonomies along paths that do not intersect any surface in this surface set. Hence, the first configuration space is localised on surfaces, while the second is not. Furthermore, there are two different $^*$-algebras of quantum holonomy variables presented in section \ref{loc}. One $^*$-algebra is constructed on the localised part of the configuration space and a convolution product between functions depending on this space. In particular, the $^*$-subalgebra of central functions on the localised part of the configuration space is used. The other $^*$-algebra is given by the original analytic holonomy $^*$-algebra, but is restricted to non-localised paths. These $^*$-algebras are completed to different $C^*$-algebras and the $C^*$-tensor product of these two $C^*$-algebras defines the new \textit{localised analytic holonomy $C^*$-algebra}. The $C^*$-algebra of central functions on the localised part of the configuration space is called the \textit{localised part of the localised analytic holonomy $C^*$-algebra}. Similarly to the $^*$-derivations of the analytic holonomy $C^*$-algebra \cite[Sect.: 4]{Kaminski3}, \cite[Sect.: 8.4.2]{KaminskiPHD} in section \ref{der} $^*$-derivations of this new $C^*$-algebra are presented.

There are some new flux operators, which are defined as difference operators between Lie algebra-valued quantum flux operators on different graphs in \ref{sec fluxdef}, and which are called the \textit{localised and discretised flux operators associated to surfaces}. The main difference between the original Lie algebra-valued quantum flux operator, and the localised and discretised Lie algebra-valued flux operator both restricted to a fixed graph is that, the second operator is only non-trivial on paths, which are not contained in a certain subgraph. The \textit{localised enveloping flux algebra associated to a surface set} is derived from the localised and discretised flux operators. Furthermore, there exists an action of this new localised and discretised flux operator on the $C^*$-algebra of central functions on the localised part of the configuration space.

In \cite{Kaminski3,KaminskiPHD} the theory of an abstract cross-product $^*$-algebra has been used to define a new holonomy-flux cross-product $^*$-algebra. This construction is also used for the definition of two new localised algebras. One algebra is based on the $^*$-algebra of central functions on the localised part of the configuration space and the other is derived from the localised analytic holonomy $^*$-algebra. The abstract cross-product $^*$-algebra, which is obtained from the localised enveloping flux algebra associated to a surface set and the $^*$-algebra of central functions on the localised part of the configuration space, is called the \textit{localised part of the localised holonomy-flux cross-product $^*$-algebra}. The \textit{localised holonomy-flux cross-product $^*$-algebra} is given by the abstract cross-product $^*$-algebra, which is obtained by the the localised analytic holonomy $^*$-algebra and the localised enveloping flux algebra associated to a surface set. There are several localised holonomy-flux cross-product $^*$-algebras for different surface sets. A $^*$-representation of this new $^*$-algebra is given similarly to the representation of the holonomy-flux cross-product $^*$-algebra in section\ref{rep}. $C^*$-dynamical systems are studied in section \ref{dyn}. Since such dynamical systems are related to KMS-states, these particular states are studied. There are certain $C^*$-algebras, which admit a KMS-state. Furthermore, there is a KMS-theory available for $O^*$-algebra presented by Inoue \cite{Inoue}, which can be used in a further work to analyse the full or parts of the localised holonomy-flux cross-product $^*$-algebra. Indeed one can show that the localised holonomy-flux cross-product $^*$-algebra is an $O^*$-algebra. 
Furthermore, there is also a construction of a $C^*$-algebra, which will be called the \textit{localised holonomy-flux cross-product $C^*$-algebra} in future work, in analogy to the holonomy-flux cross-product $C^*$-algebra derived in \cite{Kaminski2},\cite[Chap.: 7]{KaminskiPHD}. Then a KMS-theory for this new $C^*$-algebra will be studied.

In this project, a first attempt for a study of the Hamilton constraint as a generator of an automorphism group is studied. This constraint plays a fundamental role in the definition of dynamics in LQG. The object is for example derived from a volume operator in the work of Thiemann \cite{Thiembook07}. In this article a modified volume operator, which is called the \textit{discretised quantum volume operator} $\QD(V)_{\disc}$, is constructed as a sum over Lie algebra-valued quantum flux operators indexed by a triple of paths in a graph that start at a common vertex, which is an intersection of three surfaces. 

In the construction of the Hamilton constraint as a generator the following facts are used, too. Consider the \textit{Lie holonomy algebra}, which is constructed from the localised configuration space restricted to a graph $\Gamma$ and which is identified with the product group $G^{\vert\Gamma\vert}$ of a compact connected Lie group $G$. 
This Lie algebra acts on the $C^*$-algebra of central functions on the localised configuration space restricted to a graph, too. Then the $C^*$-algebra of central functions admits KMS-states with respect to this action. The \textit{modified quantum Hamilton constraint restricted to a graph} is given by
\beqs \exp(H_{\Gamma_i}^+H_{\Gamma_i}):= \left(\ho_A(\alpha)-\ho_A(\alpha)^{-1}\right)\ho_A(\gamma)[\ho_A(\gamma)^{-1},\QD(V)_{\disc}]
\eqs The \textit{modified quantum Hamilton constraint} is defined in the project \textit{AQV} as the limit\\ $H:=\lim_{i\rightarrow\infty}\sum_{\Gamma_i} \exp(H_{\Gamma_i}^+H_{\Gamma_i})$ of a sum over subgraphs of a graph of the modified quantum Hamilton constraint restricted to a graph. Note that, the limit graph is assumed to contain an infinite countable set of subgraphs. 

The next step is to show that, this modified quantum Hamilton constraint is well-defined and is given as the generator of a strongly continuous one-parameter group of automorphisms on the localised part of the localised analytic holonomy $C^*$-algebra. The analysis of parts of the modified quantum Hamilton constraint shows that, the convergence of the limit in the norm-topology is not obvious and is related to the structure of the discretised quantum volume operator $\QD(V)_{\disc}$. Summarising, the norm-convergence of the limit of $H$ is not easy to derive. The author conjectures that this limit converges and does not depend on a particular Hilbert space representation of the modified quantum Hamilton constraint. 

Thiemann has proposed in \cite{ThiemannPhoenix06,Thiembook07} a Master constraint instead of using a set of constraints. Hence in the project \textit{AQV} a modified quantum Master constraint is suggested. The ideas are the following. In \cite{Kaminski0,Kaminski1,Kaminski2} translations defined by bisections of finite path groupoids or finite graph systems play a fundamental role. The most general operators, which depend on bisections of finite graph systems that preserve a discretised surface set $\breve S_{\disc}$ associated to a surface set $\breve S$, are denoted by $D_{\breve S_{\disc},\Gamma}^\sigma$ and are called the \textit{localised quantum diffeomorphism constraints restricted to a graph}. The adjoint operator is denoted by $D_{\breve S_{\disc},\Gamma}^{\sigma,*}$. For example such operators can be defined similarly to elements of the holonomy-flux-graph-diffeomorphism cross-product $C^*$-algebra. The idea for these quantum constraints is to implement the complicated relations between the classical spatial diffeomorphism constraints and the classical Hamilton constraints on the quantum level.

Then the \textit{modified quantum Master constraint} is defined to be sum of the \textit{localised quantum diffeomorphism constraint}, which is given by
\beqs D_{\breve S_{\disc}}:=\limN\sum_{i=1}^{N}\sum_{\sigma_l}D^{\sigma_l,*}_{\breve S_{\disc},\Gamma_i}D^{\sigma_l}_{\breve S_{\disc},\Gamma_i}
\eqs and the modified quantum Hamilton constraint $H$. 

In the article \cite{Kaminski3} the holonomy-flux cross-product $^*$-algebra has been presented, this algebra is comparable with the holonomy-flux $^*$-algebra, which has been developed in \cite{LOST06}. A comparison of the localised holonomy-flux cross-product $^*$-algebra and the holonomy-flux cross-product $^*$-algebra is presented in a table in section \ref{sec comparison}. 
\section{The basic quantum operators}\label{quantum variables}
\subsection{Finite path groupoids and graph systems}\label{subsec fingraphpathgroup}
Let $c:[0,1]\rightarrow\Sigma$ be continuous curve in the domain $\bra 0,1\ket$, which is (piecewise) $C^k$-differentiable ($1\leq k\leq \infty$), analytic ($k=\omega$) or semi-analytic ($k=s\omega$) in $\bra 0,1\ket$ and oriented such that the source vertex is $c(0)=s(c)$ and the target vertex is $c(1)=t(c)$. Moreover assume that, the range of each subinterval of the curve $c$ is a submanifold, which can be embedded in $\Sigma$. An \textbf{edge} is given by a \hyperlink{rep-equiv}{reparametrisation invariant} curve of class (piecewise) $C^k$. The maps $s_{\Sigma},t_{\Sigma}:P\Sigma\rightarrow\Sigma$ where $P\Sigma$ is the path space are surjective maps and are called the source or target map.    

A set of edges $\{e_i\}_{i=1,...,N}$ is called \textbf{independent}, if the only intersections points of the edges are source $s_{\Sigma}(e_i)$ or $t_{\Sigma}(e_i)$ target points. Composed edges are called \textbf{paths}. An \textbf{initial segment} of a path $\gamma$ is a path $\gamma_1$ such that there exists another path $\gamma_2$ and $\gamma=\gamma_1\circ\gamma_2$. The second element $\gamma_2$ is also called a \textbf{final segment} of the path $\gamma$.

\begin{defi}
A \textbf{graph} $\Gamma$ is a union of finitely many independent edges $\{e_i\}_{i=1,...,N}$ for $N\in\N$. The set $\{e_1,...,e_N\}$ is called the \textbf{generating set for $\Gamma$}. The number of edges of a graph is denoted by $\vert \Gamma\vert$. The elements of the set $V_\Gamma:=\{s_{\Sigma}(e_k),t_{\Sigma}(e_k)\}_{k=1,...,N}$ of source and target points are called \textbf{vertices}.
\end{defi}

A graph generates a finite path groupoid in the sense that, the set $\PD_\Gamma\Sigma$ contains all independent edges, their inverses and all possible compositions of edges. All the elements of $\PD_\Gamma\Sigma$ are called paths associated to a graph. Furthermore the surjective source and target maps $s_{\Sigma}$ and $t_{\Sigma}$ are restricted to the maps $s,t:\PD_\Gamma\Sigma\rightarrow V_\Gamma$, which are required to be surjective.

\begin{defi}\label{path groupoid} Let $\Gamma$ be a graph. Then a \textbf{finite path groupoid} $\PD_\Gamma\Sigma$ over $V_\Gamma$ is a pair $(\PD_\Gamma\Sigma, V_\Gamma)$ of finite sets equipped with the following structures: 
\begin{enumerate}
 \item two surjective maps \(s,t:\PD_\Gamma\Sigma\rightarrow V_\Gamma\), which are called the source and target map,
\item the set \(\PD_\Gamma\Sigma^2:=\{ (\gamma_i,\gamma_j)\in\PD_\Gamma\Sigma\times\PD_\Gamma\Sigma: t(\gamma_i)=s(\gamma_j)\}\) of finitely many composable pairs of paths,
\item the  composition \(\circ :\PD_\Gamma^2\Sigma\rightarrow \PD_\Gamma\Sigma,\text{ where }(\gamma_i,\gamma_j)\mapsto \gamma_i\circ \gamma_j\), 
\item the inversion map \(\gamma_i\mapsto \gamma_i^{-1}\) of a path,
\item the object inclusion map \(\iota:V_\Gamma\rightarrow\PD_\Gamma\Sigma\) and
\item $\PD_\Gamma\Sigma$ is defined by the set $\PD_\Gamma\Sigma$ modulo the algebraic equivalence relations generated by
\beq\label{groupoid0} \gamma_i^{-1}\circ \gamma_i\simeq \idf_{s(\gamma_i)}\text{ and }\gamma_i\circ \gamma_i^{-1}\simeq \idf_{t(\gamma_i)}
\eq 
\end{enumerate}
Shortly write $\fPSGm$. 
\end{defi} 
Clearly, a graph $\Gamma$ generates freely the paths in $\PD_\Gamma\Sigma$. Moreover the map $s \times t: \PD_\Gamma\Sigma\rightarrow V_\Gamma\times V_\Gamma$ defined by $(s\times t)(\gamma)=(s(\gamma),t(\gamma))$ for all $\gamma\in\PD_\Gamma\Sigma$ is assumed to be surjective ($\PD_\Gamma\Sigma$ over $V_\Gamma$ is a transitive groupoid), too. 

A general groupoid $\GG$ over $\GG^{0}$ defines a small category where the set of morphisms is denoted in general by $\GG$ and the set of objects is denoted by $\GG^{0}$. Hence in particular the path groupoid can be viewed as a category, since,
\begin{itemize}
\item the set of morphisms is identified with $\PD_\Gamma\Sigma$,
\item the set of objects is given by $V_\Gamma$ (the units) 
\end{itemize}

From the condition (\ref{groupoid0}) it follows that, the path groupoid satisfies additionally 
\begin{enumerate}
 \item $ s(\gamma_i\circ \gamma_j)=s(\gamma_i)\text{ and } t(\gamma_i\circ \gamma_j)=t(\gamma_j)\text{ for every } (\gamma_i,\gamma_j)\in\PD_\Gamma^2\Sigma$
\item $s(v)= v= t(v)\text{ for every } v\in V_\Gamma$
\item\label{groupoid1} $ \gamma \circ\idf_{s(\gamma)} = \gamma = \idf_{t(\gamma)}\circ \gamma\text{ for every } \gamma\in \PD_\Gamma\Sigma\text{ and }$
\item $\gamma \circ (\gamma_i\circ \gamma_j)=(\gamma \circ \gamma_i) \circ \gamma_j$
\item $\gamma \circ (\gamma^{-1}\circ \gamma_1)=\gamma_1= (\gamma_1 \circ \gamma) \circ \gamma^{-1}$
\end{enumerate}

The condition \ref{groupoid1} implies that the vertices are units of the groupoid. 

\begin{defi}
Denote the set of all finitely generated paths by
\beqs \PD_\Gamma\Sigma^{(n)}:=\{(\gamma_1,...,\gamma_n)\in \PD_\Gamma\times ...\PD_\Gamma: (\gamma_i,\gamma_{i+1})\in\PD^{(2)}, 1\leq i\leq n-1 \}\eqs
The set of paths with source point $v\in V_\Gamma$ is given by
\beqs \PD_\Gamma\Sigma^{v}:=s^{-1}(\{v\})\eqs
The set of paths with target  point $v\in V_\Gamma$ is defined by
\beqs \PD_\Gamma\Sigma_{v}:=t^{-1}(\{v\})\eqs
The set of paths with source point $v\in V_\Gamma$ and target point $u\in V_\Gamma$ is 
\beqs \PD_\Gamma\Sigma^{v}_u:=\PD_\Gamma\Sigma^{v}\cap \PD_\Gamma\Sigma_{u}\eqs
\end{defi}

A graph $\Gamma$ is said to be \hypertarget{disconnected}{\textbf{disconnected}} if it contains only mutually pairs $(\gamma_i,\gamma_j)$ of non-composable independent paths $\gamma_i$ and $\gamma_j$ for $i\neq j$ and $i,j=1,...,N$. In other words for all $1\leq i,l\leq N$ it is true that $s(\gamma_i)\neq t(\gamma_l)$ and $t(\gamma_i)\neq s(\gamma_l)$ where $i\neq l$ and $\gamma_i,\gamma_l\in\Gamma$.

\begin{defi}
Let $\Gamma$ be a graph. A \textbf{subgraph $\Gp$ of $\Gamma$} is a given by a finite set of independent paths in $\PD_\Gamma\Sigma$. 
\end{defi}
For example let $\Gamma:=\{\gamma_1,...,\gamma_N\}$ then $\Gp:=\{\gamma_1\circ\gamma_2,\gamma_3^{-1},\gamma_4\}$ where $\gamma_1\circ\gamma_2,\gamma_3^{-1},\gamma_4\in\PD_\Gamma\Sigma$ is a subgraph of $\Gamma$, whereas the set $\{\gamma_1,\gamma_1\circ\gamma_2\}$ is not a subgraph of $\Gamma$. Notice if additionally $(\gamma_2,\gamma_4)\in\PD_\Gamma^{(2)}$ holds, then $\{\gamma_1,\gamma_3^{-1},\gamma_2\circ\gamma_4\}$ is a subgraph of $\Gamma$, too. Moreover for $\Gamma:=\{\gamma\}$ the graph $\Gamma^{-1}:=\{\gamma^{-1}\}$ is a subgraph of $\Gamma$. As well the graph $\Gamma$ is a subgraph of $\Gamma^{-1}$. A subgraph of $\Gamma$ that is generated by compositions of some paths, which are not reversed in their orientation, of the set $\{\gamma_1,...,\gamma_N\}$ is called an \textbf{orientation preserved subgraph of a graph}. For example for $\Gamma:=\{\gamma_1,...,\gamma_N\}$ orientation preserved subgraphs are given by $\{\gamma_1\circ\gamma_2\}$, $\{\gamma_1,\gamma_2,\gamma_N\}$ or $\{\gamma_{N-2}\circ\gamma_{N-1}\}$ if $(\gamma_1,\gamma_2)\in\PD_\Gamma\Sigma^{(2)}$ and $(\gamma_{N-2},\gamma_{N-1})\in\PD_\Gamma\Sigma^{(2)}$.   

\begin{defi}
A \textbf{finite graph system $\PD_\Gamma$ for $\Gamma$} is a finite set of subgraphs of a graph $\Gamma$. A finite graph system $\PD_{\Gp}$ for $\Gp$ is a \hypertarget{finite graph subsystem}{\textbf{finite graph subsystem}} of $\PD_\Gamma$ for $\Gamma$ if the set $\PD_{\Gp}$ is a subset of $\PD_{\Gamma}$ and $\Gp$ is a subgraph of $\Gamma$. Shortly write $\PD_{\Gp}\leq\PD_{\Gamma}$.

A \hypertarget{finite orientation preserved graph system}{\textbf{finite orientation preserved graph system}} $\PD^{\op}_\Gamma$ for $\Gamma$ is a finite set of orientation preserved subgraphs of a graph $\Gamma$. 
\end{defi}

Recall that, a finite path groupoid is constructed from a graph $\Gamma$, but a set of elements of the path groupoid need not be a graph again. For example let $\Gamma:=\{\gamma_1\circ\gamma_2\}$ and $\Gp=\{\gamma_1\circ\gamma_3\}$, then $\Gpp=\Gamma\cup\Gp$ is not a graph, since this set is not independent. Hence only appropriate unions of paths, which are elements of a fixed finite path groupoid, define graphs. The idea is to define a suitable action on elements of the path groupoid, which corresponds to an action of diffeomorphisms on the manifold $\Sigma$. The action has to be transfered to graph systems. But the action of bisection, which is defined by the use of the groupoid multiplication, cannot easily generalised for graph systems. 

\begin{problem}\label{problem group structure on graphs systems}
Let $\breve\Gamma:=\{\Gamma_i\}_{i=1,..,N}$ be a finite set such that each $\Gamma_i$ is a set of not necessarily independent paths such that 
\begin{enumerate}
\item the set contains no loops and
\item each pair of paths satisfies one of the following conditions
\begin{itemize}
\item the paths intersect each other only in one vertex,
\item the paths do not intersect each other or
\item one path of the pair is a segment of the other path.
\end{itemize}
\end{enumerate}

Then there is a map $\circ:\breve\Gamma\times \breve\Gamma\rightarrow\breve\Gamma$ of two elements $\Gamma_1$ and $\Gamma_2$ defined by
\beqs \{\gamma_1,...,\gamma_M\}\circ\{\tg_1,...,\tg_M\}:= &\Big\{ \gamma_i\circ\tg_j:t(\gamma_i)=s(\tg_j)\Big\}_{1\leq i,j\leq M}\\
\eqs for $\Gamma_1:=\{\gamma_1,...,\gamma_M\},\Gamma_2:=\{\tg_1,...,\tg_M\}$. 
Moreover define a map $^{-1}:\breve\Gamma\rightarrow\breve\Gamma$ by
\beqs  \{\gamma_1,...,\gamma_M\}^{-1}:= \{\gamma^{-1}_1,...,\gamma^{-1}_M\}\eqs 

Then the following is derived
\beqs \{\gamma_1,...,\gamma_M\}\circ\{\gamma^{-1}_1,...,\gamma^{-1}_M\}&=\Big\{ \gamma_i\circ\gamma^{-1}_j: t(\gamma_i)=t(\gamma_j)\Big\}_{1\leq i,j\leq M}\\
&=\Big\{ \gamma_i\circ\gamma^{-1}_j:t(\gamma_i)=t(\gamma_j)\text{ and }i\neq j\Big\}_{1\leq i,j\leq M}\\
&\quad\cup\{\idf_{s_{\gamma_j}}\}_{1\leq j\leq M}\\
\neq &\quad\cup\{\idf_{s_{\gamma_j}}\}_{1\leq j\leq M}
\eqs The equality is true, if the set $\breve\Gamma$ contains only graphs such that all paths are mutually non-composable. Consequently this does not define a well-defined multiplication map. Notice that, the same is discovered if a similar map and inversion operation are defined for a finite graph system $\PD_\Gamma$. 
\end{problem}

Consequently the property of paths being independent need not be dropped for the definition of a suitable multiplication and inversion operation. In fact the independence property is a necessary condition for the construction of the holonomy algebra for analytic paths. Only under this circumstance each analytic path is decomposed into a finite product of independent piecewise analytic paths again. 

\begin{defi}
A finite path groupoid $\PD_{\Gp}\Sigma$ over $V_{\Gp}$ is a \textbf{finite path subgroupoid} of $\PD_{\Gamma}\Sigma$ over $V_\Gamma$ if the set $V_{\Gp}$ is contained in $V_\Gamma$ and the set $\PD_{\Gp}\Sigma$ is a subset of $\PD_{\Gamma}\Sigma$. Shortly write $\PD_{\Gp}\Sigma\leq\PD_{\Gamma}\Sigma$.
\end{defi}

Clearly for a subgraph $\Gamma_1$ of a graph $\Gamma_2$, the associated path groupoid $\PD_{\Gamma_1}\Sigma$ over $V_{\Gamma_1}$ is a subgroupoid of $\PD_{\Gamma_2}\Sigma$ over $V_{\Gamma_2}$.  This is a consequence of the fact that, each path in $\PD_{\Gamma_1}\Sigma$ is a composition of paths or their inverses in $\PD_{\Gamma_2}\Sigma$. 

\begin{defi}
A \textbf{family of finite path groupoids} $\{\PD_{\Gamma_i}\Sigma\}_{i=1,...,\infty}$, which is a set of finite path groupoids $\PD_{\Gamma_i}\Sigma$ over $V_{\Gamma_i}$, is said to be \textbf{inductive} if for any $\PD_{\Gamma_1}\Sigma,\PD_{\Gamma_2}\Sigma$ exists a $\PD_{\Gamma_3}\Sigma$ such that $\PD_{\Gamma_1}\Sigma,\PD_{\Gamma_2}\Sigma\leq\PD_{\Gamma_3}\Sigma$.

A \textbf{family of graph systems} $\{\PD_{\Gamma_i}\}_{i=1,...,\infty}$, which is a set of finite path systems $\PD_{\Gamma_i}$ for $\Gamma_i$, is said to be \textbf{inductive} if for any $\PD_{\Gamma_1},\PD_{\Gamma_2}$ exists a $\PD_{\Gamma_3}$ such that $\PD_{\Gamma_1},\PD_{\Gamma_2}\leq \PD_{\Gamma_3}$.
\end{defi}

\begin{defi}
Let $\{\PD_{\Gamma_i}\Sigma\}_{i=1,...,\infty}$ be an inductive family of path groupoids and $\{\PD_{\Gamma_i}\}_{i=1,...,\infty}$ be an inductive family of graph systems.

The \textbf{inductive limit path groupoid $\PD$ over $\Sigma$} of an inductive family of finite path groupoids such that $\PD:=\limi\PD_{\Gamma_i}\Sigma$ is called the \textbf{(algebraic) path groupoid} $\PGoS$.

Moreover there exists an \textbf{inductive limit graph $\Gamma_\infty$} of an inductive family of graphs such that $\Gamma_\infty:=\limi \Gamma_i$.

The \textbf{inductive limit graph system} $\PD_{\Gamma_\infty}$ of an inductive family of graph systems such that $\PD_{\Gamma_\infty}:=\limi \PD_{\Gamma_i}$
\end{defi}

Assume that, the inductive limit $\Gamma_\infty$ of a inductive family of graphs is a graph, which consists of an infinite countable number of independent paths. The inductive limit $\PD_{\Gamma_\infty}$ of a inductive family $\{\PD_{\Gamma_i}\}$ of finite graph systems contains an infinite countable number of subgraphs of $\Gamma_\infty$ and each subgraph is a finite set of arbitrary independent paths in $\Sigma$. 

\subsection{Holonomy maps for finite path groupoids, graph systems and transformations}\label{subsec holmapsfinpath}
In section \ref{subsec fingraphpathgroup} the concept of finite path groupoids for analytic paths has been given. Now the holonomy maps are introduced for finite path groupoids and finite graph systems. The ideas are familar with those presented by Thiemann \cite{Thiembook07}. But for example the finite graph systems have not been studied before. Ashtekar and Lewandowski \cite{AshLew93} have defined the analytic holonomy $C^*$-algebra, which they have based on a finite set of independent hoops. The hoops are generalised for path groupoids and the independence requirement is implemented by the concept of finite graph systems. 

\subsubsection{Holonomy maps for finite path groupoids}\label{subsubsec holmap}

Let $\GGim, \GGiim$ be two arbitrary groupoids.

\begin{defi}
A \hypertarget{groupoid-morphism}{\textbf{groupoid morphism}} between two groupoids $\GG_1$ and $\GG_2$ consists of two maps  $\ho:\GG_1\rightarrow\GG_2$  and $h:\GG_1^0\rightarrow\GG_2^0$ such that
\beqs (\hypertarget{G1}{G1})\qquad \ho(\gamma\circ\gp)&= \ho(\gamma)\ho(\gp)\text{ for all }(\gamma,\gp)\in \GG_1^{(2)}\eqs
\beqs (\hypertarget{G2}{G2})\qquad s_{2}(\ho(\gamma))&=h(s_{1}(\gamma)),\quad t_2(\ho(\gamma))=h(t_{1}(\gamma))\eqs 
 
A \textbf{strong groupoid morphism} between two groupoids $\GG_1$ and $\GG_2$ additionally satisfies
\beqs (\hypertarget{SG2}{SG})\qquad \text{ for every pair }(\ho(\gamma),\ho(\gp))\in\GG_2^{(2)}\text{ it follows that }(\gamma,\gp)\in \GG_1^{(2)}\eqs
\end{defi}

Let $G$ be a Lie group. Then $G$ over $e_G$ is a groupoid, where the group multiplication $\cdot: G^2\rightarrow G$ is defined for all elements  $g_1,g_2,g\in G$ such that $g_1\cdot g_2 = g$. A groupoid morphism between a finite path groupoid $\PD_\Gamma\Sigma$ to $G$ is given by the maps
\[\ho_\Gamma: \PD_\Gamma\Sigma\rightarrow G,\quad h_\Gamma:V_\Gamma\rightarrow e_G \] Clearly
\beq \ho_\Gamma(\gamma\circ\gp)&= \ho_\Gamma(\gamma)\ho_\Gamma(\gp)\text{ for all }(\gamma,\gp)\in \PD_\Gamma\Sigma^{(2)}\\
s_G(\ho_\Gamma(\gamma))&=h_\Gamma(s_{\PD_\Gamma\Sigma}(\gamma)),\quad t_G(\ho_\Gamma(\gamma))=h_\Gamma(t_{\PD_\Gamma\Sigma}(\gamma))
\eq But for an arbitrary pair $(\ho_\Gamma(\gamma_1),\ho_\Gamma(\gamma_2))=:(g_1,g_2)\in G^{(2)}$ it does not follows that, $(\gamma_1,\gamma_2)\in \PD_\Gamma\Sigma^{(2)}$ is true. Hence $\ho_\Gamma$ is not a strong groupoid morphism.

\begin{defi}\label{def sameholanal}Let $\fPG$ be a finite path groupoid.

Two paths $\gamma$ and $\gp$ in $\PD_\Gamma\Sigma$ have the \textbf{same-holonomy for all connections} iff 
\beqs \ho_\Gamma(\gamma)=\ho_\Gamma(\gp)\text{ for all }&(\ho_\Gamma,h_\Gamma)\text{ groupoid morphisms }\\ & \ho_\Gamma:\PD_\Gamma\Sigma\rightarrow G, h:V_\Gamma\rightarrow\{e_G\}
\eqs Denote the relation by $\sim_{\text{s.hol.}}$.
\end{defi}
\begin{lem}
The same-holonomy for all connections relation is an equivalence relation. 
\end{lem}
Notice that, the quotient of the finite path groupoid and the same-holonomy relation for all connections replace the hoop group, which has been used in \cite{AshLew93}.
\begin{defi}\label{genrestgroupoidforgraph}
Let $\fPG$ be a finite path groupoid modulo same-holonomy for all connections equivalence.

A \hypertarget{holonomy map for a finite path groupoid}{\textbf{holonomy map for a finite path groupoid}} $\PD_\Gamma\Sigma$ over $V_\Gamma$ is a groupoid morphism consisting of the maps $(\ho_\Gamma,h_\Gamma)$, where
\(\ho_\Gamma:\PD_\Gamma\Sigma\rightarrow G,h_\Gamma:V_\Gamma\rightarrow \{e_G\}\). 
The set of all holonomy maps is abbreviated by $\Hom(\PD_\Gamma\Sigma,G)$.
\end{defi}

For a short notation observe the following.
In further sections it is always assumed that, the finite path groupoid $\fPG$ is considered modulo same-holonomy for all connections equivalence although it is not stated explicitly.

\subsubsection{Holonomy maps for finite graph systems}\label{subsec graphhol}

Ashtekar and Lewandowski \cite{AshLew93} have presented the loop decomposition into a finite set of independent hoops (in the analytic category). This structure is replaced by a graph, since a graph is a set of independent edges. Notice that, the set of hoops that is generated by a finite set of independent hoops, is generalised to the set of finite graph systems. A finite path groupoid is generated by the set of edges, which defines a graph $\Gamma$, but a set of elements of the path groupoid need not be a graph again. The appropriate notion for graphs constructed from sets of paths is the finite graph system, which is defined in section \ref{subsec fingraphpathgroup}. Now the concept of holonomy maps is generalised for finite graph systems. Since the set, which is generated by a finite number of independent edges, contains paths that are composable, there are two possibilities to identify the image of the holonomy map for a finite graph system on a fixed graph with a subgroup of $G^{\vert\Gamma\vert}$. One way is to use the generating set of independend edges of a graph, which has been also used in \cite{AshLew93}. On the other hand, it is also possible to identify each graph with a disconnected subgraph of a fixed graph, which is generated by a set of independent edges. Notice that, the author implements two situations. One case is given by a set of paths that can be composed further and the other case is related to paths that are not composable. This is necessary for the definition of an action of the flux operators. Precisely the identification of the image of the holonomy maps along these paths is necessary to define a well-defined action of a flux element on the configuration space \cite{Kaminski1}. 

First of all consider a graph $\Gamma$ that is generated by the set $\{\gamma_1,...,\gamma_N\}$ of edges. Then each subgraph of a graph $\Gamma$ contain paths that are composition of edges in $\{\gamma_1,...,\gamma_N\}$ or inverse edges. For example the following set $\Gp:=\{\gamma_1\circ\gamma_2\circ\gamma_3,\gamma_4\}$ defines a subgraph of $\Gamma:=\{\gamma_1,\gamma_2,\gamma_3,\gamma_4\}$. Hence there is a natural identification available.

\begin{defi}
A subgraph $\Gp$ of a graph $\Gamma$ is always generated by a subset $\{\gamma_1,...,\gamma_M\}$ of the generating set $\{\gamma_1,...,\gamma_N\}$ of independent edges that generates the graph $\Gamma$. Hence each subgraph is identified with a subset of $\{\gamma_1^{\pm 1},...,\gamma_N^{\pm 1}\}$. This is called the \hypertarget{natural identification}{\textbf{natural identification of subgraphs}}.
\end{defi}

\begin{exa}\label{exa natidentif}
For example consider a subgraph $\Gp:=\{\gamma_1\circ\gamma_2,\gamma_3\circ\gamma_4,...,\gamma_{M-1}\circ\gamma_M\}$, which is identified naturally with a set $\{\gamma_1,...,\gamma_M\}$. The set $\{\gamma_1,...,\gamma_M\}$ is a subset of $\{\gamma_1,...,\gamma_N\}$ where $N=\vert \Gamma\vert$ and $M\leq N$. 

Another example is given by the graph $\Gpp:=\{\gamma_1,\gamma_2\}$ such that $\gamma_2=\gpe\circ\gpz$, then $\Gpp$ is identified naturally with $\{\gamma_1,\gpe,\gpz\}$. This set is a subset of $\{\gamma_1,\gpe,\gpz,\gamma_3,...,\gamma_{N-1}\}$. 
\end{exa}

\begin{defi}
Let $\Gamma$ be a graph, $\PD_\Gamma$ be the finite graph system. Let $\Gp:=\{\gamma_1,...,\gamma_M\}$be a subgraph of $\Gamma$.

A \hypertarget{holonomy map for a finite graph system}{\textbf{holonomy map for a finite graph system}} $\PD_\Gamma$ is a given by a pair of maps $(\ho_\Gamma,h_\Gamma)$ such that there exists a holonomy map\footnote{In the work the holonomy map for a finite graph system and the holonomy map for a finite path groupoid is denoted by the same pair $(\ho_\Gamma,h_\Gamma)$.} $(\ho_\Gamma,h_\Gamma)$ for the finite path groupoid $\fPG$ and
\beqs &\ho_\Gamma:\PD_\Gamma\rightarrow G^{\vert \Gamma\vert},\quad \ho_\Gamma(\{\gamma_1,...,\gamma_M\})=(\ho_\Gamma(\gamma_1),...,\ho_\Gamma(\gamma_M), e_G,...,e_G)\\
&h_\Gamma:V_\Gamma\rightarrow \{e_G\}
\eqs 
The set of all holonomy maps for the finite graph system is denoted by $\Hom(\PD_\Gamma,G^{\vert \Gamma\vert})$.

The image of a map $\ho_\Gamma$ on each subgraph $\Gp$ of the graph $\Gamma$ is given by
\beqs (\ho_\Gamma(\gamma_1),...,\ho_\Gamma(\gamma_M),e_G,...,e_G)
\eqs is an element of $G^{\vert \Gamma\vert}$. The set of all images of maps on subgraphs of $\Gamma$ is denoted by $\Ab_\Gamma$.
\end{defi}
The idea is now to study two different restrictions of the set $\PD_\Gamma$ of subgraphs. For a short notation of a ''set of  holonomy maps for a certain restricted set of subgraphs of a graph'' in this article the following notions are introduced.
\begin{defi}
If the subset of all disconnected subgraphs of the finite graph system $\PD_\Gamma$ is considered, then the restriction of $\Ab_\Gamma$, which is identified with $G^{\vert \Gamma\vert}$ appropriately, is called the \hypertarget{non-standard identification}{\textbf{non-standard identification of the configuration space}}. If the subset of all natural identified subgraphs of the finite graph system $\PD_\Gamma$ is considered, then the restriction of $\Ab_\Gamma$, which is identified with $G^{\vert \Gamma\vert}$ appropriately, is called the \hypertarget{natural identification}{\textbf{natural identification of the configuration space}}.
\end{defi}

A comment on the non-standard identification of $\Ab_\Gamma$ is the following. If $\Gp:=\{\gamma_1\circ\gamma_2\}$ and $\Gpp:=\{\gamma_2\}$ are two subgraphs of $\Gamma:=\{\gamma_1,\gamma_2,\gamma_3\}$. The graph $\Gp$ is a subgraph of $\Gamma$. Then evaluation of a map $\ho_\Gamma$ on a subgraph $\Gp$ is given by
\beqs \ho_\Gamma(\Gp)=(\ho_\Gamma(\gamma_1\circ\gamma_2),\ho_\Gamma(s(\gamma_2)),\ho_\Gamma(s(\gamma_3)))=(\ho_\Gamma(\gamma_1)\ho_\Gamma(\gamma_2),e_G,e_G)\in G^3
\eqs and the holonomy map of the subgraph $\Gpp$ of $\Gp$ is evaluated by
\beqs \ho_\Gamma(\Gpp)=(\ho_\Gamma(s(\gamma_1)),\ho_\Gamma(s(\gamma_2))\ho_\Gamma(\gamma_2),\ho_\Gamma(s(\gamma_3)))=(\ho_\Gamma(\gamma_2),e_G,e_G)\in G^3
\eqs

\begin{exa}
Recall example \thesubsection.\ref{exa natidentif}.
For example for a subgraph $\Gp:=\{\gamma_1\circ\gamma_2,\gamma_3\circ\gamma_4,...,\gamma_{M-1}\circ\gamma_M\}$, which is naturally identified with a set $\{\gamma_1,...,\gamma_M\}$. Then the holonomy map is evaluated at $\Gp$ such that \[\ho_\Gamma(\Gp)=(\ho_\Gamma(\gamma_1),\ho_\Gamma(\gamma_2),....,\ho_\Gamma(\gamma_M),e_G,...,e_G)\in G^N\] where $N=\vert \Gamma\vert$. For example, let $\Gp:=\{\gamma_1,\gamma_2\}$ such that $\gamma_2=\gpe\circ\gpz$ and which is naturally identified with $\{\gamma_1,\gpe,\gpz\}$. Hence \[\ho_\Gamma(\Gp)=(\ho_\Gamma(\gamma_1),\ho_\Gamma(\gpe),\ho_\Gamma(\gpz),e_G,...,e_G)\in G^N\] is true.

Another example is given by the disconnected graph $\Gp:=\{\gamma_1\circ\gamma_2\circ\gamma_3,\gamma_4\}$, which is a subgraph of $\Gamma:=\{\gamma_1,\gamma_2,\gamma_3,\gamma_4\}$. Then the non-standard identification is given by
\[\ho_\Gamma(\Gp)=(\ho_\Gamma(\gamma_1\circ\gamma_2\circ\gamma_3),\ho_\Gamma(\gamma_4),e_G,e_G)\in G^4\]

If the natural identification is used, then $\ho_\Gamma(\Gp)$ is idenified with 
\[(\ho_\Gamma(\gamma_1),\ho_\Gamma(\gamma_2),\ho_\Gamma(\gamma_3),\ho_\Gamma(\gamma_4))\in G^4\]

Consider the following example. Let $\Gppp:=\{\gamma_1,\alpha,\gamma_2,\gamma_3\}$ be a graph such that 
 \begin{center}
\includegraphics[width=0.2\textwidth]{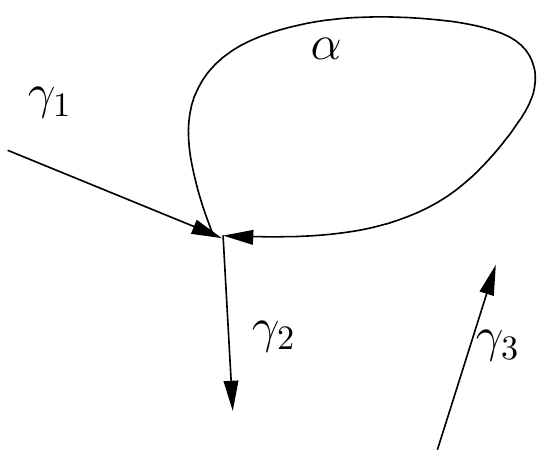}
\end{center}
Then notice the sets $\Gamma_1:=\{\gamma_1\circ\alpha,\gamma_3\}$ and $\Gamma_2:=\{\gamma_1\circ\alpha^{-1},\gamma_3\}$. In the non-standard identification of the configuration space $\Ab_{\Gppp}$ it is true that,
\beqs \ho_{\Gppp}(\Gamma_1)=(\ho_{\Gppp}(\gamma_1\circ\alpha),\ho_{\Gppp}(\gamma_3),e_G,e_G)\in G^4,\\
\ho_{\Gppp}(\Gamma_2)=(\ho_{\Gppp}(\gamma_1\circ\alpha^{-1}),\ho_{\Gppp}(\gamma_3),e_G,e_G)\in G^4
\eqs holds. Whereas in the natural identification of $\Ab_{\Gppp}$
 \beqs \ho_{\Gppp}(\Gamma_1)=(\ho_{\Gppp}(\gamma_1),\ho_{\Gppp}(\alpha),\ho_{\Gppp}(\gamma_3),e_G)\in G^4,\\
\ho_{\Gppp}(\Gamma_2)=(\ho_{\Gppp}(\gamma_1),\ho_{\Gppp}(\alpha^{-1}),\ho_{\Gppp}(\gamma_3),e_G)\in G^4
\eqs yields.
\end{exa}

The equivalence class of similar or equivalent groupoid morphisms defined in definition \ref{def similargroupoidhom} allows to define the following object.
The set of images of all holonomy maps of a finite graph system modulo the similar or equivalent groupoid morphisms equivalence relation is denoted by $\Ab_\Gamma/\bar\SimGroup_\Gamma$. 

\subsubsection{Transformations in finite path groupoids and finite graph systems}\label{subsubsec bisections}

The aim of this section is to clearify the graph changing operators in LQG framework and the role of finite diffeomorphisms in $\Sigma$. 
Therefore operations, which add, delete or transform paths, are introduced.  In particular translations in a finite path graph groupoid and in the groupoid $G$ over $\{e_G\}$ are studied. 

\paragraph*{Transformations in finite path groupoid\\[5pt]}

\begin{defi}
Let $\varphi$ be a $C^k$-diffeomorphism on $\Sigma$, which maps surfaces into surfaces. 

Then let $(\Phi_\Gamma,\varphi_\Gamma)$ be a pair of bijective maps, where $\varphi\vert_{V_\Gamma}=\varphi_\Gamma$ and 
\beq\Phi_\Gamma:\PD_\Gamma\Sigma\rightarrow\PD_\Gamma\Sigma\text{ and }\varphi_\Gamma:V_\Gamma\rightarrow V_\Gamma\eq 
such that 
\beq (s\circ\Phi_\Gamma)(\gamma)=(\varphi_\Gamma\circ s)(\gamma),\quad (t\circ \Phi_\Gamma)(\gamma)=(\varphi_\Gamma\circ t)(\gamma)\text{ for all }\gamma\in\PD_\Gamma\Sigma\eq holds such that $(\Phi_\Gamma,\varphi_\Gamma)$ defines a groupoid morphism.

Call the pair $(\Phi_\Gamma,\varphi_\Gamma)$ a \textbf{path-diffeomorphism of a finite path groupoid} $\PD_\Gamma\Sigma$ over $V_\Gamma$. Denote the set of finite path-diffeomorphisms by $\Diff(\PD_\Gamma\Sigma)$.
\end{defi}

Notice that, for $(\gamma,\gp)\in\PD_\Gamma\Sigma^{(2)}$ it is true that
\beq\label{eq requcombi0} \Phi_\Gamma(\gamma\circ\gp)=\Phi_\Gamma(\gamma)\circ\Phi_\Gamma(\gp)
\eq requires that
\beq\label{eq requcombi} (t\circ\Phi_\Gamma)(\gamma)=(s\circ\Phi_\Gamma)(\gp)
\eq Hence from \eqref{eq requcombi0} and \eqref{eq requcombi} it follows that, $\Phi_\Gamma(\idf_v)=\idf_{\varphi_\Gamma(v)}$ is true.

A path-diffeomorphism $(\Phi_\Gamma,\varphi_\Gamma)$ is lifted to $\Hom(\PD_\Gamma\Sigma,G)$. \\
The pair $(\ho_\Gamma\circ\Phi_\Gamma,h_\Gamma\circ\varphi_\Gamma)$ defined by
\beqs \ho_\Gamma\circ\Phi_\Gamma&: \PD_\Gamma\Sigma\rightarrow G,\quad \gamma\mapsto (\ho_\Gamma\circ\Phi_\Gamma)(\gamma)\\
h_\Gamma\circ\varphi_\Gamma&: V_\Gamma\rightarrow \{e_G\},\quad (h_\Gamma\circ\varphi_\Gamma)(v)=e_G
\eqs such that
\beqs &s_{\Hol}((\ho_\Gamma\circ\Phi_\Gamma)(\gamma))=(h_\Gamma\circ\varphi_\Gamma)(s(\gamma))=e_G,\\
&t_{\Hol}(\ho_\Gamma\circ\Phi_\Gamma(\gamma))=(h_\Gamma\circ\varphi_\Gamma)(t(\gamma))=e_G\text{ for all }\gamma\in\PD_\Gamma\Sigma
\eqs whenever $(\ho_\Gamma,h_\Gamma)\in\Hom(\PD_\Gamma\Sigma,G)$ and $(\Phi_\Gamma,\varphi_\Gamma)$ is a path-diffeomorphism, is a \hyperlink{holonomy map for a finite path groupoid}{holonomy map for a finite path groupoid} $\PD_\Gamma\Sigma$ over $V_\Gamma$.

\begin{defi}
A \textbf{left-translation in the finite path groupoid} $\PD_\Gamma\Sigma$ over $V_\Gamma$ at a vertex $v$ is a map defined by
\beqs L_\theta:\PD_\Gamma\Sigma^v\rightarrow \PD_\Gamma\Sigma^{w},\quad \gamma\mapsto L_{\theta}(\gamma):=\theta\circ\gamma
\eqs
for some $\theta\in\PD_\Gamma\Sigma_{v}^{w}$ and all $\gamma\in\PD_\Gamma\Sigma^v$.
\end{defi} 
In analogy a right-translation $R_\theta$ and an inner-translation $I_{\theta,\theta^\prime}$ in the finite path groupoid $\PD_\Gamma\Sigma$ over $V_\Gamma$ at a vertex $v$ can be defined.
\begin{rem}
Let $(\Phi_\Gamma,\varphi_\Gamma)$ be a path-diffeomorphism on a finite path groupoid $\PD_\Gamma\Sigma$ over $V_\Gamma$. Then a left-translation in the finite path groupoid $\PD_\Gamma\Sigma$ over $V_\Gamma$ at a vertex $v$ is defined by a path-diffeomorphism $(\Phi_\Gamma,\varphi_\Gamma)$ and the following object
\beq L_{\Phi_\Gamma}:\PD_\Gamma\Sigma^v\rightarrow \PD_\Gamma\Sigma^{\varphi_\Gamma(v)},\quad \gamma\mapsto L_{\Phi_\Gamma}(\gamma):=\Phi_\Gamma(\gamma)\text{ for }\gamma\in\PD_\Gamma\Sigma^v
\eq
Furthermore a right-translation in the finite path groupoid $\PD_\Gamma\Sigma$ over $V_\Gamma$ at a vertex $v$ is defined by a path-diffeomorphism $(\Phi_\Gamma,\varphi_\Gamma)$ and the following object
\beq R_{\Phi_\Gamma}:\PD_\Gamma\Sigma_v\rightarrow \PD_\Gamma\Sigma_{\varphi_\Gamma(v)},\quad \gamma\mapsto R_{\Phi_\Gamma}(\gamma):=\Phi_\Gamma(\gamma)\text{ for }\gamma\in\PD_\Gamma\Sigma_v
\eq

Finally an inner-translation in the finite path groupoid $\PD_\Gamma\Sigma$ over $V_\Gamma$ at the vertices $v$ and $w$ is defined by
\beqs  I_{\Phi_\Gamma}:\PD_\Gamma\Sigma^v_w\rightarrow \PD_\Gamma\Sigma^{\varphi_\Gamma(v)}_{\varphi_\Gamma(w)},\quad \gamma\mapsto I_{\Phi_\Gamma}(\gamma)=\Phi_\Gamma(\gamma)\text{ for }\gamma\in\PD_\Gamma\Sigma^v_w
\eqs where $(s\circ\Phi_\Gamma)(\gamma)=\varphi_\Gamma(v)$ and $(t\circ\Phi_\Gamma)(\gamma)=\varphi_\Gamma(w)$.
\end{rem}

In the following considerations the right-translation in a finite path groupoid is focused, but there is a generalisation to left-translations and inner-translations.
\begin{defi}
A \hypertarget{bisection of a finite path groupoid}{\textbf{bisection of a finite path groupoid}} $\PD_\Gamma\Sigma$ over $V_\Gamma$ is a map $\sigma:V_\Gamma\rightarrow\PD_\Gamma\Sigma$, which is right-inverse to the map $s:\PD_\Gamma\Sigma\rightarrow V_\Gamma$ (i.o.w. $s\circ\sigma=\id_{V_\Gamma}$) and such that $t\circ\sigma:V_\Gamma\rightarrow V_{\Gamma}$ is a bijective map\footnote{Note that in the infinite case of path groupoids an additional condition for the map $t\circ\sigma:\Sigma\rightarrow\Sigma$ has to be required. The map has to be a diffeomorphism. Observe that, the map $t\circ\sigma$ defines the finite diffeomorphism $\varphi_\Gamma:V_\Gamma\rightarrow V_\Gamma$.}. The set of bisections on $\PD_\Gamma\Sigma$ over $V_\Gamma$ is denoted $\mathfrak{B}(\PD_\Gamma\Sigma)$.
\end{defi}
\begin{rem}\label{rem defiofright}
Discover that, a bisection $\sigma\in\mathfrak{B}(\PD_\Gamma\Sigma)$ defines a path-diffeomorphism $(\varphi_\Gamma,\Phi_\Gamma)\in \Diff(\PD_\Gamma\Sigma)$, where $\varphi_\Gamma=t\circ\sigma$ and $\Phi_\Gamma$ is given by the right-translation $R_{\sigma(v)}:\PD_\Gamma\Sigma_v\rightarrow\PD_\Gamma\Sigma_{\varphi_\Gamma(v)}$ in $\fPG$, where $R_{\sigma(v)}(\gamma)=\Phi_\Gamma(\gamma)$ for all $\gamma\in\PD_\Gamma\Sigma_v$ and for a fixed $v\in V_\Gamma$. The right-translation is defined by 
\beq\label{eq Rendv}R_{\sigma(v)}(\gamma):=
\left\{\begin{array}{ll}
 \gamma\circ\sigma(v) & v=t(\gamma)\\
\gamma\circ\idf_{t(\gamma)} & v\neq t(\gamma)\\
\end{array}\right.\\
\eq whenever $t(\gamma)$ is the target vertex of a non-trivial path $\gamma$ in $\Gamma$. For a trivial path $\idf_v$ the right-translation is defined by $R_{\sigma(v)}(\idf_v)=\idf_{(t\circ\sigma)(v)}$ and $R_{\sigma(v)}(\idf_w)=\idf_{w}$ whenever $v\neq w$. The right-translation $R_{\sigma(v)}$ is required to be bijective. Before this result is proven in lemma \ref{lem path-diffeom} notice the following considerations. 
\end{rem}

Note that, $(R_{\sigma(v)},t\circ \sigma)$ transfers to the holonomy map such that
\beq\label{eq righttransl} (\ho_\Gamma\circ R_{\sigma(t(\gp))}(\gamma\circ\gp)&=\ho_\Gamma(\gamma\circ\gp\circ\sigma(t(\gp)))\\
&=\ho_\Gamma(\gamma)\ho_\Gamma(\gp\circ\sigma(t(\gp)))
\eq is true.
There is a bijective map between a right-translation $R_{\sigma(v)}:\PD_{\Gamma}\Sigma_v\rightarrow\PD_\Gamma\Sigma_{(t\circ\sigma)(v)}$ and a path-diffeomorphism $(\varphi_\Gamma,\Phi_\Gamma)$. In particular observe that, $\sigma\in\mathfrak{B}(\PD_\Gamma\Sigma_v)$ and $(\varphi_\Gamma,\Phi_\Gamma)\in\Diff(\PD_\Gamma\Sigma_v)$. Simply speaking the path-diffeomorphism does not change the source and target vertex at the same time. The path-diffomorphism changes the target vertex by a (finite) diffeomorphism and, therefore, the path is transformed. 

Bisections $\sigma$ in a finite path groupoid can be transfered, likewise path-diffeomorphisms, to holonomy maps. The pair $(\ho_\Gamma\circ \Phi_\Gamma, h_\Gamma\circ\varphi_\Gamma)$ of the maps
defines a pair of maps $(\ho_\Gamma\circ \Phi_\Gamma,h_\Gamma\circ\varphi_\Gamma)$ by 
\beq \ho_\Gamma\circ \Phi_\Gamma:\PD_\Gamma\Sigma_v\rightarrow G\text{ and } h_\Gamma\circ\varphi_\Gamma: V_\Gamma\rightarrow\{e_G\}
\eq which is a \hyperlink{holonomy map for a finite path groupoid}{holonomy map for a finite path groupoid} $\PD_\Gamma\Sigma$ over $V_\Gamma$.

\begin{lem}\label{lem groupbisection}The set $\mathfrak{B}(\PD_\Gamma\Sigma)$ of bisections on the finite path groupoid $\PD_\Gamma\Sigma$ over $V_\Gamma$ forms a group.
\end{lem}
\begin{proofs}The group multiplication is given by
\beqs (\sigma\ast\sigma^\prime)(v) =\sigma^\prime(v)\circ\sigma(t(\sigma^\prime(v)))\text{ for }v\in V_\Gamma
\eqs whenever $\sigma^\prime(v)\in\PD_\Gamma\Sigma^v_{\varphi_\Gamma^\prime(v)}$ and $\sigma(t(\sigma^\prime(v)))\in\PD_\Gamma\Sigma^{(t\circ\sigma^\prime)(v)}_{\varphi_\Gamma(v)}$.

Clearly the group multiplication is associative.
The unit $\id$ is equivalent to the object inclusion $v\mapsto\idf_v$ of the groupoid $\fPG$, where $\idf_v$ is the constant loop at $v$, and the inversion is given by
\beqs \sigma^{-1}(v)=\sigma((t\circ\sigma)^{-1}(v))^{-1}\text{ for }v\in V_\Gamma
\eqs 
\end{proofs}

The group property of bisections $\mathfrak{B}(\PD_\Gamma\Sigma)$ carries over to holonomy maps. Using the group multiplication $\cdot $ of $G$ conclude that
\beqs (\ho_\Gamma\circ R_{(\sigma\ast\sigma^\prime)(v)})(\idf_{v})=\ho_\Gamma\circ (R_{\sigma^\prime(v)}\circ R_{\sigma(t(\sigma^\prime(v)))})(\idf_{v})
= \ho_\Gamma(\sigma^\prime(v))\cdot\ho_\Gamma(\sigma(t(\sigma^\prime(v))))\text{ for }v\in V_\Gamma
\eqs is true. 
\begin{rem}
Moreover right-translations define path-diffeomorphisms, i.e. $R_{(\sigma)(v)}=\Phi_\Gamma$ and $\varphi_\Gamma=t\circ\sigma$ whenever $v\in V_\Gamma$.
But for two bisections $\sigma_\Gp,\breve\sigma_\Gp\in\mathfrak{B}(\PD_\Gamma\Sigma)$ the object $\sigma_\Gp(v)\circ\breve\sigma_\Gp(v)$ is not comparable with $(\sigma_\Gp\ast\breve\sigma_\Gp)(v)$. Then for the composition $\Phi_1(\gamma)\circ\Phi_2(\gamma)$, there exists no path-diffeomorphism $\Phi$ such that $\Phi_1(\gamma)\circ\Phi_2(\gamma)=\Phi(\gamma)$ yields in general. Moreover generally the object $\Phi_1(\gamma)\circ\Phi_2(\gp)=\Phi(\gamma\circ\gp)$ is not well-defined.

But the following is defined
\beq\label{def compositionofdiffeo} R_{(\sigma\ast\sigma^\prime)(v)}(\gamma)=\Phi_\Gamma^\prime(\gamma)\circ\Phi_\Gamma(\idf_{\varphi_\Gamma^\prime(v)})=:(\Phi_\Gamma^\prime\ast\Phi_\Gamma)(\gamma)
\eq whenever $\gamma\in\PD_\Gamma\Sigma_v$, $(\varphi_\Gamma,\Phi_\Gamma)\in\Diff(\PD_\Gamma\Sigma_v)$ and $(\varphi_\Gamma^\prime,\Phi_\Gamma^\prime)\in\Diff(\PD_\Gamma\Sigma_{\varphi_\Gamma^\prime(v)})$ are path-diffeomorphisms  such that $\varphi_\Gamma=t\circ\sigma$, $\Phi_\Gamma=R_{\sigma(\varphi_\Gamma^\prime(v))}$ and  $\varphi_\Gamma^\prime=t\circ\sigma^\prime$, $\Phi_\Gamma^\prime=R_{\sigma^\prime(v)}$.

Moreover for $(\gamma,\gp)\in\PD_\Gamma\Sigma^{(2)}$ and $\gp\in\PD_\Gamma\Sigma_v$ it is true that
\beqs (\Phi_\Gamma^\prime\ast\Phi_\Gamma)(\gamma\circ\gp)=\Phi_\Gamma^\prime(\gamma\circ\gp)\circ\Phi_\Gamma(\idf_{\varphi_\Gamma^\prime(v)}) = \Phi_\Gamma^\prime(\gamma)\circ\Phi_\Gamma^\prime(\gp)\circ\Phi_\Gamma(\idf_{\varphi_\Gamma^\prime(v)}) = \Phi_\Gamma^\prime(\gamma)\circ(\Phi_\Gamma^\prime\ast\Phi_\Gamma)(\gp)
\eqs holds.
\end{rem}

Then the following lemma easily follows.

\begin{lem}\label{lem path-diffeom}Let $\sigma$ be a bisection contained in $\mathfrak{B}(\PD_\Gamma\Sigma)$ and $v\in V_\Gamma$.

The pair $(R_{\sigma(v)},t\circ\sigma)$ of maps such that
\beqs &R_{\sigma(v)}:\PD_\Gamma\Sigma_v\rightarrow\PD_\Gamma\Sigma_{(t\circ \sigma)(v)},\quad &
s\circ R_{\sigma(v)} = (t\circ\sigma)\circ s\\
&t\circ\sigma:V_\Gamma\rightarrow V_\Gamma,\quad &t\circ R_{\sigma(v)}= (t\circ\sigma)\circ t
\eqs defined in remark \ref{rem defiofright} is a path-diffeomorphism in $\fPG$.
\end{lem}
\begin{proofs}This follows easily from the derivation
\beqs R_{\sigma(t(\gp))}(\gamma\circ\gp)&=\gamma\circ\gp\circ\sigma(t(\gp))
= R_{\sigma(t(\gp))}(\gamma)\circ R_{\sigma(t(\gp))}(\gp)
\eqs 
\beqs R_{\sigma(t(\gamma))}(\idf_{s(\gamma)}\circ\gamma)&=R_{\sigma(t(\gamma))}(\idf_{s(\gamma)})\circ R_{\sigma(t(\gamma))}(\gamma)=\idf_{s(\gamma)}\circ\gamma\circ\sigma(t(\gamma))\\
R_{\sigma(t(\gamma))}(\gamma\circ\idf_{t(\gamma)})&=R_{\sigma(t(\gamma))}(\gamma)\circ R_{\sigma(t(\gamma))}(\idf_{t(\gamma)})=\gamma\circ\sigma(t(\gamma))\circ\idf_{(t\circ\sigma)(t(\gamma))}
\eqs The inverse map satisfies 
\beqs R^{-1}_{\sigma(v)}(\gamma\circ\sigma(v))=R_{\sigma^{-1}(v)}(\gamma\circ\sigma(v))
=\gamma\circ\sigma(v)\circ\sigma^{-1}(v)=\gamma
\eqs whenever $v=t(\gamma)$, 
\beqs R^{-1}_{\sigma(v)}(\gamma)=\gamma
\eqs whenever $v\neq t(\gamma)$ and
\beqs R^{-1}_{\sigma(v)}(\idf_{(t\circ\sigma)(v)})=\idf_{v}
\eqs

Moreover derive
\beqs ( s\circ R_{\sigma(v)})(\gp)= ((t\circ\sigma)\circ s)(\gp)
\eqs for all $\gp\in\PD_\Gamma\Sigma_v$ and a fixed bisection $\sigma\in \mathfrak{B}(\PD_\Gamma\Sigma)$.
 \end{proofs}

Notice that, $L_{\sigma(v)}$ and $I_{\sigma(v)}$ similarly to the pair $(R_{\sigma(v)},t\circ\sigma)$ can be defined. Summarising the pairs $(R_{\sigma(v)},t\circ\sigma)$, $(L_{\sigma(v)},t\circ\sigma)$ and $(I_{\sigma(v)},t\circ\sigma)$ for a bisection $\sigma \in\mathfrak{B}(\PD_\Gamma\Sigma)$ are path-diffeomorphisms of a finite path groupoid $\fPG$.

In general a right-translation $(R_\sigma,t\circ\sigma)$ in the finite path groupoid $\PD_\Gamma\Sigma$ over $\Sigma$ for a bisection $\sigma\in\mathfrak{B}(\PD_\Gamma\Sigma)$ is defined by the bijective maps $R_\sigma$ and $t\circ\sigma$, which are given by 
\beqs 
&R_{\sigma}:\PD_\Gamma\Sigma\rightarrow\PD_\Gamma\Sigma,\quad &
s\circ R_{\sigma} =  s\qquad\quad\text{ }\\
&t\circ\sigma:V_\Gamma\rightarrow V_\Gamma,\quad &t\circ R_{\sigma}= (t\circ\sigma)\circ t\\
&R_\sigma(\gamma):=\gamma\circ\sigma(t(\gamma))\quad\forall \gamma\in \PD_\Gamma\Sigma;\quad R^{-1}_\sigma:=R_{\sigma^{-1}}
\eqs For example for a fixed suitable bisection $\sigma$ the right-translation is $R_\sigma(\idf_v)=\gamma$, then $R^{-1}_\sigma(\gamma)=\gamma\circ\gamma^{-1}=\idf_v$ for $v=s(\gamma)$. Clearly the right-translation $(R_\sigma,t\circ\sigma)$ is not a groupoid morphism in general. 

\begin{defi}
Define for a given bisection $\sigma$ in $\mathfrak{B}(\PD_\Gamma\Sigma)$, the \textbf{right-translation in the groupoid $G$ over $\{e_G\}$} through
\beqs &\ho_\Gamma\circ R_{\sigma}:\PD_\Gamma\Sigma\rightarrow G, \quad \gamma\mapsto (\ho_\Gamma\circ R_\sigma)(\gamma):=\ho_\Gamma(\gamma\circ\sigma(t(\gamma)))= \ho_\Gamma(\gamma)\cdot\ho_\Gamma(\sigma(t(\gamma))) \\
&h_\Gamma\circ t\circ\sigma:V_\Gamma\rightarrow e_G 
\eqs 

Furthermore for a fixed $\sigma\in\mathfrak{B}(\PD_\Gamma\Sigma)$ define
the \textbf{left-translation in the groupoid $G$ over $\{e_G\}$} by
\beqs &\ho_\Gamma\circ L_\sigma:\PD_\Gamma\Sigma\rightarrow G,\quad \gamma\mapsto \ho_\Gamma(\sigma((t\circ\sigma)^{-1}(s(\gamma)))\circ\gamma)=\ho_\Gamma(\sigma((t\circ\sigma)^{-1}(s(\gamma))))\cdot\ho_\Gamma(\gamma)\\
&h_\Gamma\circ t\circ\sigma:V_\Gamma\rightarrow e_G 
\eqs
and the \textbf{inner-translation in the groupoid $G$ over $\{e_G\}$}
\beqs &\ho_\Gamma\circ I_\sigma:\PD_\Gamma\Sigma\rightarrow G,\quad \gamma\mapsto \ho_\Gamma(\sigma((t\circ\sigma)^{-1}(s(\gamma)))\circ\gamma\circ\sigma(t(\gamma)))=\ho_\Gamma(\sigma((t\circ\sigma)^{-1}(s(\gamma))))\cdot\ho_\Gamma(\gamma)\cdot\ho_\Gamma(\sigma(t(\gamma)))\\
&h_\Gamma\circ t\circ\sigma:V_\Gamma\rightarrow e_G 
\eqs such that $I_\sigma=L_{\sigma^{-1}}\circ R_\sigma$.
\end{defi}

The pairs $(R_\sigma,t\circ\sigma)$ and $(L_\sigma,t\circ\sigma)$ are not groupoid morphisms. Whereas the pair $(I_\sigma,t\circ\sigma)$ is a groupoid morphism, since for all pairs $(\gamma,\gp)\in\PD_\Gamma\Sigma^{(2)}$ such that $t(\gamma)=s(\gp)$ it is true that $\sigma(t(\gamma)) \circ\sigma((t\circ\sigma)^{-1}(t(\gamma)))^{-1}=\idf_{t(\gamma)}$ holds. Notice that, in this situation $\sigma(t(\gamma))=\sigma(t(\gamma\circ\gp))$ is satisfied.

\begin{prop}
The map $\sigma\mapsto R_\sigma$ is a group isomorphism, i.e. $R_{\sigma\ast\sigma^\prime}=R_\sigma\circ R_{\sigma^\prime}$ and where $\sigma\mapsto t\circ\sigma$ is a group isomorphism from $\mathfrak{B}(\PD_\Gamma\Sigma)$ to the group of finite diffeomorphisms $\Diff(V_\Gamma)$ in a finite subset $V_\Gamma$ of $\Sigma$. 

The maps $\sigma\mapsto L_\sigma$ and $\sigma\mapsto I_\sigma$ are group isomorphisms.
\end{prop}

There is a generalisation of path-diffeomorphisms in the finite path groupoid, which coincide with the graphomorphism presented by Fleischhack in \cite{Fleischhack06}. In this approach the diffeomorphism $\varphi:\Sigma\rightarrow\Sigma$ changes the source and target vertex of a path $\gamma$. Consequently the path-diffeomorphism $(\Phi,\varphi)$, which implements the inner-translation $I_{\Phi}$ in the path groupoid $\PGs$, is a graphomorphism in the context of Fleischhack. Some element of the set of graphomorphisms is directly related to a right-translation $R_\sigma$ in the path groupoid.  Precisely for every $v\in\Sigma$ and $\sigma\in\mathfrak{B}(\PD\Sigma)$ the pairs $(R_{\sigma(v)},t\circ\sigma)$, $(L_{\sigma(v)},t\circ\sigma)$ and $(I_{\sigma(v)},t\circ\sigma)$ define graphomorphism. Furthermore the right-translation $R_{\sigma(v)}$,  the left-translation $L_{\sigma(v)}$ and the inner-translation $I_{\sigma(v)}$ are required to be bijective maps, and hence the maps cannot map non-trivial paths to trivial paths. This property restricts the set of all graphomorphism, which is generated by these translations. In particular in this article graph changing operations, which change the number of edges of a graph, are studied. Hence the left- or right-translation in a finite path groupoid is used in the further development. Notice that in general, these objects do not define graphomorphism.
Finally notice that, in particular for the graphomorphism $(R_{\sigma(v)},t\circ\sigma)$ and a holonomy map for the path groupoid $\PGs$ a similar relation \eqref{eq righttransl} holds. The last equation is fundamental for the construction of $C^*$-dynamical systems, which contain the analytic holonomy $C^*$-algebra restricted to a finite path groupoid $\fPG$ and a point norm continuous action of the finite path-diffeomorphism group $\Diff(V_\Gamma)$ on this algebra. Clearly the right-, left- and inner-translations $R_\sigma$, $L_\sigma$ and $I_\sigma$ are constructed such that \eqref{eq righttransl} generalises. But note that, in the infinite case considered by Fleischhack the action of the bisections $\mathfrak{B}(\PD\Sigma)$ are not point-norm continuous implemented. The advantage of the usage of bisections is that, the map $\sigma\mapsto t\circ\sigma$ is a group morphism between the group $\mathfrak{B}(\PD\Sigma)$ of bisections in $\PGs$ and the group $\Diff(\Sigma)$ of diffeomorphisms in $\Sigma$. Consequently there is an action of the group of diffeomorphisms in $\Sigma$ on the finite path groupoid, which is used to define an action of the group of diffemorphisms in $\Sigma$ on the analytic holonomy $C^*$-algebra. 

\paragraph*{Transformations in finite graph systems\\[5pt]}
To proceed it is necessary to transfer the notion of bisections and right-translations to finite graph systems. A right-translation $R_{\sigma_\Gamma}$  is a mapping that maps graphs to graphs. Each graph is a finite union of independent edges. This causes problems. Since the definition of right-translation in a finite graph system $\PD_\Gamma$ is often not well-defined for all bisections in the finite graph system and all graphs. For example if the graph $\Gamma:=\{\gamma_1,\gamma_2\}$ is disconnected and the bisection $\tilde\sigma$ in the finite path groupoid $\PD_{\Gamma}\Sigma$ over $V_\Gamma$ is defined by $\tilde\sigma(s(\gamma_1))=\gamma_1$, $\tilde\sigma(s(\gamma_2))=\gamma_2$, $\tilde\sigma(t(\gamma_1))=\gamma_1^{-1}$ and $\tilde\sigma(t(\gamma_2))=\gamma_2^{-1}$ where $V_\Gamma:=\{s(\gamma_1),t(\gamma_1),s(\gamma_2),t(\gamma_2)\}$. Let $\idf_\Gamma$ be the set given by the elements $\idf_{s(\gamma_1)}$,$\idf_{s(\gamma_2)}$,$\idf_{t(\gamma_1)}$ and $\idf_{t(\gamma_2)}$. Then notice that, a bisection $\sigma_\Gamma$, which maps a set of vertices in $V_\Gamma$ to a set of paths in $\PD_\Gamma\Sigma$, is given for example by $\sigma_\Gamma(V_\Gamma)= \{\gamma_1,\gamma_2,\gamma_1^{-1},\gamma_2^{-1}\}$. In this case the right-translation $R_{\sigma_\Gamma(V_\Gamma)}(\idf_\Gamma)$ is equivalent to $\{\gamma_1,\gamma_2,\gamma_1^{-1},\gamma_2^{-1}\}$, which is not a set of independent edges and hence not a graph. Loosely speaking the graph-diffeomorphism acts on all vertices in the set $V_\Gamma$ and hence implements four new edges. But a bisection $\sigma_\Gamma$, which maps a subset $V:= \{s(\gamma_1),s(\gamma_2)\}$ of $V_\Gamma$ to a set of paths,  leads to a translation $R_{\sigma_\Gamma(V)}(\{\idf_{s(\gamma_1)},\idf_{s(\gamma_2)}\})=\{\gamma_1,\gamma_2\}$, which is indeed a graph. Set $\Gp:=\{\gamma_1\}$ and $V^\prime=\{s(\gamma_1)\}$. Then observe that, for a restricted bisection, which maps a set $V^\prime$ of vertices in $V_\Gamma$ to a set of paths in $\PD_\Gp\Sigma$, the right-translation become $R_{\sigma_\Gp(V^\prime)}(\{\idf_{s(\gamma_1)}\})=\{\gamma_1\}$, which defines a graph, too. Notice that $\idf_{s(\gamma_1)}$ is a subgraph of $\Gp$. Hence in the simpliest case new edges are emerging. The next definition of the right-tranlation shows that composed paths arise, too.

\begin{defi}\label{defi bisecongraphgroupioid}
Let $\Gamma$ be a graph, $\fPG$ be a finite path groupoid and let $\PD_\Gamma$ be a finite graph system. Moreover the set $V_\Gamma$ is given by $\{v_1,...,v_{2N}\}$.

A \hypertarget{bisection of a finite graph system}{\textbf{bisection of a finite graph system}} $\PD_\Gamma$ is a map $\sigma_\Gamma:V_\Gamma\rightarrow\PD_\Gamma$ such that there exists a bisection $\tilde \sigma\in\mathfrak{B}(\PD_\Gamma\Sigma)$ such that $\sigma_\Gamma(V)=\{\tilde\sigma(v_i):v_i\in V\}$ whenever $V$ is a subset of $V_\Gamma$.

Define a restriction $\sigma_\Gp:V_\Gp\rightarrow\PD_\Gp$ of a bisection $\sigma_\Gamma$ in $\PD_\Gamma$ by
\beqs \sigma_\Gp(V):=\{ \tilde\sigma(w_k) : & w_k\in V\}
\eqs for each subgraph $\Gp$ of $\Gamma$ and $V\subseteq V_\Gp$. 

A \textbf{right-translation in the finite graph system} $\PD_\Gamma$ is a map $R_{\sigma_\Gp}: \PD_\Gp\rightarrow \PD_\Gp$, which is given by a  bisection $\sigma_\Gp:V_\Gp\rightarrow \PD_\Gp$ such that
\beqs
&R_{\sigma_\Gp}(\Gpp)= R_{\sigma_\Gp}(\{\gppe,...,\gppm,\idf_{w_i}:w_i\in\{s(\gpe),...,s(\gpk)\in V^s_{\Gp}:s(\gpi)\neq s(\gpj) \forall i\neq j\}\setminus V_{\Gpp}\})\\[4pt]
&:=\left\{
\begin{array}{ll}
\gppe,...,\gppj,\gppje\circ\tilde\sigma(t(\gppje)),...,\gppm\circ\tilde\sigma(t(\gppm)),\idf_{w_i}\circ\tilde\sigma(w_i) : & \\[4pt]
w_i\in \{s(\gpe),...,s(\gpk)\in V^s_{\Gp}:s(\gpi)\neq s(\gpj) \forall i\neq j\}\setminus V_{\Gpp}, \quad t(\gppi)\neq t(\gppl)\quad\forall i\neq l; i,l\in\bra j+1,M\ket &\\
\end{array}\right\}\\
&=\Gamma^{\prime\prime}_\sigma\\
\eqs where $\tilde\sigma\in\mathfrak{B}(\PD_{\Gamma}\Sigma)$, $K:=\vert\Gp\vert$ and $M:=\vert\Gpp\vert$, $V^s_{\Gp}$ is the set of all source vertices of $\Gp$ and such that $\Gpp:=\{\gppe,...,\gppm\}$ is a subgraph of $\Gp:=\{\gpe,...,\gpk\}$ and $\Gamma^{\prime\prime}_\sigma$ is a subgraph of $\Gp$.
\end{defi} 

Derive that, for $\tilde\sigma(t(\gamma_i))=\gamma_i^{-1}$ it is true that
$(t\circ\tilde\sigma)(s(\gamma_i^{-1}))=s(\gamma_i)=(t\circ\tilde\sigma)(t(\gamma_i))$ holds.

\begin{exa}
Let $\Gamma$ be a disconnected graph.
Then for a bisection $\tilde\sigma\in\mathfrak{B}(\PD_\Gamma\Sigma)$ such that $\sigma(t(\gamma_i))=\gamma_i^{-1}$ for all $1\leq i\leq \vert\Gamma\vert$ it is true that 
\beqs R_{\sigma_\Gamma}(\Gamma)&=\Big\{\gamma_1\circ\tilde\sigma(t(\gamma_1)),...,\gamma_N\circ\tilde\sigma(t(\gamma_N)),\idf_{s(\gamma_1)}\circ\tilde\sigma(s(\gamma_1)),...,\idf_{s(\gamma_N)}\circ\tilde\sigma(s(\gamma_N))\Big\}\\
&=\{\idf_{s(\gamma_1)},...,\idf_{s(\gamma_N)}\}
\eqs yields. Set $\Gp:=\{\gpe,...,\gpm\}$, then derive
\beqs R_{\sigma_\Gamma}(\Gp)&=\Big\{\gpe\circ\tilde\sigma(t(\gpe)),...,\gpm\circ\tilde\sigma(t(\gpm)),\idf_{s(\gamma_1)}\circ\tilde\sigma(s(\gamma_1)),...,\idf_{s(\gamma_{N-M})}\circ\tilde\sigma(s(\gamma_{N-M}))\Big\}\\
&=\{\idf_{s(\gpe)},...,\idf_{s(\gpm)},\gamma_{1},...,\gamma_{N-M}\}
\eqs if $\Gamma=\Gp\cup\{\gamma_1,...,\gamma_{N-M}\}$.
\end{exa}
To understand the definition of the right-translation notice the following problem.

\begin{problem}\label{prob withoutcond} Consider a subgraph $\Gamma$ of $\tilde\Gamma:=\{\gamma_1,\gamma_2,\gamma_3,\gamma_4\}$, a map $\tilde\sigma:V_{\tilde\Gamma}\rightarrow \PD_{\tilde\Gamma}\Sigma$. Then the map 
\beqs R_{\sigma_{\tilde\Gamma}}(\Gamma)=\{\gamma_1\circ\gamma_1^{-1},\gamma_2\circ\idf_{t(\gamma_2)},\gamma_3\circ\idf_{t(\gamma_3)},\idf_{s(\gamma_1)}\circ\gamma_4 \}=:\Gamma_\sigma
\eqs 
 \begin{center}
\includegraphics[width=0.5\textwidth]{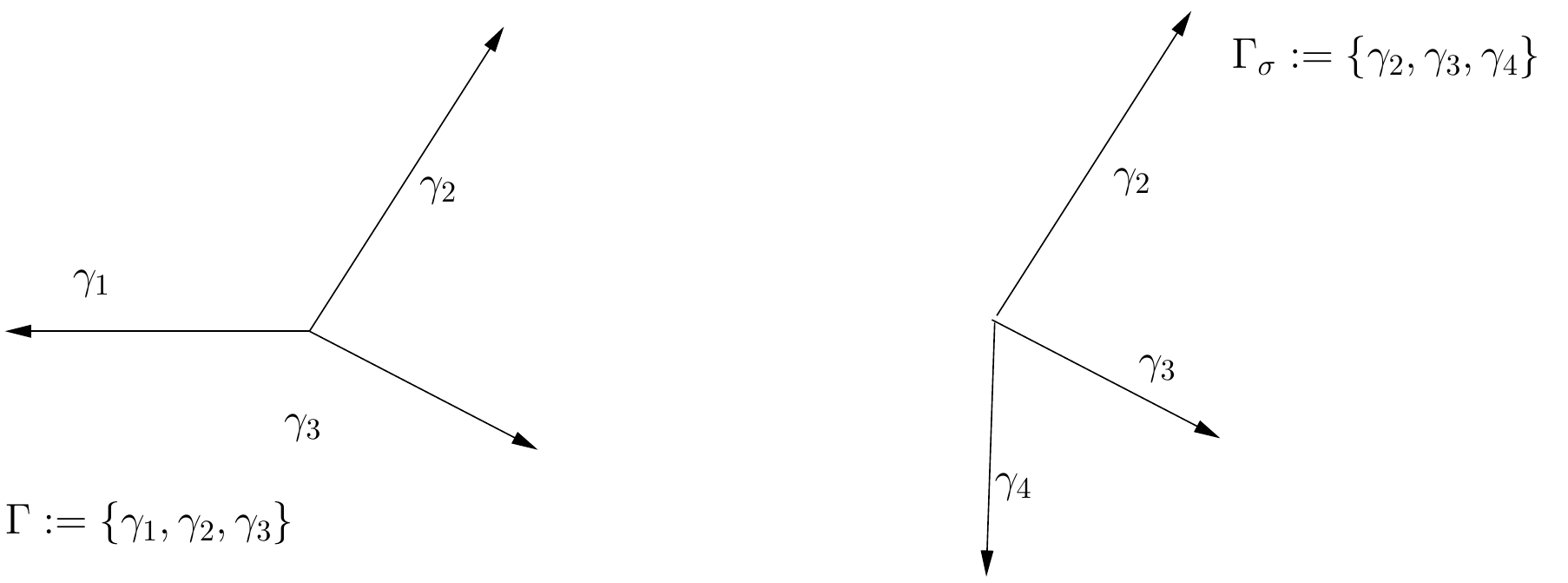}
\end{center} is not a right-translation. This follows from the following fact. Notice that,  the map $\sigma$ maps $t(\gamma_1)\mapsto s(\gamma_1)$, $t(\gamma_2)\mapsto t(\gamma_2)$, $t(\gamma_3)\mapsto t(\gamma_3)$ and $s(\gamma_1)\mapsto t(\gamma_4)$. Then the map $\tilde\sigma$ is not a bisection in the finite path groupoid $\PD_{\tilde\Gamma}\Sigma$ over $V_{\tilde\Gamma}$ and does not define a right-translation $R_{\sigma_{\tilde\Gamma}}$ in the finite graph system $\PD_{\tilde\Gamma}$.

This is a general problem. For every bisection $\tilde\sigma$ in a finite path groupoid such that a graph $\Gamma:=\{\gamma\}$ is translated to $\{\gamma\circ\tilde\sigma(t(\gamma),\tilde\sigma(s(\gamma))\}$. Hence either such translations in the graph system are excluded or the definition of the bisections has to be restricted to maps such that the map $t\circ\tilde\sigma$ is not bijective.  Clearly, the restriction of the right-translation such that $\Gamma$ is mapped to $\{\gamma\circ\tilde\sigma(t(\gamma),\idf_{s(\gamma)}\}$ implies that a simple path orientation transformation is not implemented by a right-translation.

Furthermore there is an ambiguity for graph containing to paths $\gamma_1$ and $\gamma_2$ such that $t(\gamma_1)=t(\gamma_2)$. Since in this case a bisection $\sigma$, which maps $t(\gamma_1)$ to $t(\gamma_3)$, the right-translation is $\{\gamma_1\circ\gamma_3,\gamma_2\circ\gamma_3\}$, is not a graph anymore. 
\end{problem}

\begin{exa}
Otherwise there is for example a subgraph $\Gp$ of $\tilde\Gamma:=\{\gamma_1,\gamma_2,\gamma_3,\gamma_4\}$ and a bisection $\tilde\sigma_{\tilde\Gamma}$ such that 
\beqs \Gamma^\prime_\sigma:=\{\gamma_1\circ\gamma_1^{-1},\gamma_2\circ\idf_{s(\gamma_2)},\gamma_3\circ\idf_{s(\gamma_3)}
\}
\eqs
 \begin{center}
\includegraphics[width=0.5\textwidth]{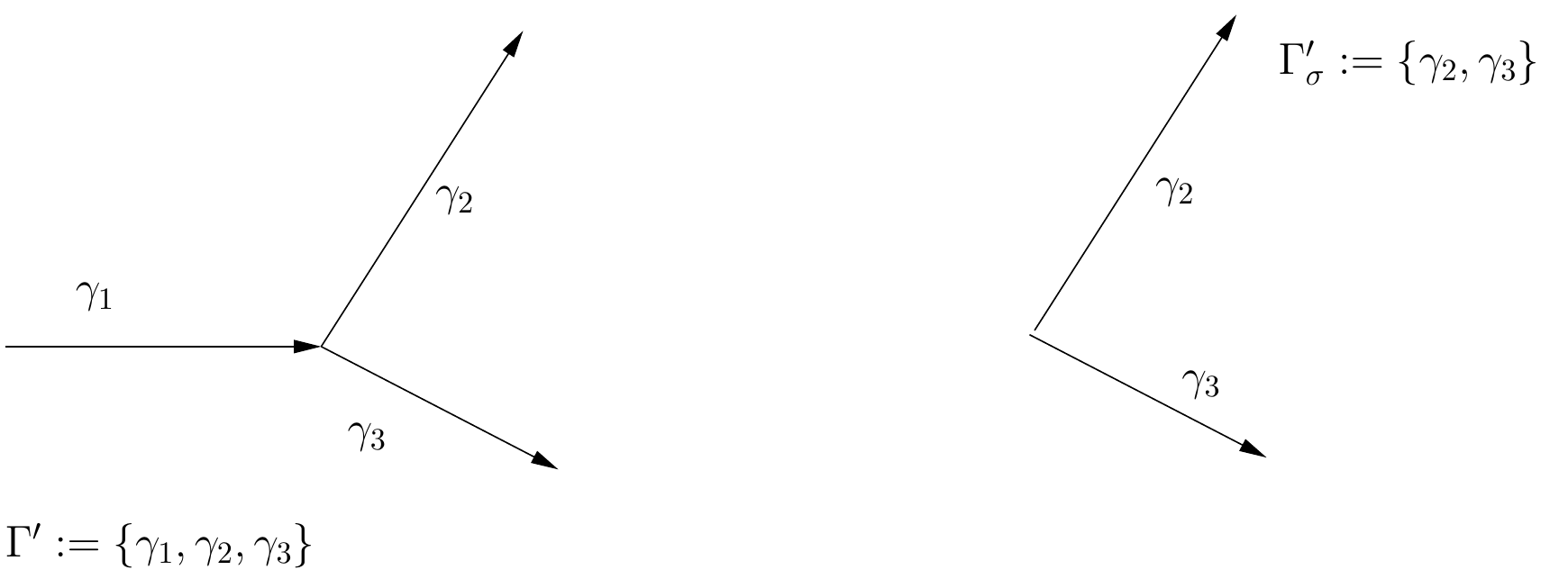}
\end{center} Notice that, $t(\gamma_1)\mapsto s(\gamma_1)$, $t(\gamma_2)\mapsto t(\gamma_2)$, $t(\gamma_3)\mapsto t(\gamma_3)$ and $t(\gamma_4)\mapsto t(\gamma_4)$. Hence the the map $\tilde\sigma_{\tilde\Gamma}:V_{\tilde\Gamma}\rightarrow\PD_{\tilde\Gamma}\Sigma$ is bijective map and consequently a bisection. The bisection $\sigma_{\tilde\Gamma}$ in the graph system $\PD_{\tilde\Gamma}$ defines a right-translation $R_{\sigma_{\tilde\Gamma}}$ in $\PD_{\tilde\Gamma}$.

Moreover for a subgraph $\Gpp:=\{\gamma_2,\gamma_3\}$ of the graph $\breve\Gamma:=\{\gamma_1,\gamma_2,\gamma_3\}$ there exists a map $\sigma_{\breve \Gamma}:V_{\breve\Gamma}\rightarrow\PD_{\breve\Gamma}$ such that
\beqs R_{\sigma_{\breve\Gamma}}(\Gpp)=\{\gamma_2,\gamma_3,\tilde\sigma(s(\gamma_1))\}=\{\gamma_2,\gamma_3,\gamma_1\}
\eqs
\begin{center}
\includegraphics[width=0.5\textwidth]{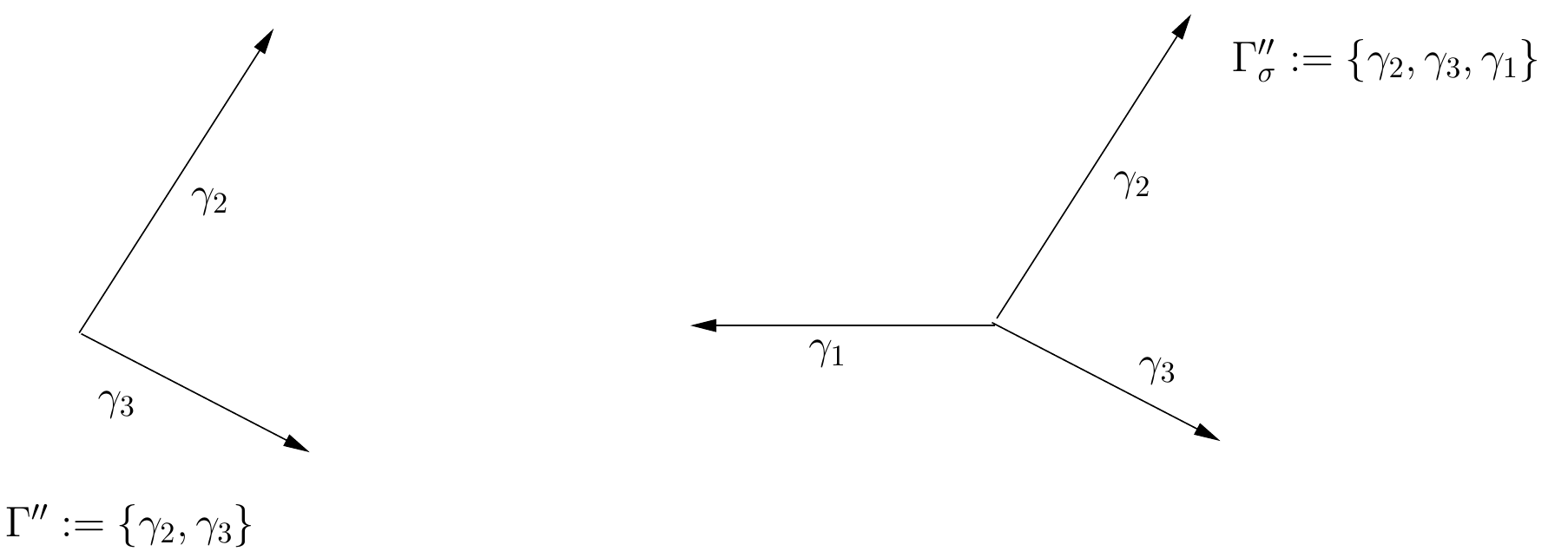}
\end{center}
where $t(\gamma_2)\mapsto t(\gamma_2)$, $t(\gamma_3)\mapsto t(\gamma_3)$ and $s(\gamma_1)\mapsto t(\gamma_1)$. Consequently in this example the map $\tilde\sigma_{\breve \Gamma}$ is a bisection, which defines a right-translation in $\PD_{\breve\Gamma}$.

Note that, for a graph $\Gamma$ such that $\tilde\Gamma$ and $\breve\Gamma$ are subgraphs the bisection $\sigma_{\tilde\Gamma}$ extends to a bisection $\sigma$ in $\PD_\Gamma$ and $\sigma_{\breve\Gamma}$ extends to a bisection $\breve\sigma$ in $\PD_\Gamma$.
\end{exa}

Moreover the bisections of a finite graph system are transfered, analogously, to bisections of a finite path groupoid $\fPG$ to the group $G^{\vert\Gamma\vert}$. Let $\sigma\in\mathfrak{B}(\PD_\Gamma)$ and $(\ho_\Gamma,h_\Gamma)\in\Hom(\PD_\Gamma,G^ {\vert\Gamma\vert})$. Thus there are two maps
\beq \ho_\Gamma\circ R_\sigma:\PD_\Gamma\rightarrow G^{\vert\Gamma\vert}\text{ and }h_\Gamma\circ(t\circ\sigma):V_\Gamma\rightarrow \{e_G\}
\eq which defines a \hyperlink{holonomy map for a finite graph system}{holonomy map for a finite graph system} if $\sigma$ is suitable.

Now, a similar right-translation in a finite graph system in comparison to the right-translation $R_{\sigma(v)}$ in a finite path groupoid is studied.
Let $\sigma_\Gp:V_\Gp\rightarrow\PD_\Gp$ be a restriction of $\sigma_\Gamma\in\mathfrak{B}(\PD_\Gamma)$.  Moreover let $V$ be a subset of $V_ \Gp$, let $\Gpp$ be a subgraph of $\Gp$ and $\Gppp$ be a subgraph of $\Gpp$. Then a right-translation is given by
\beqs &R_{\sigma_\Gp(V)}(\Gpp)\\
&:= 
\left\{\begin{array}{ll}
 R_{\sigma_\Gp}(\{\gppe,...,\gppm,\idf_{w_i}:w_i\in\{s(\gpe),...,s(\gpk)\in V_{\Gp}:s(\gpi)\neq s(\gpj) \forall i\neq j\}\setminus V_{\Gpp}\}& : V_{\Gpp}\subset V\\
R_{\sigma_\Gp}(\{\gppe,...,\gpppn,\idf_{w_i}:w_i\in\{s(\gpe),...,s(\gpk)\in V_{\Gpp}:s(\gpi)\neq s(\gpj) \forall i\neq j\}\setminus V_{\Gppp}\})& \\
\quad \cup\{\idf_{x_i}:x_i\in V\setminus  V_{\Gpp}\}\cup\{\Gpp\setminus\Gppp\}
&: V_{\Gpp}\not\subset V,V_{\Gppp}\subset V\\
       \end{array}\right.
\eqs Loosely speaking, the action of a path-diffeomorphism is somehow localised on a fixed vertex set $V$. 

For example note that for a subgraph $\Gp:=\{\gamma\circ\gp\}$ of $\Gamma:=\{\gamma,\gp\}$ and a subset $V:=\{t(\gp)\}$ of $V_\Gamma$, it is true that
\beqs (\ho_\Gamma\circ R_{\sigma_{\Gamma}(V)})(\gamma\circ\gp)= (\ho_\Gamma\circ R_{\sigma_{\Gamma}(V)})(\gamma)\cdot (\ho_\Gamma\circ R_{\sigma_{\Gamma}(V)})(\gp)= \ho_\Gamma(\gamma)\cdot (\ho_\Gamma\circ R_{\sigma_{\Gamma}(V)})(\gp)
=\ho_\Gamma(\gamma\circ\gp\circ\sigma(t( \gp)))
\eqs yields whenever $\sigma_{\Gamma}\in\mathfrak{B}(\PD_\Gamma\Sigma)$. For a special bisection $\breve\sigma_\Gamma$ it is true that,
\beqs (\ho_\Gamma\circ R_{\breve\sigma_{\Gamma}})(\gamma)= \ho_\Gamma(\gamma\circ\gp)= (\ho_\Gamma\circ R_{\breve\sigma_{\Gamma}})(\gamma)\cdot (\ho_\Gamma\circ R_{\breve\sigma_{\Gamma}})(\gp)
\eqs holds whenever $\breve\sigma_\Gamma\in\mathfrak{B}(\PD_\Gamma\Sigma)$, $\breve\sigma_\Gamma(t(\gp))=\idf_{t(\gp)}$ and $\breve\sigma_\Gamma(t(\gamma))=\gp$. Let $\tilde\sigma$ be the bisection in the finite path groupoid $\PD_\Gamma\Sigma$ that defines the bisection $\breve\sigma$ in $\PD_\Gamma$. Then the last statement is true, since $R_{\breve\sigma_{\Gamma}}(\gp)=\gp\circ\gp^{-1}$ requires $\tilde\sigma_\Gamma:t(\gp)\mapsto s(\gp)$ and $R_{\breve\sigma_{\Gamma}}(\gamma)=\gamma\circ\gp$ needs $\tilde\sigma_\Gamma: t(\gamma)\mapsto t(\gp)$, where $s(\gp)=t(\gamma)$. Then
$R_{\breve\sigma_\Gamma}(\gamma)$ and $R_{\sigma_\Gamma(t(\gp))}(\gamma)$ coincide if $\breve\sigma_\Gamma(t(\gamma))=\sigma_\Gamma(t(\gamma))$ and $\breve\sigma_\Gamma(t(\gp))=\idf_{t(\gp)}$ holds. 
 
\begin{problem}\label{prob righttranslgrouoid}Let $\Gp$ be a subgraph of the graph $\Gamma$, $\sigma_\Gamma$ be a bisection in $\PD_\Gamma$, $\sigma_\Gp:V_\Gp\rightarrow\PD_\Gp$ be a restriction of $\sigma_\Gamma\in\mathfrak{B}(\PD_\Gamma)$.  Moreover let $V$ be a subset of $V_ \Gp$, let $\Gpp:=\{\gamma\circ\gp\}$ be a subgraph of $\Gp$. Let $(\gamma,\gp)\in\PD_\Gp\Sigma^{(2)}$.

Then even for a suitable bisection $\sigma_\Gp$ in $\PD_\Gamma$ it follows that,
\beq\label{eq ineqsolv} R_{\sigma_{\Gp}(V)}(\gamma\circ\gp)\neq R_{\sigma_{\Gp}(V)}(\gamma)\circ R_{\sigma_{\Gp}(V)}(\gp) 
\eq yields. This is a general problem. In comparison with problem \ref{subsec fingraphpathgroup}.\ref{problem group structure on graphs systems} the multiplication map $\circ$ is not well-defined and hence
 \beqs R_{\sigma_{\Gp}(V)}(\gamma)\circ R_{\sigma_{\Gp}(V)}(\gp)
\eqs  is not well-defined. Recognize that, $R_{\sigma_{\Gp}(V)}:\PD_{\Gamma}\rightarrow\PD_\Gamma$.

Consequently in general it is not true that, 
 \beq\label{eq ineqsolv2}(\ho_\Gamma\circ R_{\sigma_{\Gp}(V)})(\gamma\circ\gp)=\ho_\Gamma(R_{\sigma_{\Gp}(V)}(\gamma)\circ R_{\sigma_{\Gp}(V)}(\gp))=(\ho_\Gamma\circ R_{\sigma_{\Gp}(V)})(\gamma)\cdot (\ho_\Gamma\circ R_{\sigma_{\Gp}(V)})(\gp)
\eq yields.
\end{problem}

With no doubt the left-tranlation $L_{\sigma_\Gp}$ and the inner autmorphisms $I_{\sigma_\Gp}$ in a finite graph system $\PD_\Gamma$ for every $\Gp\in\PD_\Gamma$ are defined similarly.

\begin{defi}Let $\sigma_\Gamma\in \mathfrak{B}(\PD_\Gamma)$ be a bisection in the finite graph system $\PD_\Gamma$. Let $R_{\sigma_\Gamma(V)}$ be a right-tranlation, where $V$ is a subset of $V_\Gamma$.

Then the pair $(\Phi_\Gamma,\varphi_\Gamma)$ defined by $\Phi_\Gamma=R_{\sigma_\Gamma(V)}$ (or, respectively, $\Phi_\Gamma=L_{\sigma_\Gamma(V)}$, or $\Phi_\Gamma=I_{\sigma_\Gamma(V)}$) for a subset $V\subseteq V_\Gamma$ and $\varphi_\Gamma=t\circ \sigma_\Gamma$ is called a \textbf{graph-diffeomorphism of a finite graph system}. 
Denote the set of finite graph-diffeomorphisms by $\Diff(\PD_\Gamma)$.
\end{defi}

Let $\Gp$ be a subgraph of $\Gamma$ and $\sigma_\Gp$ be a restriction of  bisection $\sigma_\Gamma$ in $\PD_\Gamma$. Then for example another graph-diffeomorphism $(\Phi_\Gp,\varphi_\Gp)$ in $\Diff(\PD_\Gamma)$ is defined by $\Phi_\Gp=R_{{\sigma_\Gp}(V)}$ for a subset $V\subseteq V_\Gp$ and $\varphi_\Gp=t\circ \sigma_\Gp$.

Remembering that the set of bisections of a finite path groupoid forms a group (refer \ref{lem groupbisection}) one may ask if the bisections of a finite graph system form a group, too. 

\begin{prop}\label{lemma bisecform}The set of bisections $\mathfrak{B}(\PD_\Gamma)$ in a finite graph system $\PD_\Gamma$ forms a group.
\end{prop}
\begin{proofs}Let $\Gamma$ be a graph and let $V_\Gamma$ be equivalent to the set $\{v_1,...,v_{2N}\}$.
  
First two different multiplication operations are studied. The studies are comparable with the results of the definition \ref{defi bisecongraphgroupioid} of a right-translation in a finite graph system. The easiest multiplication operation is given by $\ast_1$, which is defined by 
\beqs (\sigma\ast_1\sigma^\prime)(V_\Gamma)&:=\{
(\tilde\sigma\ast \tilde\sigma^\prime)(v_1),...,(\tilde\sigma\ast \tilde\sigma^\prime)(v_{2N}):v_i\in V_\Gamma\}
\eqs where $\ast$ denotes the multiplication of bisections on the finite path groupoid $\fPG$. Notice that, this operation is not well-defined in general. In comparison with the definition of the right-translation in a finite graph system one has to take care. First the set of vertices doesn't contain any vertices twice, the map $\sigma$ in the finite path system is bijective, the mapping $\sigma$ maps each set to a set of vertices containing no  
vertices twice and the situation in problem \thesection.\ref{prob withoutcond} has to be avoided.

Fix a bisection $\tilde\sigma$ in a finite path groupoid $\fPG$.  Let $V_{\sigma^\prime}$ be a subset of $V_\Gamma$ where $\Gamma:=\{\gamma_1,...,\gamma_N\}$ and for each $v_i$ in $V_{\sigma^\prime}$ it is true that $v_i\neq v_j$ and $v_i\neq (t\circ\tilde\sigma^\prime)(v_j)$ for all $i\neq j$. Define the set $V_{\sigma,\sigma^\prime}$ to be equal to a subset of the set of all vertices $\{v_k\in V_{\sigma^\prime}: 1\leq k\leq 2N\}$ such that each pair $(v_i,v_j)$ of vertices in $V_{\sigma,\sigma^\prime}$ satisfies $(t\circ(\tilde\sigma\ast\tilde\sigma^\prime))(v_i)\neq (t\circ\tilde\sigma^\prime)(v_j)$ and $(t\circ\tilde\sigma^\prime)(v_i)\neq (t\circ\tilde\sigma^\prime)(v_j)$ for all $i\neq j$.
Define
\beqs W_{\sigma,\sigma^\prime}:=\Big\{
w_i\in \{V_{\sigma}\cap V_{\sigma^\prime}\}\setminus V_{\sigma,\sigma^\prime}:
& (t\circ\tilde\sigma)(w_j)\neq (t\circ\tilde\sigma^\prime)(w_i)\quad\forall i\neq j,\quad 1\leq i,j\leq l\Big\}
\eqs
The set $V_{\sigma,\sigma^\prime,\breve \sigma}$ is a subset of all vertices $\{v_k\in V_{\sigma,\sigma^\prime}: 1\leq k\leq 2N\}$ such that each pair $(v_i,v_j)$ of vertices in $V_{\sigma,\sigma^\prime,\breve\sigma}$ satisfies $(t\circ(\breve\sigma\ast\tilde\sigma\ast\tilde\sigma^\prime))(v_i)\neq (t\circ(\tilde\sigma\ast\tilde\sigma^\prime))(v_j)$ and $(t\circ(\tilde\sigma\ast\tilde\sigma^\prime))(v_i)\neq (t\circ(\tilde\sigma\ast\tilde\sigma^\prime))(v_j)$ for all $i\neq j$.

Consequently define a second multiplication on $\mathfrak{B}(\PD_\Gamma)$ similarly to the operation $\ast_1$. This is done by the following definition. Set
\beqs (\sigma\ast_2\sigma^\prime)(V_{\sigma^\prime})
:=&\left\{
(\tilde\sigma\ast \sigma^\prime)(v_1),...,(\tilde\sigma\ast\sigma^\prime)(v_{k}):
v_1,...,v_k\in V_{\sigma,\sigma^\prime}, 1\leq k\leq 2N\right\}\\
&\cup\left\{ \tilde\sigma(w_1),\sigma^\prime(w_1)...,\tilde\sigma(w_{l}),\sigma^\prime(w_{l}): w_1,...,w_l\in W_{\sigma,\sigma^\prime}, 1\leq l\leq 2N
\right\}\\
&\cup\left\{\idf_{p_1},...,\idf_{p_n}:p_1,...,p_n\in V_{\sigma^\prime}\setminus \{V_{\sigma,\sigma^\prime}\cup W_{\sigma,\sigma^\prime}\}, 1\leq n\leq 2N\right\}
\eqs 

Hence the inverse is supposed to be $\sigma^{-1}(V_\Gamma)=\sigma((t\circ\sigma)^{-1}(V_\Gamma))^{-1}$ such that
\beqs (\sigma\ast_2\sigma^{-1})(V_{\sigma^{-1}})=&\{ (\tilde\sigma\ast \tilde\sigma^{-1})(v_1),...,(\tilde\sigma\ast \tilde\sigma^{-1})(v_{2N}):v_i\in V_{\sigma,\sigma^{-1}}\}\\
&\cup\left\{ \tilde\sigma(w_1),\sigma^\prime(w_1)^{-1}...,\tilde\sigma(w_{l}),\sigma^\prime(w_{l})^{-1}: w_1,...,w_l\in W_{\sigma,\sigma^{-1}}, 1\leq l\leq 2N
\right\}\\
&\cup\left\{\idf_{p_1},...,\idf_{p_n}:p_1,...,p_n\in V_{\sigma^\prime}\setminus \{V_{\sigma,\sigma^{-1}}\cup W_{\sigma,\sigma^{-1}}\}, 1\leq n\leq 2N\right\}
\eqs
\end{proofs} 
Notice that, the problem \thesection.\ref{prob righttranslgrouoid} is solved by a multiplication operation $\circ_2$, which is defined similarly to $\ast_2$. Hence the equality of \eqref{eq ineqsolv} is available and consequently \eqref{eq ineqsolv2} is true. Furthermore a similar remark to \ref{def compositionofdiffeo} can be done.

\begin{exa}
Now consider the following example. Set $\Gp:=\{\gamma_1,\gamma_3\}$, let $\Gamma:=\{\gamma_1,\gamma_2,\gamma_3\}$ and\\ $V_\Gamma:=\{s(\gamma_1),t(\gamma_1),s(\gamma_2),t(\gamma_2),s(\gamma_3),t(\gamma_3): s(\gamma_i)\neq s(\gamma_j),t(\gamma_i)\neq t(\gamma_j)\text{ }\forall i\neq j\}$. \\
Set $V$ be equal to $\{s(\gamma_1),s(\gamma_2),s(\gamma_3)\}$. Take two maps $\sigma$ and $\sigma^\prime$ such that $\sigma^\prime(V)=\{\gamma_1,\gamma_3\}$, $\sigma(V)=\{ \gamma_2\}$, where $(t\circ\tilde\sigma)(s(\gamma_3))=t(\gamma_3)$, $\tilde\sigma^\prime(s(\gamma_3))=\gamma_3$, $\tilde\sigma^\prime(s(\gamma_1))=\gamma_1$ and  $\tilde\sigma(t(\gamma_3))=\gamma_2$. Then $s(\gamma_3)\in V_{\sigma_{\Gp},\sigma^\prime_\Gp}$ and $s(\gamma_1)\in W_{\sigma_{\Gp},\sigma^\prime_\Gp}$.
Derive 
\beqs (\sigma\ast_1\sigma^\prime)(V)= \{\gamma_3\circ\gamma_2,\gamma_1\}
\eqs 
Then conclude that,
\beqs (\sigma\ast_2\sigma^\prime)(V_\Gamma)= \{\gamma_3\circ\gamma_2,\gamma_1\}
\eqs holds. Notice that
\beqs (\sigma\ast_2\sigma^\prime)(V)\neq (\sigma^\prime\ast_2\sigma)(V)= \{\gamma_2,\gamma_1,\gamma_3\}
\eqs is true.
Finally obtain
\beqs (\sigma\ast_2\sigma^{-1})(V_\Gamma)= \{\gamma_3\circ\gamma_3^{-1},\gamma_1\circ\gamma_1^{-1}\}=\{\idf_{s(\gamma_3)},\idf_{s(\gamma_1)}\}
\eqs 
Let $\sigma^\prime(V_\Gamma)=\{\gamma_1,\gamma_3\}$ and $\breve \sigma(V_\Gamma)=\{\gamma_2,\gamma_4\}$. Then notice that,
\beqs (\breve \sigma\ast_1\sigma^\prime)(V_\Gamma)=\{\gamma_3\circ\gamma_2,\gamma_1\}
\eqs and 
\beqs (\breve \sigma\ast_2\sigma^\prime)(V_\Gamma)=\{\gamma_3\circ\gamma_2,\gamma_1,\gamma_4\}
\eqs yields.

Furthermore assume supplementary that $t(\gamma_3)=t(\gamma_1)$ holds. Then calculate the product of the maps $\sigma$ and $\sigma^\prime$:
\beqs (\sigma\ast_1\sigma^\prime)(V)= \{\gamma_3\circ\gamma_2,\gamma_1\circ\gamma_2\}\notin\PD_\Gamma
\eqs  and
\beqs (\sigma\ast_2\sigma^\prime)(V_\Gamma)= \{\idf_{t(\gamma_1)}, \idf_{t(\gamma_3)}\}\in\PD_\Gamma
\eqs
\end{exa}

The group structure of $\mathfrak{B}(\PD_\Gamma)$ transferes to $G$. Let $\tilde\sigma$ be a bisection in the finite path groupoid $\fPSGm$, which defines a bisection $\sigma$ in $\PD_\Gamma$ and let $\tilde\sigma^\prime$ be a bisection in $\fPSGm$, which defines another bisection $\sigma^\prime$ in $\PD_\Gamma$. Let $V_{\sigma,\sigma^\prime}$ be equal to $V_\Gamma$, then derive
\beq &\ho_\Gamma\left((\sigma\ast_2\sigma^\prime)(V_\Gamma)\right) 
= \{\ho_\Gamma((\tilde\sigma\ast\sigma^\prime)(v_1)),...,\ho_\Gamma((\tilde\sigma\ast\sigma^\prime)(v_{2N})) \}\\ 
&=\ho_\Gamma(\sigma^\prime(V_\Gamma)\circ\sigma(t(\sigma^\prime(V_\Gamma))))
=\{\ho_\Gamma(\sigma^\prime(v)\circ\tilde\sigma(t(\sigma^\prime(v_1)))),...,\ho_\Gamma(\sigma^\prime(v_{N})\circ\tilde\sigma(t(\sigma^\prime(v_{N}))))\}\\
&=\{\ho_\Gamma(\sigma^\prime(v))\ho_\Gamma(\tilde\sigma(t(\sigma^\prime(v_1)))),...,\ho_\Gamma(\sigma^\prime(v_{N}))\ho_\Gamma(\tilde\sigma(t(\sigma^\prime(v_{N}))))\}\\
&= \ho_\Gamma(\sigma^\prime(V_\Gamma))\ho_\Gamma(\sigma(V_\Gamma))
\eq
Consequently the right-translation in the finite product $G^{\vert\Gamma\vert}$ is definable.

\begin{defi}Let $\sigma_\Gp$ be in $\mathfrak{B}(\PD_\Gamma)$, $\Gp$ a subgraph of $\Gamma$, $\Gpp$ a subgraph of $\Gp$ and $R_{\sigma_\Gp}$ a right-translation, $L_{\sigma_\Gp}$ a left-translation and $I_{\sigma_\Gp}$ an inner-translation in $\PD_\Gamma$.

Then the \textbf{right-translation in the finite product $G^{\vert\Gamma\vert}$} is given by
\beqs \ho_\Gamma\circ R_{\sigma_\Gp}:\PD_\Gamma\rightarrow G^{\vert\Gamma\vert}, \quad \Gpp\mapsto (\ho_\Gamma\circ R_{\sigma_\Gp})(\Gpp)
\eqs
Furthermore define the \textbf{left-translation in the finite product $G^{\vert\Gamma\vert}$} by
\beqs \ho_\Gamma\circ L_{\sigma_\Gp}:\PD_\Gamma\rightarrow G^{\vert\Gamma\vert},\quad \Gpp\mapsto (\ho_\Gamma\circ L_{\sigma_\Gp})(\Gpp)
\eqs
and the \textbf{inner-translation in the finite product $G^{\vert\Gamma\vert}$}
\beqs \ho_\Gamma\circ I_{\sigma_\Gp}:\PD_\Gamma\rightarrow G^{\vert\Gamma\vert},\quad \Gpp\mapsto (\ho_\Gamma\circ I_{\sigma_\Gp})(\Gpp)
\eqs such that $I_{\sigma_\Gp}=L_{\sigma_\Gp^{-1}}\circ R_{\sigma_\Gp}$.
\end{defi}
\begin{lem}It is true that $R_{\sigma_\Gp\ast_2\sigma_\Gp^\prime}=R_{\sigma_\Gp}\circ R_{\sigma_\Gp^\prime}$, $L_{\sigma_\Gp\ast_2\sigma_\Gp^\prime}=L_{\sigma_\Gp}\circ L_{\sigma_\Gp^\prime}$ and $I_{\sigma_\Gp\ast_2\sigma_\Gp^\prime}=I_{\sigma_\Gp}\circ I_{\sigma_\Gp^\prime}$ for all bisections $\sigma_\Gp$ and $\sigma^\prime_\Gp$ in $\mathfrak{B}(\PD_\Gamma)$.
\end{lem}

There is an action of $\mathfrak {B}(\PD_\Gamma)$ on $G^{\vert \Gamma\vert}$ by
\beqs (\zeta_{\sigma_\Gp}\circ\ho_\Gamma)(\Gpp):= (\ho_\Gamma\circ R_{\sigma_\Gp})(\Gpp)
\eqs whenever $\sigma_\Gp\in \mathfrak {B}(\PD_\Gamma)$, $\Gpp\in\PD_\Gp$ and $\Gp\in \PD_\Gamma$.
Then for another $\breve\sigma\in \mathfrak {B}(\PD_\Gamma)$ it is true that,
\beqs ((\zeta_{\breve\sigma_\Gp}\circ\zeta_{\sigma_\Gp})\circ\ho_\Gamma)(\Gpp)= (\ho_\Gamma\circ R_{\breve\sigma\ast_2\sigma_\Gp})(\Gpp)=(\zeta_{\breve\sigma_\Gp\ast_2\sigma_\Gp}\circ\ho_\Gamma)(\Gpp)
\eqs yields.

Recall that, the map $\tilde\sigma\mapsto t\circ\tilde\sigma$ is a group isomorphism between the group of bisections $\mathfrak {B}(\PD_\Gamma\Sigma)$ and the group $\Diff(V_\Gamma)$ of finite diffeomorphisms in $V_\Gamma$. Therefore if the graphs $\Gp=\Gpp$ contain only the path $ \gamma$, then the action $\zeta_{\sigma_\Gp}$ is equivalent to an action of the finite diffeomorphism group $\Diff(V_\Gamma)$. Loosely speaking, the graph-diffeomorphisms $(R_{\sigma_\Gp(V)},t\circ\sigma_\Gp)$ on a subgraph $\Gpp$ of $\Gp$ transform graphs and respect the graph structure of $\Gp$. The diffeomorphism $t\circ\tilde\sigma$ in the finite path groupoid only implements the finite diffeomorphism in $\Sigma$, but it doesn't adopt any path groupoid or graph preserving structure. Summarising the bisections of a finite graph system respect the graph structure and implement the finite diffeomorphisms in $\Sigma$. There is another reason why the group of bisections is more fundamental than the path- or graph-diffeomorphism group. In \cite{Kaminski1,KaminskiPHD} the concept of  $C^*$-dynamical systems has been studied. It turns out that, there are three different $C^*$-dynamical systems, each is build from the analytic holonomy $C^*$-algebra and a point-norm continuous action of the group of bisections of a finite graph system. The actions are implemented by one of the three translations, i.e. the left-, right- or inner-translation in the finite product $G^{\vert\Gamma\vert}$.

\paragraph*{Transformations and discretised surface sets\\[5pt]}

Now restricted sets of bisections are concerned. Consider a finite set of paths starting at a discretised surface. The idea is to define a bisection $\sigma$ such that the map $t\circ\sigma$ preserves the set $\breve S_{\disc}$ of discretised surfaces and each path of the certain set of paths composed with the bisection $\sigma$ at the target vertex of this path is again a path that start at a discretised surface. The definition follows.

Define the set $V^{\breve S_{\disc}}$, which contains all target vertices of paths in $\PD^{\breve S_{\disc}}_\Gamma\Sigma$. Note that, the base point $v$ with respect to $\PD_\Gamma\Sigma^v$ is the source vertex of all paths in $\PD_\Gamma\Sigma^v$ and all paths in $\PD_\Gamma\Sigma^v$ are contained in $\PD^{\breve S_{\disc}}_\Gamma\Sigma$ for all $v\in\breve S_{\disc}$. 
Denote the set of bisections, which are bijective maps from the set $V^{\breve S_{\disc}}$ to paths in $\PD^{\breve S_{\disc}}_\Gamma\Sigma$, by $\mathfrak{B}(\PD^{\breve S_{\disc}}_\Gamma\Sigma)$. On the level of graphs the restricted set of bisections in a graph system $\PD_\Gamma$ is denoted by  $\mathfrak{B}(\PD^{\breve S_{\disc}}_\Gamma)$. Denote the set of graph-diffeomorphisms, which are defined by a bisection in $\mathfrak{B}(\PD^{\breve S_{\disc}}_\Gamma)$ and the right-translation $R_\sigma$,  by the term $\Diff(\PD^{\breve S_{\disc}}_\Gamma)$.

\subsection{The quantum flux operators associated to surfaces and graphs}\label{sec fluxdef}
\paragraph*{The quantum flux operators presented by Lie algebra elements associated to surfaces and graphs\\[5pt]}

The quantum analog of a classical connection $A_a(v)$ is given by the holonomy along a path $\gamma$ and is denoted by $\ho(\gamma)$. The quantum flux operator $E_S(\gamma)$, which replaces the classical flux variable $E(S,f^S)$, is given by a map $E_S$ from a graph to the Lie algebra $\go$. Let $\Exp$ be the exponential map from the Lie algebra $\go$ to $G$ and set $U_t(E_S(\gamma)):=\Exp(tE_S(\gamma))$. Then the quantum flux operator $E_S(\gamma)$ and the quantum holonomies $\ho(\gamma)$ satisfy the following canonical commutator relation 
\[  E_S(\gamma)\ho(\gamma) = i\frac{\dif}{\dif t}\Big\vert_{t=0}U_t(E_S(\gamma))\ho(\gamma)\] where $\gamma$ is a path that intersects the surface $S$ in the target vertex of the path and lies below with respect to the surface orientation of $S$.

In this section different definitions of the quantum flux operator, which is associated to a fixed surface $S$, are presented. For example the quantum flux operator $E_S$ is defined to be a map from a graph $\Gamma$ to a direct sum $\go\oplus\go$ of the Lie algebra $\go$ associated to the Lie group $G$. This is related to the fact that, one distinguishes between paths that are ingoing and paths that are outgoing with resepect to the surface orientation of $S$. If there are no intersection points of the surface $S$ and the source or target vertex of a path $\gamma_i$ of a graph $\Gamma$, then the map maps the path $\gamma_i$ to zero in both entries. For different surfaces  or for a fixed surface different maps will refer to different quantum flux operators.  Furthermore, the quantum flux operators is also defined as maps form the graph $\Gamma$ to direct sum $\E\oplus\E$ of the universal enveloping algebra $\E$ of $\go$.
 
\begin{defi}\label{defi intersefunc}Let $\breve S$ be a finite set $\{S_i\}$ of surfaces in $\Sigma$, which is closed under a flip of orientation of the surfaces. Let $\Gamma$ be a graph such that each path in $\Gamma$ satisfies one of the following conditions 
\begin{itemize}
 \item the path intersects each surface in $\breve S$ in the source vertex of the path and there are no other intersection points of the path and any surface contained in $\breve S$,
 \item the path intersects each surface in $\breve S$ in the target vertex of the path and there are no other intersection points of the path and any surface contained in $\breve S$,
 \item the path intersects each surface in $\breve S$ in the source and target vertex of the path and there are no other intersection points of the path and any surface contained in $\breve S$,
 \item the path does not intersect any surface $S$ contained in $\breve S$.
\end{itemize}

Then define the intersection functions $\iota_L:\breve S\times \Gamma\rightarrow \{\pm 1,0\}$ such that
\beqs \iota_L(S,\gamma):=
\left\{\begin{array}{ll}
1 &\text{ for a path }\gamma\text{ lying above and outgoing w.r.t. }S\\
-1 &\text{ for a path }\gamma\text{ lying below and outgoing w.r.t. }S\\
0 &\text{ the path }\gamma\text{ is not outgoing w.r.t. }S
\end{array}\right.
\eqs
and the intersection functions $\iota_R:\breve S\times \Gamma\rightarrow\{\pm 1,0\}$ such that
\beqs \iota_L(S,\gamma):= \left\{\begin{array}{ll}
-1 &\text{ for a path }\gp\text{ lying above and ingoing w.r.t. }S\\
1 &\text{ for a path }\gp\text{ lying below and ingoing w.r.t. }S\\
0 &\text{ the path }\gp\text{ is not ingoing w.r.t. }S
\end{array}\right.
\eqs whenever $S\in\breve S$ and $\gamma\in\Gamma$.

Define a map $\sigma_L:\breve S\rightarrow \go$ such that
\beqs \sigma_L(S)&=\sigma_L(S^ {-1})
\eqs whenever $S\in\breve S$ and $S^ {-1}$ is the surface $S$ with reversed orientation. Denote the set of such maps by $\breve\sigma_L$. Respectively, the map $\sigma_R:\breve S\rightarrow \go$ such that
\beqs \sigma_R(S)&=\sigma_R(S^ {-1})
\eqs whenever $S\in\breve S$. Denote the set of such maps by $\breve\sigma_R$.
Moreover, there is a map $\sigma_L\times \sigma_R:\breve S\rightarrow \go\oplus\go$ such that
\beqs (\sigma_L,\sigma_R)(S)&=(\sigma_L,\sigma_R)(S^ {-1})
\eqs whenever $S\in\breve S$. Denote the set of such maps by $\breve\sigma$.

Finally, define the \textbf{Lie algebra-valued quantum flux set for paths}
\beqs \gop_{\breve S,\Gamma}
:=\bigcup_{\sigma_L\times\sigma_R\in\breve\sigma}\bigcup_{S\in\breve S}\Big\{& (E^L,E^R)\in\Map(\Gamma,\go\oplus\go): 
&(E^L, E^R)(\gamma):=(\iota_L(S,\gamma)\sigma_L(S),\iota_R(S,\gamma)\sigma_R(S))\Big\}
\eqs
where $\Map(\Gamma,\go\oplus\go)$ is the set of all maps from the graph $\Gamma$ to the direct sum  $\go\oplus\go$ of Lie algebras.   
\end{defi}

Observe that, $(\iota_L\times \iota_R)(S^{-1},\gamma)=(-\iota_L\times -\iota_R)(S,\gamma)$ holds for every $\gamma\in\Gamma$. 

Remark that, the condition $E^L(\gamma)=E^R(\gamma^{-1})$ is not required. 

\begin{exa}\label{exa Exa1}
Analyse the following example. Consider a graph $\Gamma$ and two disjoint surface sets $\breve S$ and $\breve T$.
\begin{center}
\includegraphics[width=0.45\textwidth]{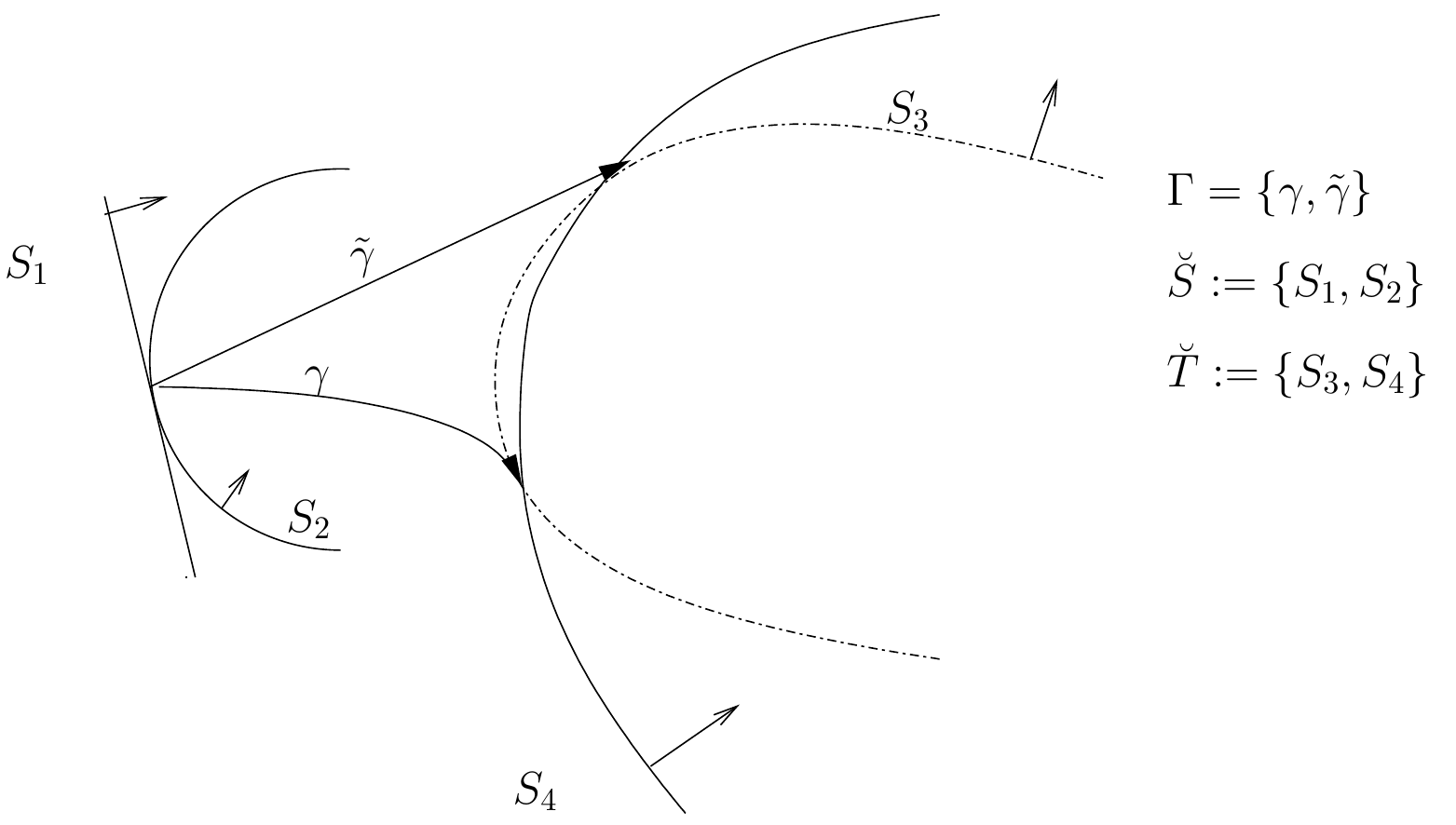}
\end{center}
Then the elements of $\gop_{\breve S,\Gamma}$ are for example given by the maps $E^L_{i}\times E^R_{i}$ for $i=1,2$ such that 
\beqs 
E_1(\gamma)&:= (E^L_{1}, E^R_{1})(\gamma)=(\iota_L(S_1,\gamma)\sigma_L(S_1),\iota_R(S_1,\gamma)\sigma_R(S_1))=(X_{1},0)\\
E_1(\tg)&:= (E^L_{1}, E^R_{1})(\tg)=(\iota_L(S_1,\tg)\sigma_L(S_1),\iota_R(S_1,\tg)\sigma_R(S_1))= (X_1,0)\\
E_2(\gamma)&:= (E^L_{2}, E^R_{2})(\gamma)=(\iota_L(S_2,\gamma)\sigma_L(S_2),\iota_R(S_2,\gamma)\sigma_R(S_2))
=(X_{2},0)\\
E_2(\tg)&:= (E^L_{2}, E^R_{2})(\tg)=(\iota_L(S_2,\tg)\sigma_L(S_2),\iota_R(S_2,\tg)\sigma_R(S_2))
=(X_{2},0)\\
E_3(\gamma)&:= (E^L_{3}, E^R_{3})(\gamma)=(\iota_L(S_3,\gamma)\sigma_L(S_3),\iota_R(S_3,\gamma)\sigma_R(S_3))
=(0,-Y_{3})\\
E_3(\tg)&:= (E^L_{3}, E^R_{3})(\tg)=(\iota_L(S_3,\tg)\sigma_L(S_3),\iota_R(S_3,\tg)\sigma_R(S_3))= (0,-Y_3)\\
E_4(\gamma)&:= (E^L_{4}, E^R_{4})(\gamma)=(\iota_L(S_4,\gamma)\sigma_L(S_4),\iota_R(S_4,\gamma)\sigma_R(S_4))
=(0,Y_{4})\\
E_4(\tg)&:= (E^L_{4}, E^R_{4})(\tg)=(\iota_L(S_4,\tg)\sigma_L(S_4),\iota_R(S_4,\tg)\sigma_R(S_4))= (0,Y_4)
\eqs 

This example shows that, the surfaces $\{S_1,S_2\}$ are similar, whereas the surfaces $\{T_1,T_2\}$ produce different signatures for different paths. Moreover, the set of surfaces are chosen such that one component of the direct sum is always zero. 
\end{exa}

For a particular surface set $\breve S$ the set 
\beqs\bigcup_{\sigma_L\times\sigma_R\in\breve\sigma}\bigcup_{S\in\breve S}
\Big\{ (E^L,E^R)\in\Map(\Gamma,\go\oplus\go): \quad(E^L, E^R)(\gamma):=(\iota_L(S,\gamma)\sigma_L(S),0)\Big\}\eqs is identified with 
\beqs\bigcup_{\sigma_L\in\breve\sigma_L}\bigcup_{S\in\breve S}\Big\{E\in\Map(\Gamma,\go): \quad
E(\gamma):=\iota_L(S,\gamma)\sigma_L(S)\Big\}
\eqs 
The same is observed for another surface set $\breve T$ and the set $\gop_{\breve T,\Gamma}$ is identifiable with 
\beqs\bigcup_{\sigma_R\in\breve\sigma_R}\bigcup_{T\in\breve T}
\Big\{E\in\Map(\Gamma,\go): \quad
E(\gamma):=\iota_R(T,\gamma)\sigma_R(T)\Big\}
\eqs

The intersection behavoir of paths and surfaces plays a fundamental role in the definition of the quantum flux operator. There are exceptional configurations of surfaces and paths in a graph. One of them is the following.
\begin{defi}
A set $\breve S$ of $N$ surfaces has the \hypertarget{simple surface intersection property for a graph}{\textbf{simple surface intersection property for a graph $\Gamma$}} with $N$ independent edges iff it contains only surfaces, for which each path $\gamma_i$ of a graph $\Gamma$ intersects only one surface $S_i$ only once in the target vertex of the path $\gamma_i$, the path $\gamma_i$ lies above and there are no other intersection points of each path $\gamma_i$ and each surface in $\breve S$. 
\end{defi}
\begin{exa}Consider the following example.
\begin{center}
\includegraphics[width=0.45\textwidth]{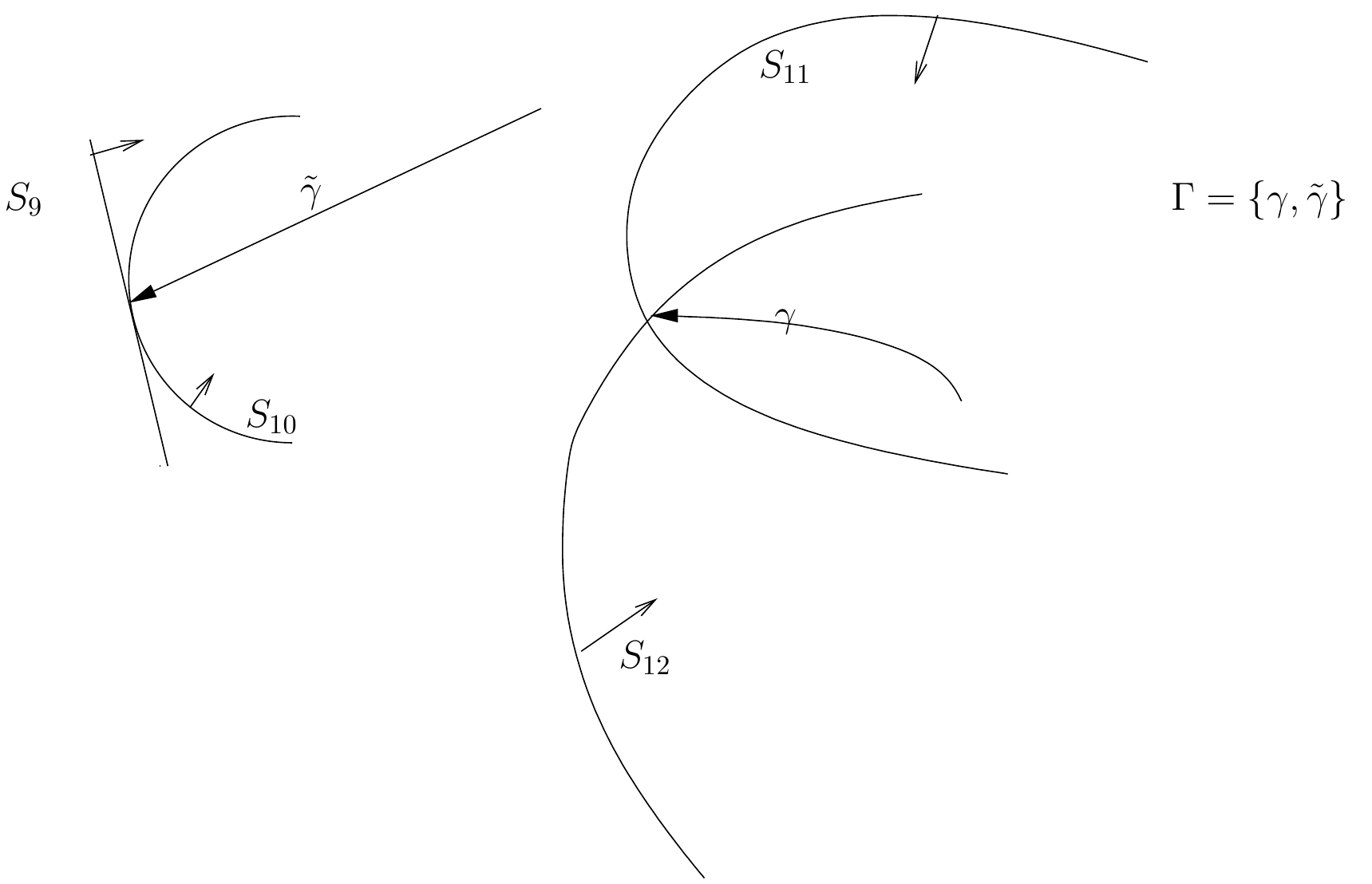}
\end{center} The sets $\{S_{9},S_{11}\}$ or $\{S_{10},S_{12}\}$ have the simple surface intersection property for the graph $\Gamma$.
Calculate
\beqs E_{9}(\tg)=(0,-Y_{9}),\quad E_{11}(\gamma)=(0,-Y_{11})
\eqs
\end{exa}
In this case the set $\gop_{\breve S,\Gamma}$ reduces to
\beqs\bigcup_{\sigma_R\in\breve\sigma_R}\bigcup_{S\in\breve S}\Big\{E\in\Map(\Gamma,\go): \quad
E(\gamma):=-\sigma_R(S)\text{ for }\gamma\cap S= t(\gamma)\Big\}
\eqs Notice that, the set $\Gamma\cap \breve S=\{t(\gamma_i)\}$ for a surface $S_i\in\breve S$ and $\gamma_i\cap S_j\cap S_i =\{\varnothing\}$ for a path $\gamma_i$ in $\Gamma$ and $i\neq j$.

On the other hand, the set of surfaces can be such that each path of a graph intersect all surfaces of the set in the same vertex. This contradicts the assumption that each path of a graph intersect only one surface once. 
\begin{defi}Let $\Gamma$ be a graph that contains no loops.

A set $\breve S$ of surfaces has the \hypertarget{same intersection property}{\textbf{same surface intersection property for a graph}} $\Gamma$ iff each path $\gamma_i$ in $\Gamma$ intersects with all surfaces of $\breve S$ in the same source vertex $v_i\in V_\Gamma$ ($i=1,..,N$), all paths are outgoing and lie below each surface $S\in\breve S$ and there are no other intersection points of each path $\gamma_i$ and each surface in $\breve S$. 
\end{defi}

Recall the example \thesection.\ref{exa Exa1}. Then the set $\{S_{1},S_{2}\}$ has the same surface intersection property for the graph $\Gamma$.

Then the set $\gop_{\breve S,\Gamma}$ reduces to
\beqs\bigcup_{\sigma_L\in\breve\sigma_L}\bigcup_{S\in\breve S}\Big\{E\in\Map(\Gamma,\go): \quad
E(\gamma):= -\sigma_L(S)\text{ for }\gamma\cap S= s(\gamma)\Big\}
\eqs Notice that, $\gamma\cap S_1\cap ...\cap S_N=s(\gamma)$ for a path $\gamma$ in $\Gamma$, whereas $\Gamma\cap\breve S=\{s(\gamma_i)\}_{1\leq i\leq N}$. Clearly, $\Gamma\cap S_i=s(\gamma_i)$ holds for a surface $S_i$ in $\breve S$. 

Simply speaking the physical intution is that, fluxes associated to different surfaces should act on the same path.

Notice that, both properties can be restated for other surface and path configurations. Hence, a surface set can have the simple or same surface intersection property for paths that are outgoing and lie above (or ingoing and below, or outgoing and below). The important fact is related to the question if the intersection vertices are the same for all surfaces or not.

In section \ref{subsec fingraphpathgroup} the concept of finite graph systems has been introduced. The following remark shows that, the properties simply generalises to this new structure.
\begin{rem}
A set $\breve S$ has the same surface intersection property for a finite orientation preserved graph system $\PD^{\op}_\Gamma$ associated to a graph $\Gamma$ (with no loops) iff the set $\breve S$ has the same surface intersection property the graph $\Gamma$.

A set $\breve S$ has the simple surface intersection property for a finite orientation preserved\footnote{Let $\breve S$ be equal to $S$. Then  notice that the property of all graphs being orientation preserved subgraphs is necessary, since, for a subgraph $\Gp:=\{\gp\}$ of $\Gamma$ the graph $\{\gp^{-1}\}$ is a subgraph of $\Gamma$, too. Consequently, if there is a surface $S$ intersecting a path $\gp$ such that $\gp$ is ingoing and lies above, then $S$ intersects the path $\gp^{-1}$ such that $\gp^{-1}$ is outgoing and lies above. This implies that, the surface $S$ cannot have the same surface intersection property for each subgraph of $\Gamma$.} graph system $\PD^{\op}_\Gamma$ associated to a graph $\Gamma$ if the set $\breve S$ has the simple surface intersection property for the graph $\Gamma$.
\end{rem}

\begin{defi}Let $\breve S$ be a surface set and $\Gamma$ be a graph such that the only intersections of the graph and each surface in $\breve S$ are contained in the vertex set $V_\Gamma$.

Then the set of images $\{E(\gamma): E\in\gop_{\breve S,\Gamma}\}$ of flux maps for a fixed path $\gamma$ in $\Gamma$ is denoted by $\bar\gop_{\breve S,\gamma}$. 
\end{defi}

\begin{prop}\label{prop Liealgebrastructfluxes}Let $\breve S$ be a set of surfaces and $\Gamma$ be a fixed graph (with no loops) such that the set $\breve S$ has the same surface intersection property for a graph $\Gamma$. Moreover, let $\breve T$ be a set of surfaces and $\Gamma$ be a fixed graph such that the set $\breve T$ has the simple surface intersection property for a graph $\Gamma$.

Then the set $\bar\gop_{\breve S,\gamma}$ is equipped with a structure, which is induced from the Lie algebra structure of $\go$, such that it forms a Lie algebra. 
The the set $\bar\gop_{\breve T,\gamma}$ is equipped with a structure to form a Lie algebra, too.
\end{prop}
\begin{proofs}
\textbf{Step 1: linear space over $\CB$}\\Consider a path $\gamma$ in $\Gamma$ that lies above and ingoing w.r.t. the surface orientation of each surface $S$ in $\breve S$ and ingoing and above with respect to $T$. Then there is a map $E_S$ such that
\beqs E_S(\gamma)=- X
\eqs There exists an operation $+$ given by the map $s: \bar\gop_{\breve S,\gamma}\times\bar\gop_{\breve S,\gamma} \rightarrow \bar\gop_{\breve S,\gamma}$ such that 
\beqs (E^L_{1}(\gamma), E^L_{2}(\gamma))\mapsto s(E^L_{1}(\gamma), E^L_{2}(\gamma)):=E^L_{1}(\gamma)+ E^L_{2}(\gamma)=-\sigma^1_L(S_1)-\sigma^2_L(S_2)
=-\sigma^3_L([S])
\eqs since $\sigma_L^i\in\breve\sigma_L$ and where $[S]$ denotes an arbitrary representative of the set $\breve S$.
Respectively it is defined 
\beqs (E^L_{1}(\gamma), E^L_{2}(\gamma))\mapsto s(E^L_{1}(\gamma), E^L_{2}(\gamma)):=E^L_{1}(\gamma)+ E^L_{2}(\gamma)=-\sigma_R^1(T)-\sigma^2_R(T)=-\sigma_R^3(T)
\eqs whenever $\sigma^i_R\in\breve\sigma_R$ and $T\in\breve T$.
There is an inverse
\beqs E(\gamma) - E(\gamma)=X - X=0\eqs
and a null element
\beqs E(\gamma) +E_0(\gamma)=X\eqs
whenever $E_0(\gamma)=-\sigma_L(S) =0$.
Notice the following map 
\beq\label{eq Liealghom1}\bar\gop_{\breve S,\gamma}\times\bar\gop_{\breve S,\gp}\ni( E_{1}(\gamma), E_{2}(\gp))\mapsto E_{1}(\gamma)\dot{+} E_{2}(\gp)\in\go
\eq is not considered, since, this map is not well-defined. 
One can show easily that $(\bar\gop_{\breve S,\gamma},+)$ is an additive group. The scalar multiplication is defined by
\beqs \lambda\cdot E(\gamma)=\lambda X
\eqs for all $\lambda\in\CB$ and $X\in\go$.
Finally, prove that $(\bar\gop_{\breve S,\gamma},+)$ is a linear space over $\CB$.

\textbf{Step 2: Lie bracket} is defined by the Lie bracket of the Lie algebra $\go$ and
\beqs \bra E_{1}(\gamma),E_{2}(\gamma)\ket: =\bra X_{1},X_{2}\ket\eqs
for $E_{1}(\gamma),E_{2}(\gamma)\in\bar\gop_{\breve S,\gamma}$ and $\gamma\in\Gamma$. 
\end{proofs}

If a surface set $\breve S$ does not have the same or simple surface intersection property for the graph $\Gamma$, then the surface set can be decomposed into several sets and the graph $\Gamma$ can be decomposed into a set of subgraphs. Then for each modified surface set there is a subgraph such that required condition is fulfilled. 

\begin{defi}Let $\breve S$ a set of surfaces and $\Gamma$ be a fixed graph (with no loops) such that the set $\breve S$ has the same (or simple) surface intersection property for a graph $\Gamma$. 

The universal enveloping Lie algebra of the Lie algebra $\bar\gop_{\breve S,\gamma}$ of fluxes for paths of a path $\gamma$ in $\Gamma$ and all surfaces in $\breve S$ is called the \textbf{universal enveloping flux algebra $\bar\Ep_{\breve S,\gamma}$ associated to a path and a finite set of surfaces}.
\end{defi}

Now, the definitions are rewritten for finite orientation preserved graph systems. 

\begin{defi}Let $\breve S$ be a surface set and $\Gamma$ be a graph such that the only intersections of the graph and each surface in $\breve S$ are contained in the vertex set $V_\Gamma$. $\PD_\Gamma$ denotes the finite graph system associated to $\Gamma$. Let $\E$ be the universal Lie enveloping algebra of $\go$.

Define the set of \textbf{Lie algebra-valued quantum fluxes for graphs}
\beqs \go_{\breve S,\Gamma}:= \bigcup_{\sigma_L\times\sigma_R\in\breve\sigma}\bigcup_{S\in\breve S}\Big\{ E_{S,\Gamma}\in\Map(\PD_\Gamma,\bigoplus_{\vert E_\Gamma\vert}\go\oplus \bigoplus_{\vert E_\Gamma\vert}\go):\quad 
&E_{S,\Gamma}:=E_S\times...\times E_S\\&\text{ where }E_S(\gamma):=(\iota_L(\gamma,S)\sigma_L(S),\iota_R(\gamma,S)\sigma_R(S)),\\
&E_S\in\gop_{\breve S,\Gamma},S\in\breve S,\gamma\in\Gamma\Big\}\eqs 

Moreover, define
\beqs \E_{\breve S,\Gamma}:= 
\bigcup_{\sigma_L\times\sigma_R\in\breve\sigma}\bigcup_{S\in\breve S}\Big\{ E_{S,\Gamma}\in\Map(\PD_\Gamma,\bigoplus_{\vert E_\Gamma\vert}\E\oplus \bigoplus_{\vert E_\Gamma\vert}\E):\quad 
&E_{S,\Gamma}:=E_S\times...\times E_S\\&\text{ where }E_S(\gamma):=(\iota_L(\gamma,S)\sigma_L(S),\iota_R(\gamma,S)\sigma_R(S)),\\
&E_S\in\E_{\breve S,\Gamma},S\in\breve S,\gamma\in\Gamma\Big\}
\eqs 

The set of all images of the linear hull of all maps in $\go_{\breve S,\Gamma}$ for a fixed surface set $\breve S$ and a fixed graph $\Gamma$ is denoted by $\bar\go_{\breve S,\Gamma}$.
The set of all images of the linear hull of all maps in $\go_{\breve S,\Gamma}$ for a fixed surface set $\breve S$ and a fixed subgraph $\Gp$ of $\Gamma$ is denoted by $\bar\go_{\breve S,\Gp\leq \Gamma}$. 
\end{defi}

Note that, the set of Lie algebra-valued quantum fluxes for graphs is generalised for the inductive limit graph system $\PD_{\Gamma_\infty}$. This follows from the fact that, each element of the inductive limit graph system $\PD_{\Gamma_\infty}$ is a graph. 

\begin{prop}
Let $\breve S$ be a set of surfaces and $\PD^{\op}_\Gamma$ be a finite orientation preserved graph system such that the set $\breve S$ has the same surface intersection property for a graph $\Gamma$ (with no loops).  

The set $\bar\go_{\breve S,\Gamma}$ forms a Lie algebra and is called the \textbf{Lie flux algebra associated a graph and a finite surface set}.The \textbf{universal enveloping flux algebra $\bar\E_{\breve S,\Gamma}$ associated a graph and a finite surface set} is the enveloping Lie algebra of $\bar\go_{\breve S,\Gamma}$.
\end{prop}
\begin{proof}
This follows from the observation that, $\go_{\breve S,\Gamma}$ is identified with
\beqs  \bigcup_{\sigma_L\in\breve\sigma_L}\bigcup_{S\in\breve S}\Big\{ E_{S,\Gamma}\in\Map(\PD^{\op}_\Gamma,\bigoplus_{\vert E_\Gamma\vert}\go):\quad 
&E_{S,\Gamma}:=E_S\times...\times E_S\\&\text{ where }E_S(\gamma):=-\sigma_L(S),
E_S\in\gop_{\breve S,\Gamma},S\in\breve S,\gamma\in\Gamma\Big\}\eqs 
and the addition operation
\beqs
E^1_{S_1,\Gamma}(\Gamma) + E^2_{S_2,\Gamma}(\Gamma)&:=
\big(E^1_{S_1}(\gamma_1) +E^2_{S_2}(\gamma_1),...,E^1_{S_1}(\gamma_N) +E^2_{S_2}(\gamma_N)\big)\\
&=(-\sigma^1_L(S_1)-\sigma^2_L(S_2),...,-\sigma^1_L(S_1)-\sigma^2_L(S_2))\\
&=(E^3_{[S]}(\gamma_1),...,E^3_{[S]}(\gamma_N))
\eqs whenever $\Gamma:=\gamma_1,...,\gamma_N$.
\end{proof}
Notice that indeed it is true that,
\beqs \go_{\breve S,\Gamma}=\go_{S_i,\Gamma}
\eqs yields for every $S_i\in\breve S$. The more general definition is due to physical arguments.

\begin{prop}
Let $\breve T$ be a set of surfaces and $\PD^{\op}_\Gamma$ be a finite orientation preserved graph system such that the set $\breve T$ has the simple surface intersection property for $\Gamma$. 

The set $\bar\go_{\breve T,\Gamma}$ forms a Lie algebra.
\end{prop}
Notice this follows from the fact that $\go_{\breve T,\Gamma}$ reduces to
\beqs \bigcup_{\sigma_L\in\breve\sigma_L}\Big\{ E_{\breve T,\Gamma}\in\Map(\PD^{\op}_\Gamma,\bigoplus_{\vert E_\Gamma\vert}\go):\quad 
&E_{\breve T,\Gamma}:=E_{T_1}\times...\times E_{T_N}\\&\text{ where }E_{T_i}(\gamma_i):=-\sigma_L(T_i),E_S\in\gop_{\breve S,\Gamma},T_i\in\breve T,\\
&\gamma_i\cap T_i=t(\gamma_i),\gamma\in\Gamma\Big\}\eqs 
since,
\beqs E_{S_1,\Gamma}(\Gamma)+ ...+ E_{S_N,\Gamma}(\Gamma)&= (E_{T_1}(\gamma_1),0,...,0)) + (0,E_{T_2}(\gamma_2),0,...,0)) + ...+ (0,...,0,E_{T_N}(\gamma_N))) \\
&=(E_{T_1}(\gamma_1),...,E_{T_N}(\gamma_N))=:E_{\breve T,\Gamma}(\Gamma)
\eqs

The Lie flux algebra and the universal enveloping flux algebra for the inductive limit graph system $\PD_{\Gamma_\infty}$  and a fixed suitable surface set $\breve S$ are denoted by $\bar\go_{\breve S}$ and $\bar\E_{\breve S}$. 

Finally assume that $\bar G_{\breve S,\Gamma}$, resp. $\bar G_{\breve S}$, denotes the Lie flux group associated to Lie flux algebra $\bar\go_{\breve S,\Gamma}$, resp. $\bar\go_{\breve S}$.

\paragraph*{The discretised and localised quantum flux operator associated to surfaces and graphs\\[5pt]}

Now consider a restriction of the quantum flux operators to discretised surfaces and graphs. Notice that, the Lie algebra-valued quantum flux operator usually distinguishes between paths, which are lying below, and paths, which are lying above the surface in a surface set. For simplicity in this article the case of paths lying below is considered only. With no doubt the second case can be defined analogously. The discretised surfaces do not allow to distinguish between paths lying above or below with respect to a surface orientation of a surface. Hence in this situation the discretised surface set has to be associated to a set of surfaces with a surface orientation. Summarising the cases below or above are not treated in the context of discretised surfaces. In this way, the intersection functions of definition \ref{defi intersefunc} are maps such that $\iota_L:\breve S_{\disc}\times\Gamma\rightarrow\{0,-1\}$ and $\iota_R:\breve S_{\disc}\times\Gamma\rightarrow\{0,1\}$.

\begin{defi}\label{defi locfluxop} Let $\breve S_{\disc}$ be a set of discretised surfaces, which is constructed from a set $\breve S$ of surfaces, and let $\Gamma$ be a graph.
Let $\{\Gamma_i\}_{i=1,..\infty}$ be an inductive family of graphs such that for every graph $\Gamma_i$ the intersection points of a surface set $\breve S_{\disc}$ and the graph $\Gamma_i$ are vertices of $V_{\Gamma_i}$. Denote the set of intersections of a graph $\Gamma_i$ and a discretised surface set $\breve S_{\disc}$ by $i_S(\{\Gamma_{i}\})$.

Let $\Gamma_\infty$ be the inductive limit of a family of graphs $\{\Gamma_i\}$. Furthermore, assume that, the set $\breve S$ is chosen such that  
\begin{enumerate}
 \item for each graph of the family the surface set $\breve S$ has the same surface intersection property,
 \item the inductive structure preserves the same surface intersection property\footnote{In particular, a graph $\Gp$, which has the same intersection surface property for $\breve S$, has the same intersection behavior for $\breve S$ if $\Gp$ is considered as a subgraph of a graph $\Gamma$, too.} for $\breve S$ and
 \item each surface in $\breve S$ intersects the inductive limit $\Gamma_\infty$ a finite or an infinite number of vertices.
\end{enumerate}

Then $E_{S_{\disc}}(\Gamma)^+E_{S_{\disc}}(\Gamma)$ denote the (Lie algebra-valued) \textbf{discretised quantum flux operator associated a surface $S_{\disc}$ and a graph $\Gamma$} such that $S_{\disc}\cap\Gamma$ is a subset of the set of vertices $V_\Gamma$ and $E_{S_{\disc}}\in\go_{\breve S_{\disc},\Gamma}$. 

The (Lie algebra-valued) \textbf{discretised and localised quantum flux operator $\tilde E_{S_{\disc}}(\Gamma_{i+1})^+\tilde E_{S_{\disc}}(\Gamma_{i+1})$ associated a surface $S_{\disc}$ and an inductive family of graphs $\{\Gamma_i\}_{i=1,..\infty}$} is defined by the difference operator \[\tilde E_{S_{\disc}}(\Gamma_{i+1})^+\tilde E_{S_{\disc}}(\Gamma_{i+1}):= E_{S_{\disc}}(\Gamma_{i+1})^+E_{S_{\disc}}(\Gamma_{i+1})-E_{S_{\disc}}(\Gamma_i)^+E_{S_{\disc}}(\Gamma_i)\] for $E_{S_{\disc}}(\Gamma_{i})\in\bar\go_{\breve S_{\disc},\Gamma_{i}}$ and $E_{S_{\disc}}(\Gamma_{i+1})\in\bar\go_{\breve S_{\disc},\Gamma_{i+1}}$ such that
\begin{enumerate}
 \item  $\tilde E_{S_{\disc}}(\Gamma_{i+1})^+\tilde E_{S_{\disc}}(\Gamma_{i+1})$ is non-trivial only for intersections of the surfaces in $\breve S$ and the graph $\Gamma_{i+1}$ in vertices contained in the set $i_S(\{\Gamma_{i+1}\})\setminus i_S(\{\Gamma_{i}\})$ and
 \item $(\tilde E_{S_{\disc}}(\Gamma_{i+1})^+\tilde E_{S_{\disc}}(\Gamma_{i+1}))^+=\tilde E_{S_{\disc}}(\Gamma_{i+1})^+\tilde E_{S_{\disc}}(\Gamma_{i+1})$ yields.
\end{enumerate}

The set of such discretised and localised quantum flux operator $\tilde E_{S_{\disc}}(\Gamma)^+\tilde E_{S_{\disc}}(\Gamma)$ associated a graph $\Gamma$ is denoted by $\bar\go^{\loc}_{\breve S_{\disc},\Gamma}$ and called the \textbf{localised Lie flux algebra associated a discretised surface set and a graph}. 
The set of such discretised and localised quantum flux operator $\tilde E_{S_{\disc}}(\Gamma_{i+1})^+\tilde E_{S_{\disc}}(\Gamma_{i+1})$ associated an inductive family of graphs $\{\Gamma_i\}_{i=1,..\infty}$ is denoted by $\bar\go^{\loc}_{\breve S_{\disc}}$ and called the \textbf{localised Lie flux algebra associated to a discretised surface set} (and an inductive family of graphs). 

The \textbf{localised enveloping flux algebra $\E^{\loc}_{\breve S_{\disc},\Gamma}$ associated to a discretised surface set and a graph} is given by the enveloping algebra of the localised Lie flux algebra associated to a discretised surface set $\breve S_{\disc}$ and the graph $\Gamma$.

Finally, the \textbf{localised enveloping flux algebra $\E^{\loc}_{\breve S_{\disc}}$ associated to a discretised surface set} (and an inductive family of graphs) is given by the enveloping algebra of the localised Lie flux algebra associated to a discretised surface set $\breve S_{\disc}$.
\end{defi}
If the situation of all paths are ingoing w.r.t all surfaces in a set $\breve S$, then the localised Lie flux algebra (resp. localised enveloping flux algebra) associated to a discretised surface set $\breve S_{\disc}$ associated to $\breve S$ and an inductive family of graphs is denoted by $\bar\go_{\loc}^{\breve S_{\disc}}$ (resp. $\bar\E_{\loc}^{\breve S_{\disc}}$). 

Finally assume that, $\bar G_{\breve S_{\disc},\Gamma}$ (resp. $\bar G_{\breve S_{\disc}}$) denotes the Lie flux group associated to Lie flux algebra $\bar\go_{\breve S_{\disc},\Gamma}^{\loc}$ (resp. $\bar\go_{\breve S_{\disc}}^{\loc}$).
\section{The localised holonomy-flux cross-product $^*$-algebra}\label{subsec restrdiffeo}
\subsection{The localised holonomy $^*$-algebra}\label{loc}
\paragraph*{The construction of the localised holonomy $C^*$-algebra\\[15pt]}

The construction of the new algebra of quantum configuration variables combines a lot of the structures, which have been presented in \cite{Kaminski0,Kaminski1,Kaminski2,Kaminski3,KaminskiPHD}. In particular the idea for the definition of the holonomy-flux cross-product $^*$-algebra \cite[Sec.: 3.1]{Kaminski3},\cite[Sec.: 8.4]{KaminskiPHD} is used.

Assume that $G$ is a compact connected Lie group, $\Gamma$ be a graph, $\breve S$ a surface set and $\breve S_{\disc}$ a discretised surface set (associated to $\breve S$). 

The convolution holonomy $^*$-algebra associated to $\Gamma$ is denoted by $\CD(\Ab_{\Gamma}^{\disc})$ (resp. $\CD(\Ab^{\Gamma}_{\disc})$). This algebra is completed with respect to an appropriate norm to a $C^*$-algebra, which is called the \textbf{non-commutative holonomy $C^*$-algebra} $C^*(\Ab_{\Gamma}^{\disc})$ (resp. $C^*(\Ab^{\Gamma}_{\disc})$) associated to a graph $\Gamma$. 
Moreover, the $C^*$-algebra $C^*(\Ab_{\Gamma}^{\disc})$ (resp. $C^*(\Ab^{\Gamma}_{\disc})$) is isomorphic to a infinite matrix $C^*$-algebra $M_{\Gamma}(\CB)$. The \textbf{analytic holonomy $C^*$-algebra} associated to the graph $\bar\Gamma$ is denoted by $C(\Ab_{\bar\Gamma})$. Note that, the graph $\bar\Gamma$ is defined such that there are no intersections with elements of $\breve S_{\disc}$.
Now new $C^*$-algebras are constructed from $C^*$-tensor product algebras.

\begin{defi}Let $\Gamma$ be a graph, $\breve S$ a surface set and $\breve S_{\disc}$ a discretised surface set (associated to $\breve S$). Then denote the subgraph of $\Gamma$ such that, this graph contains all edges of the graph $\Gamma$ that do not intersect with any vertex of the discretised surface set $\breve S_{\disc}$, by $\bar\Gamma$.

Define
\beqs&C^*(\Ab_{\disc,\Gamma})
:=C^*(\Ab^{\disc}_{\Gamma})\otimes C^*(\Ab_{\disc}^{\Gamma})\text{ where }\\
&C^*(\Ab^{\disc}_{\Gamma}):=\bigotimes_{i\in I}\bigotimes_{{k=1,...,N_k^i}}C^*(\Ab_{\disc, \gamma_{i,1}\circ...\circ\gamma_{i,k}})\eqs

The \textbf{localised holonomy $C^*$-algebra associated to a graph and a discretised surface set} is given by the tensor product
$C^*(\Ab_{\disc,\Gamma})\otimes C(\Ab_{\bar\Gamma})$ (with respect to the minimal $C^*$-norm).
\end{defi}

In this article only certain graphs are studied. These graphs are assumed to decompose into two sets of graphs: one set contains disconnected graphs that contains only paths such that either the source or target vertex is an element of each surface $S_{\disc}$ in $\breve S_{\disc}$, and the other set of disconnected graphs contains graphs $\bar\Gamma_i$ that contains paths, which does not intersect any point of each discretised surface set $S_{\disc}$ in $\breve S_{\disc}$. Hence this property generalises to set of graphs. In particular such a decomposition exists for an inductive family of graphs. 

\begin{defi}Let $\{\Gamma_i\}$ be an inductive family of graphs, which contain only paths such that either the source or target vertex is an element of each surface $S_{\disc}$ in $\breve S_{\disc}$. Moreover, let $\{\bar\Gamma_i\}$ be inductive family $\{\bar\Gamma_i\}$ of graphs that contains no paths, which start or end in a vertex contained in any set of the discretised surface set $\breve S_{\disc}$.

There is a increasing family of matrix algebras $\{C^*(\Ab^{\disc} _{\Gamma_i}),\beta^{\disc}_{\Gamma_i,\Gamma_{i+1}}\}_{i=1,..,\infty}$ with $\beta^{\disc}_{\Gamma_i,\Gamma_{i+1}}$ unit-preserving $^*$-homomorphisms such that the union of all matrix algebras is a normed $^*$-algebra, which can be completed by the minimal tensor product norm to a $C^*$-algebra
\[C^*(\Ab^{\disc}):=\bigcup_{m=1,...,\infty} C^*(\Ab^{\disc}_{\Gamma_m})\]

There is a increasing family of matrix algebras $\{C^*(\Ab_{\disc, \Gamma_i}),\beta_{\disc,\Gamma_i,\Gamma_{i+1}}\}_{i=1,..,\infty}$ with $\beta_{\disc,\Gamma_i,\Gamma_{i+1}}$ unit-preserving $^*$-homomorphisms such that the union of all matrix algebras is a normed $^*$-algebra, which can be completed by the minimal tensor product norm to a $C^*$-algebra
\[C^*(\Ab_{\disc}^{\disc}):=\bigcup_{m=1,...,\infty} C^*(\Ab_{\disc,\Gamma_m})\]

Furthermore, there is an inductive limit $C^*$-algebra $C(\Ab_{\loc})$, which is constructed from an inductive family $\{C(\Ab_{\bar\Gamma_i}),\beta_{\bar\Gamma_i,\bar\Gamma_{i+1}})\}_{i=1,...,\infty} $ of $C^*$-algebras. 
\end{defi}

In particular, an element of $C^*(\Ab_{\disc,\Gamma})$ is for example given by \beqs f^1_{\Gamma}(\ho_{\Gamma}(\gamma_{1,1}),...,\ho_{\Gamma}(\gamma_{N,1}))\otimes f^2_{\Gamma}(\ho_{\Gamma}(\gamma_{1,1}\circ\gamma_{1,2}),...,\ho_{\Gamma}(\gamma_{N,1}\circ\gamma_{N,2}))\eqs

Notice that, for $\Gamma_i\cap\Gamma_j=\varnothing$ for $i\neq j$ then $C^*(\Ab_{\disc,\Gamma_i\cup\Gamma_j})=C^*(\Ab_{\disc,\Gamma_i})\otimes C^*(\Ab_{\disc,\Gamma_j})$, $C^*(\Ab_{\disc,\Gamma_i})$ is isomorphic to the $C^*$-subalgebra $C^*(\Ab_{\disc,\Gamma_i})\otimes\idf_{\Gamma_j}$ of $C^*(\Ab_{\disc,\Gamma_i})\otimes C^*(\Ab_{\disc,\Gamma_j})$ where $\idf_{\Gamma_j}$ is the identity operator in $C^*(\Ab_{\disc,\Gamma_j})$. 

\begin{defi}
The \textbf{localised holonomy $C^*$-algebra} is the $C^*$-tensor product algebra $C(\Ab_{\loc})\otimes C^*(\Ab_{\disc}^{\disc})$ (with respect to the minimal $C^*$-norm) associated to a discretised set of surfaces.
 \end{defi}
In this definition the notion of localisation is emphasized, since the elements of this new algebra really depend on a chosen discretised surface set associated to a surface set.

\paragraph*{Actions of the group of bisections on the localised holonomy $C^*$-algebra associated to a graph and a discretised surface set\\[15pt]}

In this paragraph graph changing operations are studied. First observe that, there are some certain bisections, which map target vertices of certain paths to suitable paths. The set of these bisections in a finite graph system $\PD_\Gamma^{\breve S_{\disc}}$ has been introduced at the end of section \ref{subsubsec bisections} and is denoted by $\mathfrak{B}(\PD^{\breve S_{\disc}}_\Gamma)$. 
Only these bisections restricted to a set $V^{\breve S_{\disc}}$ are used to define an action of bisections on the localised analytic holonomy $C^*$-algebra associated to a graph and a discretised surface set. The action is for example given by
\beqs (\zeta_\sigma f_\Gamma)(\ho_\Gamma(\Gp))= f_{\Gamma}(\ho_{\Gamma}(\Gp_\sigma))\text{ for }\sigma\in\mathfrak{B}(\PD^{\breve S_{\disc}}_\Gamma)
\eqs whenever 
\begin{itemize}
 \item $f_{\Gamma}\in C^*(\Ab_{\disc,\Gamma})$, 
 \item $\Gp:=\{\gp_i\},\Gp_{\sigma}:=\{\gp_i\circ\sigma(t(\gp_i))\}$ are subgraphs of $\Gamma$.
\end{itemize}
In general the action is defined as follows.
\begin{lem}
There is an action $\alpha$ of the group $\mathfrak{B}(\PD^{\breve S_{\disc}}_\Gamma)$ of bisections on the $C^*$-algebra $C^*(\Ab^{\disc}_{\Gamma})$, which is defined by 
\beqs \zeta_\sigma( f_\Gamma):= f_\Gamma\circ R_\sigma
\eqs whenever $f_\Gamma\in C^*(\Ab^{\disc}_{\Gamma})$.
\end{lem}
\begin{proofs}
Let $\sigma\in\mathfrak{B}(\PD^{\breve S_{\disc}}_\Gamma)$ then $\sigma\mapsto \zeta_\sigma$ is a group homomorphism and
\beqs &(\zeta_{\sigma_1}\circ \zeta_{\sigma_2})(f_\Gamma)=\zeta_{\sigma_1\ast_2\sigma_2}(f_\Gamma)\\
&\zeta_{\sigma}(f_\Gamma^*)=\zeta_{\sigma}(f_\Gamma)^*
\eqs yields sfor all $\sigma,\sigma_1,\sigma_2\in \mathfrak{B}(\PD^{\breve S_{\disc}}_\Gamma)$ and $f_\Gamma\in C^* (\Ab^{\disc}_{\Gamma})$.
\end{proofs}

Now focus paths, which do not have any intersection with a discretised surface in $\breve S_{\disc}$. Then there is an action of $\Diff(\PD_{\bar\Gamma})$ and hence $\mathfrak{B}(\PD_{\bar\Gamma})$ on $C(\Ab_{\bar\Gamma})$. This action is a point-norm continuous automorphic action of $\Diff(\PD_{\bar\Gamma})$ on $C(\Ab_{\bar\Gamma})$ for every graph $\bar\Gamma$ of the inductive family $\{\bar\Gamma_i\}$ of graphs.

\paragraph*{Derivations defined by the discretised and localised flux operator for surfaces and graphs\\[5pt]}\label{der}

In section \ref{sec fluxdef} the discrete and localised flux operator $\tilde E_{S_{\disc}}(\Gamma_{i+1})^+\tilde E_{S_{\disc}}(\Gamma_{i+1})$ has been introduced in definition \ref{defi locfluxop}. The definition of this operator is chosen such that this operator acts non-trivial on elements of $C^*(\Ab^{\disc}_{\Gamma_{i+1}})$ and commute with all elements contained in $C^*(\Ab^{\disc}_{ \Gamma_i })$. 
\begin{defi}\label{defi derivonmatrix}
Define the derivation $\tilde\delta_{S_{\disc},\Gamma_j}$ on $C^*(\Ab^{\disc}_{\Gamma_{j}})$ with domain $\DD(\tilde \delta_{S_{\disc},\Gamma_{i+1}})$ by the following commutator 
\beqs\tilde \delta_{S_{\disc},\Gamma_{i+1}}& (f_{\Gamma_{i+1}})
&:= \bra \tilde E_{S_{\disc}}(\Gamma_{i+1})^+\tilde E_{S_{\disc}}(\Gamma_{i+1}) ,f_{\Gamma_{i+1}}\ket\\
\eqs for a fixed $\tilde E_{S_{\disc}}(\Gamma_{i+1})^+\tilde E_{S_{\disc}}(\Gamma_{i+1})\in\bar\go^{\loc}_{\breve S_{\disc},\Gamma_i}$ and $f_{\Gamma_{i+1}}\in \DD(\tilde \delta_{S_{\disc},\Gamma_{i+1}})$. 
\end{defi}
The domain $\DD(\tilde \delta_{S_{\disc},\Gamma_{i+1}})$ is a $^*$-subalgebra of $C^*(\Ab^{\disc}_{\Gamma_{j}})$.
\begin{lem}
The linear operator $\tilde \delta_{S_{\disc},\Gamma_{i}}$ is a symmetric unbounded $^*$-derivation with the domain $\DD(\tilde \delta_{S_{\disc},\Gamma_{i}})$ of the unital $C^*$-algebra $C^*(\Ab^{\disc}_{\Gamma_{i}})$. The domain $\DD(\tilde \delta_{S_{\disc},\Gamma_{i}})$ is a dense $^*$-subalgebra of $C^*(\Ab^{\disc}_{\Gamma_i})$.
\end{lem}
\begin{proofs}
To show that, the domain $\DD(\tilde \delta_{S_{\disc},\Gamma_{i}})$ is a dense $^*$-subalgebra of $C^*(\Ab^{\disc}_{\Gamma_i})$ recognize that, $\DD(\tilde \delta_{S_{\disc},\Gamma_{i}}):= C^\infty(\Ab^{\disc}_{\Gamma_i})$ is indeed dense in $C^*(\Ab^{\disc}_{\Gamma_i})$.
\end{proofs}

\begin{cor}\label{cor commdiscrflux} The limit
\beqs \tilde\delta_{S_{\disc}}(f)&:
=i \lim_{j\rightarrow\infty}\bra\tilde E_{S_{\disc}}(\Gamma_{j+1})^+\tilde E_{S_{\disc}}(\Gamma_{j+1}),f\ket
\eqs for every $f\in \DD(\tilde \delta_{S_{\disc}})$ and an element $\tilde E_{S_{\disc}}(\Gamma_{j+1})^+\tilde E_{S_{\disc}}(\Gamma_{j+1})\in\bar\go^{\loc}_{\breve S_{\disc},\Gamma_{j+1}}$ for every $j$, is well-defined in the norm topology. The domain is given by
\beqs \DD(\tilde \delta_{S_{\disc}})=\bigcup_{j=1,...,\infty}\DD(\tilde \delta_{S_{\disc},\Gamma_{j}})
\eqs
\end{cor}
\begin{proofs}
Note that,
\beqs \bra\tilde E_{S_{\disc}}(\Gamma_{j+1})^+\tilde E_{S_{\disc}}(\Gamma_{j+1}),f_{\Gamma_k} \ket =0
\eqs yields whenever $\PD_{\Gamma_k}\leq \PD_{\Gamma_{j+1}}$ and $0\leq k\leq j$ and $f_{\Gamma_k}\in C^*(\Ab^{\disc}_{\Gamma_k})$. Consequently, derive
\beqs\tilde\delta_{S_{\disc}}(f)&:
=i \lim_{j\rightarrow\infty}\bra\tilde E_{S_{\disc}}(\Gamma_{j+1})^+\tilde E_{S_{\disc}}(\Gamma_{j+1}),f\ket\\
&=i \lim_{j\rightarrow\infty}\bra\tilde E_{S_{\disc}}(\Gamma_{j+1})^+\tilde E_{S_{\disc}}(\Gamma_{j+1}),f_{\Gamma_0}\otimes ...\otimes f_{\Gamma_j}\otimes f_{\Gamma_{j+1}} \ket\\
&=i \lim_{j\rightarrow\infty}\bra\tilde E_{S_{\disc}}(\Gamma_{j+1})^+\tilde E_{S_{\disc}}(\Gamma_{j+1}), f_{\Gamma_{j+1}\setminus\Gamma_j} \ket \\
=0
\eqs
\end{proofs}

Redefine the symmetric unbounded $^*$-derivation for the discretised flux operator $E_{S_{\disc}}(\Gamma_i)$ for a graph $\Gamma_i$, which is given by 
\beq \delta_{S_{\disc},\Gamma_{j}}(f)=\bra E_{S_{\disc}}(\Gamma_j)^+E_{S_{\disc}}(\Gamma_j),f\ket
\eq whenever $f\in \DD(\delta_{S_{\disc}})$ and for a fixed $E_{S_{\disc}}\in\go_{\breve S_{\disc},\Gamma_j}$ . 

In contrast to the property of the $^*$-derivation of the $C^*$-algebra $C(\Ab)$ presented in \cite[Prop.: 4.7]{Kaminski3},\cite[Prop.: 8.2.19]{KaminskiPHD}, the $^*$-derivation of $C^*(\Ab^{\disc})$ exists under weaker conditions for the surface set and the directed family of graphs. In the previous construction the set $\breve S$ of surfaces has to be chosen such that for each graph of the inductive family of graphs $\{\Gamma_i\}$ there is only a finite number of intersection vertices with each surface of the set $\breve S$.

\begin{prop}\label{defi derivonmatrix2} Let $\breve S_{\disc}$ be an arbitrary discretised surface set and $\{\Gamma_i\}_{i=1,...,\infty}$ be an inductive family of graphs.

Then the limit 
\beq \delta_{S_{\disc}}(f)&:= i \lim_{j\rightarrow\infty}\delta_{S_{\disc},\Gamma_{j+1}}(f) 
\eq whenever $f\in  \DD(\delta_{S_{\disc}})$ exists in norm. The domain of the limit is given by
\beqs \DD(\delta_{S_{\disc}})=\bigcup_{j=1,...,\infty}\DD(\delta_{S_{\disc},\Gamma_{j}})
\eqs and $\DD(\delta_{S_{\disc}})$ is a $^*$-subalgebra of $C^*(\Ab^{\disc})$.
\end{prop}
\begin{proofs}Derive
\beq \delta_{S_{\disc}}(f)&:= i \lim_{j\rightarrow\infty}\delta_{S_{\disc},\Gamma_{j+1}}(f) 
= i \lim_{j\rightarrow\infty}\big(\tilde\delta_{S_{\disc},\Gamma_{j+1}}(f) + \delta_{S,\Gamma_{j}}(f)\big)\\
&=i \lim_{j\rightarrow\infty}\bra\tilde E_{S_{\disc}}(\Gamma_{j+1})^+\tilde E_{S_{\disc}}(\Gamma_{j+1}),f\ket
+ i \lim_{j\rightarrow\infty}\bra E_{S_{\disc}}(\Gamma_{j})^+ E_{S_{\disc}}(\Gamma_{j}),f\ket\\
&=i \lim_{j\rightarrow\infty}\bra\tilde E_{S_{\disc}}(\Gamma_{j+1})^+\tilde E_{S_{\disc}}(\Gamma_{j+1}),f\ket+
... + i \lim_{j\rightarrow\infty}\bra\tilde E_{S_{\disc}}(\Gamma_{1})^+\tilde E_{S_{\disc}}(\Gamma_{1}),f\ket\\
&\qquad+ i \lim_{j\rightarrow\infty}\bra E_{S_{\disc}}(\Gamma_{0})^+ E_{S_{\disc}}(\Gamma_{0}),f\ket\\
\eq by using corollary \ref{cor commdiscrflux}.
\end{proofs}

\subsection{The general localised part of the localised holonomy-flux cross-product $^*$-algebra}\label{con}
\paragraph*{The construction of the general localised part of the localised holonomy-flux cross-product $^*$-algebra\\[15pt]}

Recall the concept of abstract cross-product algebras, which has been presented by Schm\"udgen and Klimyk \cite{KlimSchmued94}. This concept has been used in \cite[Def.: 3.10]{Kaminski3},\cite[Sec.: 8.2]{KaminskiPHD} for the definition of the holonomy-flux cross-product $^*$-algebra associated to a surface set. In analogy a similar cross-product $^*$-algebra can be defined as follows.

In the following considerations the $^*$-algebra $C^*(\Ab^{\disc}_\Gamma)$ (resp. $C^*(\Ab_{\disc}^\Gamma)$ and $C^*(\Ab_{\disc,\Gamma})$) has to be restricted to functions in $C^\infty(\Ab^{\disc}_\Gamma)$ (resp. $C^\infty(\Ab_{\disc}^\Gamma)$ and $C^\infty(\Ab_{\disc,\Gamma})$). The resulting $^*$-subalgebra is denoted by $\Cinf^*(\Ab^{\disc}_\Gamma)$ (resp. $\Cinf^*(\Ab_{\disc}^\Gamma)$ and $\Cinf^*(\Ab_{\disc,\Gamma})$) and is called the \textbf{localised analytic holonomy $^*$-algebra }associated to a graph $\Gamma$ and a discretised surface set $\breve S_{\disc}$ again. 

To start with a right-invariant flux vector field is defined as follows.
For simplicity, the investigations start with a graph $\Gamma$, which contains only one path $\gamma$, and one discretised surface $S_{\disc}$. Clearly, the following definition generalises to graphs and a suitable discretised surface set $\breve S_{\disc}$. 
\begin{defi}
Let the graph $\Gamma$ contain only a path $\gamma$ and $S_{\disc}$ be a discrete surface associated to a surface $S$ such that the path lies below and outgoing w.r.t. the surface orientation of $S$. Set $\tilde E_{S_{\disc}}(\Gamma)^+\tilde E_{S_{\disc}}(\Gamma)=:X_{S_{\disc}}^+X_{S_{\disc}}$. Then the \textbf{right-invariant flux vector field} $e^{\overrightarrow{L}}$ is defined by
\beqs\bra \tilde E_{S_{\disc}}(\Gamma)^+\tilde E_{S_{\disc}}(\Gamma),f_\Gamma\ket:=e^{\overrightarrow{L}}(f_\Gamma)\eqs where
\beq\label{CommRel1} e^{\overrightarrow{L}}(f_\Gamma)(\ho_\Gamma(\gamma))=\frac{\dif}{\dif t}\Big\vert_{t=0} f_\Gamma(\exp(t X_{S_{\disc}}^+X_{S_{\disc}}\ho_\Gamma(\gamma))&\text{ for }X_{S_{\disc}}\in\go, \ho_\Gamma(\gamma)\in G, t\in\R
\eq whenever $f_\Gamma\in \Cinf^*(\Ab^{\disc}_\Gamma)$ and $\tilde E_{S_{\disc}}(\Gamma)\in\bar \go_{S_{\disc},\Gamma}$. Set 
\beq\label{CommRel2} e^{\overleftarrow{L}}(f_\Gamma)(\ho_\Gamma(\gamma))=\frac{\dif}{\dif t}\Big\vert_{t=0} f_\Gamma(\exp(-tX_{S_{\disc}}^+X_{S_{\disc}})\ho_\Gamma(\gamma)),\text{ for }X_{S_{\disc}}\in\go,\ho_\Gamma(\gamma)\in G, t\in\R
\eq such that
\beqs \bra \tilde E_{S_{\disc}}(\Gamma)\tilde E_{S_{\disc}}(\Gamma)^+,f_\Gamma\ket &= e^{\overleftarrow{L}}(f_\Gamma)
\eqs 
\end{defi}

The definition of a right invariant vector field is needed to study the following well-defined structure.

\begin{lem}Let $\breve S_{\disc}$ be a set of discretised surfaces, which is constructed from a set $\breve S$ of surfaces, and let $\Gamma$ be a graph.
Let $\{\Gamma_i\}_{i=1,..\infty}$ be an inductive family of graphs such that for every graph $\Gamma_i$ the intersection points of a surface set $\breve S_{\disc}$ and the graph $\Gamma_i$ are vertices of $V_{\Gamma_i}$. Denote the set of intersections of a graph $\Gamma_i$ and a discretised surface set $\breve S_{\disc}$ by $i_S(\{\Gamma_{i}\})$.

Let $\Gamma_\infty$ be the inductive limit of a family of graphs $\{\Gamma_i\}$. Furthermore assume that, the set $\breve S$ is chosen such that  
\begin{enumerate}
 \item for each graph of the family the surface set $\breve S$ has the same surface intersection property,
 \item the inductive structure preserves the same surface intersection property for $\breve S$ and
 \item each surface in $\breve S$ intersects the inductive limit $\Gamma_\infty$ a finite or an infinite number of vertices.
\end{enumerate}

Then $\Cinf^*(\Ab^{\disc})$ is a left $\bar\E_{\breve S_{\disc}}^{\loc}$-module algebra. The action of $\bar\E_{\breve S_{\disc}}^{\loc}$ on  $\Cinf^*(\Ab^{\disc})$ is given by\\ $\tilde E_{S_{\disc}}(\Gamma)^+\tilde E_{S_{\disc}}(\Gamma)\rhd f_\Gamma:= e^{\overrightarrow{L}}(f_\Gamma)$ whenever $E_{S_{\disc}}(\Gamma)^+\tilde E_{S_{\disc}}(\Gamma)\in\bar\E_{\breve S_{\disc}}^{\loc}$ and $f_\Gamma\in\Cinf^*(\Ab^{\disc})$.
\end{lem}
In analogy $\Cinf^*(\Ab_{\disc}^\Gamma)$ is a right $\bar\E^{\breve S_{\disc},\Gamma}_{\loc}$-module algebra and is defined by right invariant vector fields. Now a construction of a cross-product $^*$-algebra is given as follows.

\begin{defi}\label{defi holfluxcross} Let $\breve S_{\disc}$ be a set of discretised surfaces associated to a surface set $\breve S$, which has appropriate poperties with respect to a graph $\Gamma$ and an inductive family $\{\Gamma_i\}_{i=1,...,\infty}$ of graphs.

The \textbf{general localised part of the localised holonomy-flux cross-product $^*$-algebra for a graph $\Gamma$ and a discretised surface set $\breve S_{\disc}$} is given by the left (or right) cross-product $^*$-algebra
\[\Cinf (\Ab_\Gamma^{\disc})\rtimes_{L}\bar\E^{\loc}_{\breve S_{\disc},\Gamma} (\text{ or } \Cinf (\Ab^\Gamma_{\disc})\rtimes_{R}\bar\E_{\loc}^{\breve S_{\disc},\Gamma} )\] which are defined by the vector space $\Cinf (\Ab_\Gamma^{\disc})\otimes\bar\E^{\loc}_{\breve S_{\disc},\Gamma}$ with the multiplication given by
\beqs (f^1_\Gamma\otimes E^{\disc}_{S_1}(\Gamma))\cdot_L(f^2_\Gamma\otimes E^{\disc}_{S_2}(\Gamma))
:=f_\Gamma^1(E^{\disc}_{S_1}(\Gamma)\rhd f^2_\Gamma)\otimes E^{\disc}_{S_2}(\Gamma) 
+ f_\Gamma^1 f^2_\Gamma \otimes E^{\disc}_{S_1}(\Gamma)\cdot E^{\disc}_{S_2}(\Gamma)
\eqs 
and the involution
\beqs (f_\Gamma\rhd E^{\disc}_S(\Gamma))^*:=\bar f_\Gamma\rhd E^{\disc}_S(\Gamma)^+
\eqs
whenever $E^{\disc}_{S_1}(\Gamma),E^{\disc}_{S_2}(\Gamma),E^{\disc}_{S}(\Gamma)\in\bar\E^{\loc}_{\breve S_{\disc},\Gamma}$ and $f^1_\Gamma,f^2_\Gamma,f_\Gamma\in \Cinf(\Ab^{\disc}_\Gamma)$.

The \textbf{general localised part of the localised holonomy-flux cross-product $^*$-algebra associated to a discretised surface set $\breve S_{\disc}$} is given by the left (or right) cross-product $^*$-algebra
\[\Cinf(\Ab^{\disc})\rtimes_{L} \bar \E^{\loc}_{\breve S_{\disc}}(\text{ or }\Cinf (\Ab_{\disc})\rtimes_{R} \bar \E_{\loc}^{\breve S_{\disc}})\]
which are the inductive limit of the families $\{(\Cinf (\Ab^{\disc}_\Gamma)\rtimes_{L}\bar\E^{\loc}_{\breve S_{\disc},\Gamma},\beta^{\disc}_{\Gamma,\Gp}\times\check\beta^{\disc}_{\Gamma,\Gp})\}$ (or $\{(\Cinf (\Ab_{\disc}^\Gamma)\rtimes_{R}\bar\E_{\loc}^{\breve S_{\disc},\Gamma},\beta_{\disc}^{\Gamma,\Gp}\times\check\beta_{\disc}^{\Gamma,\Gp})\}$) where $\check\beta_{\Gamma,\Gp}: \bar\E^{\loc}_{\breve S_{\disc},\Gamma}\rightarrow \bar\E^{\loc}_{\breve S_{\disc},\Gp}$ are suitable unit-preserving $^*$-homomorphisms for a suitable set $\breve S_{\disc}$ of discretised surfaces that preserve the left (or right) vector field structure. 
\end{defi}

Summarising the general localised part of the localised holonomy-flux cross-product $^*$-algebra is a certain cross-product $^*$-algebra, which is defined by the localised analytic holonomy $^*$-algebra and the localised enveloping flux algebra associated to a discretised surface set.

\paragraph*{A representation of the general localised part of the localised holonomy-flux cross-product $^*$-algebra\\[15pt]}\label{rep}

In \cite[Sec.: 4]{Kaminski3},\cite[Sec.: 8.2.3]{KaminskiPHD} a certain $^*$-representation of a Lie algebra has been studied. This $^*$-representation is called the infinitesimal representation of a Lie algebra on a Hilbert space. Similarly the representation of the general localised part of the localised holonomy-flux cross-product $^*$-algebra is presented as follows.

First in this article the $^*$-representation of the Lie flux algebra $\bar\go_{\breve S_{\disc},\Gamma}$ is implemented by the infinitesimal representation $\dif u$ on the Hilbert space $\HS^{\disc}_\Gamma$, which is given by $L^2(\Ab^{\disc}_{\Gamma}, \mu^{\disc}_{\Gamma})$. Notice that, the configuration space $\Ab^{\disc}_{\Gamma}$ is equivalent to $G^M$ for a suitable $M\in \N$ and $\mu^{\disc}_{\Gamma}$ is the corresponding Haar measure on $G^M$. The infinitesimal representation corresponds to the unitary representation $u$ of the Lie flux group $\bar G_{\breve S_{\disc},\Gamma}$ in the $C^*$-algebra $C^*(\Ab^{\disc}_{\Gamma})$, i.o.w. $u\in\Rep(\bar G_{\breve S_{\disc},\Gamma}, C^*(\Ab^{\disc}_{\Gamma}))$.
The $^*$-representation of the general localised part of the localised holonomy-flux cross-product $^*$-algebra is derived from this $^*$-representation.

Summarising the next definition the $^*$-representations of the following algebras:
\begin{itemize}
 \item the localised analytic holonomy $^*$-algebra $\Cinf^*(\Ab^{\disc}_{\Gamma})$,
 \item the localised enveloping flux algebra $\bar\E^{\loc}_{\breve S_{\disc},\Gamma}$ and 
\item the general localised part of the localised holonomy-flux cross-product $^*$-algebra $\Cinf^*(\Ab^{\disc}_{\Gamma})\rtimes_L\bar\E^{\loc}_{\breve S_{\disc},\Gamma}$
\end{itemize} on the Hilbert space $\HS^{\disc}_{\Gamma}$ are presented.

\begin{defi}\label{defi represlocholfluxcross}
The $^*$-representation of $\Cinf^*(\Ab^{\disc}_{\Gamma})$ is defined by
\beqs 
&\Phi_M(f_\Gamma) \psi_\Gamma = f_\Gamma\psi_\Gamma\text{ for } f_\Gamma\in \Cinf^*(\Ab^{\disc}_{\Gamma})\text{ and }\psi_\Gamma\in \HS^{\disc}_{\Gamma} \\
&\Phi_M(f^*_\Gamma) \psi_\Gamma = f_\Gamma^*\psi_\Gamma\text{ for } f_\Gamma\in \Cinf^*(\Ab^{\disc}_{\Gamma}) \text{ and }\psi_\Gamma\in \HS^{\disc}_{\Gamma}
\eqs

There exists a positive self-adjoint operator $\dif u (E_{S_{\disc}}(\Gamma)^+E_{S_{\disc}}(\Gamma))$ or, respectively,\\ the adjoint operator $\dif u (E_{S_{\disc}}(\Gamma)E_{S_{\disc}}(\Gamma)^+)$ on the Hilbert space $\HS^{\disc}_{\Gamma}$ defined by
\beqs
&\dif u (E_{S_{\disc}}(\Gamma)^+E_{S_{\disc}}(\Gamma))\psi_\Gamma :=i \bra E_{S_{\disc}}(\Gamma)^+E_{S_{\disc}}(\Gamma), \psi_\Gamma \ket\\
&\qquad\text{ for a fixed } E_{S_{\disc}}(\Gamma)^+E_{S_{\disc}}(\Gamma)\in \bar\E^{\loc}_{\breve S_{\disc},\Gamma}\text{ and }\psi_\Gamma\in \DD(\dif u(E_{S_{\disc}}(\Gamma)^+E_{S_{\disc}}(\Gamma))) \\
&\dif u(E_{S_{\disc}}(\Gamma)E_{S_{\disc}}(\Gamma)^+)\psi_\Gamma :=-i \bra E_{S_{\disc}}(\Gamma)^+E_{S_{\disc}}(\Gamma), \psi_\Gamma \ket\\
&\qquad\text{ for a fixed } E_{S_{\disc}}(\Gamma)^+E_{S_{\disc}}(\Gamma)\in \bar\E^{\loc}_{\breve S_{\disc},\Gamma}\text{ and }\psi_\Gamma\in \DD(\dif u(E_{S_{\disc}}(\Gamma)^+E_{S_{\disc}}(\Gamma))) 
\eqs and $u\in\Rep(\bar G_{\breve S_{\disc},\Gamma},\Cinf^*(\Ab^{\disc}_{\Gamma}))$.

The $^*$-representation of the $^*$-algebra $\Cinf^*(\Ab^{\disc}_{\Gamma})\rtimes\bar\E^{\loc}_{\breve S_{\disc},\Gamma}$ on $\HS^{\disc}_{\Gamma}$ is defined by
\beqs
&\hat\pi_\Gamma(f_\Gamma\otimes i E_{S_{\disc}}(\Gamma_{j})^+E_{S_{\disc}}(\Gamma_{j})) \psi_\Gamma := 
\frac{1}{2} i\bra E_{S_{\disc}}(\Gamma_{j})^+E_{S_{\disc}}(\Gamma_{j}), f_\Gamma \ket \psi_\Gamma 
+\frac{1}{2} if_\Gamma\bra E_{S_{\disc}}(\Gamma_{j})^+E_{S_{\disc}}(\Gamma_{j}),\psi_\Gamma\ket \\
&\qquad\text{ for } f_\Gamma\in\Cinf^*(\Ab^{\disc}_{\Gamma})\text{ and for a fixed } E_{S_{\disc}}(\Gamma)^+E_{S_{\disc}}(\Gamma)\in \bar\E^{\loc}_{\breve S_{\disc},\Gamma}\\
&\hat\pi_\Gamma(f_\Gamma\otimes i E_{S_{\disc}}(\Gamma_{j})^+E_{S_{\disc}}(\Gamma_{j})) \psi_\Gamma := 
\frac{1}{2} i\bra E_{S_{\disc}}(\Gamma_{j})^+E_{S_{\disc}}(\Gamma_{j}), f_\Gamma \ket \psi_\Gamma 
+\frac{1}{2} if_\Gamma\bra E_{S_{\disc}}(\Gamma_{j})^+E_{S_{\disc}}(\Gamma_{j}),\psi_\Gamma\ket\\ 
&\qquad\text{ for } f_\Gamma\in\Cinf^*(\Ab^{\disc}_{\Gamma})\text{ and for a fixed } E_{S_{\disc}}(\Gamma)^+E_{S_{\disc}}(\Gamma)\in \bar\E^{\loc}_{\breve S_{\disc},\Gamma} \\
\eqs whenever $\psi_\Gamma\in \DD(\dif u(E_{S_{\disc}}(\Gamma_{j})^+E_{S_{\disc}}(\Gamma_{j})))$.
\end{defi}

\subsection{$C^*$-dynamical systems, KMS-states and the localised holonomy-flux cross-product $^*$-algebra}\label{dyn}

In this section different $C^*$-dynamical systems are constructed from different actions and algebras. The aim is to implement a strongly continuous one-parameter automorphism group such that a modified quantum Hamilton constraint is the generator of this automorphism group. This will be done in several steps. In this section the basic $C^*$-dynamical systems are introduced, which are used in the section \ref{subsec QHamilton} for the analysis of the modified quantum Hamilton constraint. 

First notice the following result. In general, for every $C^*$-algebra $\Alg$ and a point norm-continuous automorphic action $\beta$ of $\R$ on $\Alg$, there is a set $\Alg^\beta$ defined by all element $A\in\Alg$ such that $\beta_t(A)=A$ for every $t\in\R$.
Then the set $\Alg^\beta$ is a norm-dense $^*$-subalgebra of $\Alg$.

Set $\ho_\Gamma(\Gamma)=:\ho_\Gamma\in\Ab^{\disc}_{\Gamma}$.
Let $\LAb_{\disc,\Gamma}$ be the enveloping Lie algebra of the Lie algebra associated to $\Ab^{\disc}_{\Gamma}$. Consider the $C^*$-subalgebra $\ZD(\Ab^{\disc}_{\Gamma})$ of $C^*(\Ab^{\disc}_{\Gamma})$, which is generated by all central functions, i.e. all functions $f_{\Gamma}\in C^*(\Ab^{\disc}_{\Gamma})$ such that 
$f_{\Gamma}(\ho_{\Gamma})=f_{\Gamma}(\go_{\Gamma}^{-1}\ho_{\Gamma}\go_{\Gamma})$ for all $\go_{\Gamma}\in \Ab^{\disc}_{\Gamma}$. 

Finally, consider an action  $\beta_{\ha_{\disc},\Gamma_{i}}$ of $\R$ on $C^*(\Ab^{\disc}_{\Gamma})$ defined by
\beqs(\beta_{\ha_{\disc},\Gamma_{i}}(t) f_{\Gamma_i})(\ho_{\Gamma_i})
:=f_{\Gamma_i}(\exp(-t \ha_{\disc,\Gamma_{i}})\ho_{\Gamma_i}\exp(t \ha_{\disc,\Gamma_{i}}))
\eqs and notice that
\beqs(\beta_{\ha_{\disc},\Gamma_{i}}(t) f_{\Gamma_i})^*(\ho_{\Gamma_i})
=\overline{(\beta_{\ha_{\disc,\Gamma_{i}}^+}(t) f_{\Gamma_i}^*)(\ho_{\Gamma_i}^{-1})}
\eqs for a fixed Lie algebra element $\ha_{\disc,\Gamma_{i}}$ in $\LAb_{\disc,\Gamma}$ yields. Set $C^*(\Ab^{\disc}_{\Gamma})$ be equal to $\Alg$. 
Then $\Alg^{\beta}$ is isomorphic to $\ZD(\Ab^{\disc}_{\Gamma})$.

\begin{prop}\label{prop dynsysbetaholflux}
 The triple $(C^*(\Ab^{\disc}_{\Gamma_i}),\R,\beta_{\ha_{\disc},\Gamma_{i}})$ is a $C^*$-dynamical system.
\end{prop}
This is verifed easily, after the following considerations.

Furthermore, there is an action $\tilde\alpha_{\ha_{\disc,\Gamma_{i}}}$ of $\R$ on $\ZD(\Ab^{\disc}_{\Gamma})$ defined by
\beqs(\tilde\alpha_{\ha_{\disc,\Gamma_{i}}}(t) f_{\Gamma_i})(\ho_{\Gamma_i})
:=f_{\Gamma_i}(\exp(-t \ha_{\disc,\Gamma_{i}})\ho_{\Gamma_i})
\eqs and 
\beqs(\tilde\alpha_{\ha_{\disc,\Gamma_{i}}}(t) f_{\Gamma_i})^*(\ho_{\Gamma_i})
:=\overline{(\tilde\alpha_{\ha_{\disc,\Gamma_{i}}^+}(t) f_{\Gamma_i})(\ho_{\Gamma_i}^{-1})}
\eqs whenever $f_{\Gamma_i}\in \ZD(\Ab^{\disc}_{\Gamma_i})$ and $\ha_{\disc,\Gamma_{i}}\in\LAb_{\disc,\Gamma}$.

\begin{prop}
 The triple $(\ZD(\Ab^{\disc}_{\Gamma_i}),\R,\tilde\alpha_{\ha_{\disc},\Gamma_{i}})$ is a $C^*$-dynamical system.
\end{prop}
\begin{proofs}Derive
 \beqs&(\tilde\alpha_{\ha_{\disc,\Gamma_{i}}}(t_1+t_2) f_{\Gamma_i})(\ho_{\Gamma_i})\\
&=f_{\Gamma_i}(\exp(-(t_1+t_2)\ha_{\disc,\Gamma_{i}})\ho_{\Gamma_i})\\
&=f_{\Gamma_i}(\exp(-t_1\ha_{\disc,\Gamma_{i}})\exp(-t_2\ha_{\disc,\Gamma_{i}})\ho_{\Gamma_i}) \\
&=(\tilde\alpha_{\ha_{\disc,\Gamma_{i}}}(t_1)\circ \tilde\alpha_{\ha_{\disc,\Gamma_{i}}}(t_2)f_{\Gamma_i})(\ho_{\Gamma_i})
\eqs and, since, $f_{\Gamma_i}(\ho_{\Gamma_i})=f_{\Gamma_i}(\go_{\Gamma_i}^{-1}\ho_{\Gamma_i}\go_{\Gamma_i})$ for all $\go_{\Gamma_i}\in \Ab^{\disc}_{\Gamma_i}$ it follows that,
\beqs(\tilde\alpha_{\ha_{\disc,\Gamma_{i}}}(t) f_{\Gamma_i})^*(\ho_{\Gamma_i})
&=\overline{f_{\Gamma_i}((\exp(-t \ha_{\disc,\Gamma_{i}})\ho_{\Gamma_i})^{-1})}\\
&=\overline{f_{\Gamma_i}(\ho_{\Gamma_i}^{-1}(\Gamma_i)\exp(t \ha_{\disc,\Gamma_{i}}))}\\
&=\overline{(\tilde\alpha_{\ha^+_{\disc,\Gamma_{i}}}(t) f_{\Gamma_i})(\ho_{\Gamma_i}^{-1}(\Gamma_i))}\\
&=(\tilde\alpha_{\ha_{\disc,\Gamma_{i}}}(t) f_{\Gamma_i}^*)(\ho_{\Gamma_i})
\eqs yields whenever $t\in \R$ and $f_{\Gamma_i}\in \ZD(\Ab^{\disc}_{\Gamma_i})$.
\beqs &(\tilde\alpha_{\ha_{\disc,\Gamma_{i}}}(t) (k_{\Gamma_i}\ast f_{\Gamma_i}))(\ho_{\Gamma_i})
=(k_{\Gamma_i}\ast f_{\Gamma_i})(\exp(-t \ha_{\disc,\Gamma_{i}})\ho_{\Gamma_i})\\
&=\int_{\Ab^{\disc}_{\Gamma_i}}\dif\mu_{\breve S_{\disc},\Gamma_i}(\go_{\Gamma_i}(\Gamma_i))  k_{\Gamma_i}(\go_{\Gamma_i}(\Gamma_i)) 
f_{\Gamma_i}(\go_{\Gamma_i}(\Gamma_i)^{-1}\exp(-t \ha_{\disc,\Gamma_{i}})\ho_{\Gamma_i})\\
&=\int_{\Ab^{\disc}_{\Gamma_i}}\dif\mu_{\breve S_{\disc},\Gamma_i}(\go_{\Gamma_i}(\Gamma_i)) k_{\Gamma_i}(\exp(-t \ha_{\disc,\Gamma_{i}})\go_{\Gamma_i}(\Gamma_i)) f_{\Gamma_i}(\ho_{\Gamma_i})\\
&=(\tilde\alpha_{\ha_{\disc,\Gamma_{i}}}(t) (k_{\Gamma_i})\ast \tilde\alpha_{\ha_{\disc,\Gamma_{i}}}(t)(f_{\Gamma_i}))(\ho_{\Gamma_i})
\eqs whenever $t\in \R$ and $k_{\Gamma_i},f_{\Gamma_i}\in \ZD(\Ab^{\disc}_{\Gamma_i})$.
Moreover, $t\mapsto \tilde\alpha_{\ha_{\disc,\Gamma_{i}}}(t)(f_{\Gamma_i})$ is point-norm continuous.
\end{proofs}
Clearly, the same calculations can be done to verify proposition \ref{prop dynsysbetaholflux}.

\begin{prop}
There is a state $\tilde\omega^{\Gamma}_{\LD}$ on $C^*(\Ab^{\disc}_{\Gamma})$ associated to the GNS-triple $(L^2(\Ab^{\disc}_{\Gamma},\mu_{\breve S_{\disc},\Gamma}),\Phi_M,\Omega_{\Gamma}^{\disc})$, which consists of the  $^*$-representation $\Phi_M$ presented in definition \ref{defi represlocholfluxcross}, the Hilbert space $L^2(\Ab^{\disc}_{\Gamma},\mu_{\breve S_{\disc},\Gamma})$ and the cyclic vector $\Omega_\Gamma^{\disc}$.
The state is given by
\beqs \tilde\omega^{\Gamma}_{\LD}(f_\Gamma)
:=\int_{\Ab_\Gamma}\dif\mu_{\Ab_\Gamma}(\ho_\Gamma)\vert f_\Gamma(\ho_\Gamma)\vert^2
\eqs whenever $f_{\Gamma} \in \ZD(\Ab^{\disc}_{\Gamma})$. 
\end{prop}

The set $\ZD(\Ab^{\disc}_{\Gamma_i})$ contains only entire analytic elements for $\beta_{\ha_{\disc,\Gamma_i}}$.
\begin{prop}Let $\Gamma_i$ be a graph, $(C^*(\Ab^{\disc}_{\Gamma_i}),\R,\beta_{\ha_{\disc},\Gamma_{i}})$ and $(\ZD(\Ab^{\disc}_{\Gamma_i}),\R,\tilde\alpha_{\ha_{\disc},\Gamma_{i}})$ be two $C^*$-dynamical systems defined above.

Then the state $\tilde\omega^{\Gamma_{i}}_{\LD}$ is a KMS-state at value $\beta\in\R$ on $C^*(\Ab^{\disc}_{\Gamma_i})$ or, respectively, on $\ZD(\Ab^{\disc}_{\Gamma_i})$. 
\end{prop}
\begin{proofs}
Calculate for $k_{\Gamma_i},f_{\Gamma_i}\in\ZD(\Ab^{\disc}_{\Gamma_i})$ and $f_{\Gamma_i}(\ho_ {\Gamma_i})=f_{\Gamma_i}(\go_{\Gamma_i}\ho_ {\Gamma_i})$ for all $\go_{\Gamma_i}\in\Ab^{\disc}_{\Gamma}$ it is derived that,
\beqs& (k_{\Gamma_i}\ast\tilde\alpha_{\ha_{\disc,\Gamma_{i}}}(i\beta)(f_{\Gamma_i}))(\ho_{\Gamma_i})\\
&=\int_{\Ab^{\disc}_{\Gamma_i}}\dif\mu_{\breve S_{\disc},\Gamma_i}(\go_{\Gamma_i}) k_{\Gamma_i}(\go_{\Gamma_i}) f_{\Gamma_i}(\go_{\Gamma_i}^{-1}\exp(-i\beta \ha_{\disc,\Gamma_{i}})\ho_{\Gamma_i})\\
&=\int_{\Ab^{\disc}_{\Gamma_i}}\dif\mu_{\breve S_{\disc},\Gamma_i}(\go_{\Gamma_i}(\Gamma_i)) k_{\Gamma_i}(\go_{\Gamma_i}) f_{\Gamma_i}(\exp(i\beta \ha_{\disc,\Gamma_{i}})\go_{\Gamma_i}^{-1}\ho_{\Gamma_i})\\
&=\int_{\Ab^{\disc}_{\Gamma_i}}\dif\mu_{\breve S_{\disc},\Gamma_i}(\go_{\Gamma_i}) k_{\Gamma_i}(\go_{\Gamma_i}) f_{\Gamma_i}(\go_{\Gamma_i}^{-1}\ho_{\Gamma_i}\exp(-i\beta \ha_{\disc,\Gamma_{i}}))
\eqs is true.
Then derive
\beqs &\tilde\omega^{\Gamma_{i}}_{\LD}(k_{\Gamma_i}\ast\tilde\alpha_{\ha_{\disc,\Gamma_{i}}}(i\beta)(f_{\Gamma_i}))\\
&= \int_{\Ab^{\disc}_{\Gamma_i}}\dif\mu_{\breve S_{\disc},\Gamma_i}(\ho_{\Gamma_i})\vert (k_{\Gamma_i}\ast\tilde\alpha_{\ha_{\disc,\Gamma_{i}}}(i\beta)(f_{\Gamma_i}))(\ho_{\Gamma_i})\vert^2\\
&=\int_{\Ab^{\disc}_{\Gamma_i}}\dif\mu_{\breve S_{\disc},\Gamma_i}(\ho_{\Gamma_i})\int_{\Ab^{\disc}_{\Gamma_i}}\dif\mu_{\breve S_{\disc},\Gamma_i}(\go_{\Gamma_i}) \vert k_{\Gamma_i}(\go_{\Gamma_i}) f_{\Gamma_i}(\go_{\Gamma_i}^{-1}\ho_{\Gamma_i}\exp(-i\beta \ha_{\disc,\Gamma_{i}}))\vert^2\\
&=\int_{\Ab^{\disc}_{\Gamma_i}}\dif\mu_{\breve S_{\disc},\Gamma_i}(\ho_{\Gamma_i})\int_{\Ab^{\disc}_{\Gamma_i}}\dif\mu_{\breve S_{\disc},\Gamma_i}(\go_{\Gamma_i}) \vert k_{\Gamma_i}(\go_{\Gamma_i}) f_{\Gamma_i}(\go_{\Gamma_i}^{-1}\ho_{\Gamma_i})\vert^2\\
&=\int_{\Ab^{\disc}_{\Gamma_i}}\dif\mu_{\breve S_{\disc},\Gamma_i}(\go_{\Gamma_i}) \int_{\Ab^{\disc}_{\Gamma_i}}\dif\mu_{\breve S_{\disc},\Gamma_i}(\ho_{\Gamma_i})\vert f_{\Gamma_i}(\go_{\Gamma_i}) k_{\Gamma_i}(\go_{\Gamma_i}^{-1}\ho_{\Gamma_i})\vert^2\\
&=\tilde\omega^{\Gamma_{i}}_{\LD}(f_{\Gamma_i}\ast k_{\Gamma_i})
\eqs
\end{proofs}

Clearly, the state $\tilde\omega^\Gamma_{\LD}$ is $\R$-invariant 
\beqs \tilde\omega^\Gamma_{\LD}(f_\Gamma)=\tilde\omega^\Gamma_{\LD}(\tilde \alpha_{\ha_{\disc,\Gamma_{i}}}(t)(f_\Gamma))\text{ for }f_\Gamma\in \ZD  (\Ab^{\disc}_{\Gamma})\text{ and all }t\in\R
\eqs 

Recall the $^*$-derivation $\delta_{S_{\disc},\Gamma_{i+1}}$ given in definition \ref{defi derivonmatrix2}, then the state $\tilde\omega^\Gamma_{\LD}$ on $\ZD(\Ab^{\disc}_{\Gamma_i})$ satisfies
\beq \tilde\omega^{\Gamma_{i+1}}_{\LD}(\delta_{S_{\disc},\Gamma_{i+1}}(f_{\Gamma_{i+1}}))
&\quad = \tilde\omega^{\Gamma_{i+1}}_{\LD}(\delta_{S_{\disc},\Gamma_{i}}(f_{\Gamma_{i+1}}) + \bra  \tilde E_{S_{\disc}}^+(\Gamma_{i+1})\tilde E_{S_{\disc}}(\Gamma_{i+1}),f_{\Gamma_{i+1}}\ket)\\
&\quad = \tilde\omega^{\Gamma_{i}}_{\LD}(\delta_{S_{\disc},\Gamma_{i}}(f_{\Gamma_{i}}))
\eq

Hence, the limit state $\tilde\omega_{\LD}$ of the states $\tilde\omega^{\Gamma_{i+1}}_{\LD}$ on the $^*$-algebra $\ZD(\Ab^{\disc})$ is compatible with the family of $^*$-derivations $\{\delta_{S_{\disc},\Gamma_{i+1}}\}$. 
Recall the $^*$-representations presented in definition  \ref{defi represlocholfluxcross}.

\begin{cor}
The state $\tilde\omega^{\Gamma}_{\LD}$ on $\ZD(\Ab^{\disc}_{\Gamma})$ extends to a state $\hat\omega^{\Gamma}_{\LD}$ on $\ZD(\Ab^{\disc}_{\Gamma})\rtimes_L \bar\E^{\loc}_{\breve S_{\disc},\Gamma}$.

Equivalently, the $^*$-representation $\Phi_M$ on $\HS_\Gamma^{\disc}$ with cyclic vector $\Omega_\Gamma^{\disc}$ constructed from $\tilde\omega^{\Gamma}_{\LD}$ of $\ZD(\Ab^{\disc}_{\Gamma})$ extends to a $^*$-representation $\hat\pi_\Gamma^{\disc}$ on $\HS_\Gamma^{\disc}$ with cyclic vector $\hat\Omega_\Gamma^{\disc}$ of $\ZD(\Ab^{\disc}_{\Gamma}) \rtimes_L \bar\E^{\loc}_{\breve S_{\disc},\Gamma}$.
\end{cor}
\begin{proofs}Notice that, it is true that $\bra E_{S_{\disc}}(\Gamma_{j})^+E_{S_{\disc}}(\Gamma_{j}),\Omega_\Gamma^{\disc}\ket=0$ and, hence, derive 
\beqs \hat\omega^{\Gamma}_{\LD}(f_\Gamma\otimes iE_{S_{\disc}}(\Gamma)^+E_{S_{\disc}}(\Gamma))
&=\tilde\omega^{\Gamma}_{\LD}\Big(\frac{1}{2}\delta_{S_{\disc},\Gamma}(f_\Gamma)\Big)
+\la \Omega_\Gamma^{\disc}, \frac{1}{2} i\bra E_{S_{\disc}}(\Gamma_{j})^+E_{S_{\disc}}(\Gamma_{j})\Omega_\Gamma^{\disc}\ket \ra\\
&=\tilde\omega^{\Gamma}_{\LD}\Big(\frac{1}{2}\delta_{S_{\disc},\Gamma}(f_\Gamma)\Big)
\eqs whenever $E_{S_{\disc}}\in \go_{\breve S_{\disc},\Gamma}$.
\end{proofs}

Notice that, the state $\tilde\omega^{\Gamma}_{\LD}$ on $\ZD(\Ab^{\disc}_{\Gamma})$ is not unique. Since, for $K_\Gamma\in L^1(\Ab^{\disc}_{\Gamma},\mu_{\disc,\Gamma})$ there is another state defined by
\beqs \tilde\omega^{\Gamma}_{\LD,K}(f_\Gamma)
:=\int_{\Ab^{\disc}_{\Gamma}}K_\Gamma(\ho_\Gamma)\dif\mu_{\Ab^{\disc}_{\Gamma}}(\ho_\Gamma)\vert f_\Gamma(\ho_\Gamma)\vert^2
\eqs whenever $f_{\Gamma} \in \ZD(\Ab^{\disc}_{\Gamma})$.

There exists a limit state $\hat\omega_{\LD}$ of the states $\{\hat\omega^{\Gamma_{i+1}}_{\LD}\}$ on the $^*$-algebra $\ZD(\Ab^{\breve S_{\disc}})\rtimes_L\bar\E^{\loc}_{\breve S_{\disc}}$. 

Recall that, there is a group action of $\mathfrak{B}(\PD^{\breve S_{\disc}}_\Gamma)$ on $\ZD(\Ab^{\breve S_{\disc}})$. This action is also action on $\ZD(\Ab^{\breve S_{\disc}}) \rtimes_L \bar\E^{\loc}_{\breve S_{\disc}}$, since, $E_{S_{\disc}}^+(\Gamma_\sigma)E_{S_{\disc}}(\Gamma_\sigma)=E_{S_{\disc}}^+(\Gamma)E_{S_{\disc}}(\Gamma)$ is true and, hence, 
\beqs\zeta_{\sigma}(f_{\Gamma}\otimes iE_{S_{\disc}}^+(\Gamma)E_{S_{\disc}}(\Gamma)) =\zeta_{\sigma}(f_{\Gamma})\otimes iE_{S_{\disc}}^+(\Gamma)E_{S_{\disc}}(\Gamma))
\eqs whenever $\sigma\in\mathfrak{B}(\PD^{\breve S_{\disc}}_\Gamma)$ holds.

\begin{defi}
Denote the center of the Lie flux group $\bar G_{\breve S_{\disc},\Gamma}$ by $\hat\ZD(\bar G_{\breve S_{\disc},\Gamma})$ and the Lie flux algebra associated to $\hat\ZD(\bar G_{\breve S_{\disc},\Gamma})$ by $\mathfrak{z}_{\breve S_{\disc},\Gamma}$. Finally, the enveloping algebra of the Lie flux algebra $\mathfrak{z}_{\breve S_{\disc},\Gamma}$ is denoted by $\mathfrak{E}_{\breve S_{\disc},\Gamma}$
\end{defi}

Recall that, the space $\Ab_\Gamma^{\disc}$ is identified with $G^{\vert\Gamma\vert}$. The state $\hat\omega_{\LD}^\Gamma$ on $\ZD(\Ab^{\disc}_{\Gamma})$ is already $\Diff(\PD^{\breve S_{\disc}}_\Gamma)$-invariant. 

\begin{defi}\label{defi locholflux}Let $\{\Gamma_i\}$ be an inductive family of graphs with inductive limit $\Gamma_\infty$, $\breve S$ be a set of surfaces and $\breve S_{\disc}$ a set of discretised surfaces associated to $\breve S$ such that the assumptions in definition \ref{defi locfluxop} are satisfied.

Then the \textbf{localised holonomy-flux cross-product $^*$-algebra associated to a discretised surface set $\breve S_{\disc}$} is given by the following tensor product 
\[ C (\Ab_{\loc}) \otimes \ZD(\Ab^{\breve S_{\disc}}) \rtimes_L \mathfrak{E}_{\breve S_{\disc}}\] 

The cross-product $^*$-algebra $\ZD(\Ab^{\breve S_{\disc}}) \rtimes_L \mathfrak{E}_{\breve S_{\disc}}$ is called the \textbf{localised part of the localised holonomy-flux cross-product $^*$-algebra} $C(\Ab_{\loc}) \otimes\ZD(\Ab^{\breve S_{\disc}}) \rtimes_L \mathfrak{E}_{\breve S_{\disc}}$.
\end{defi}
Note that, the localised holonomy-flux cross-product $^*$-algebra associated to a discretised surface set $\breve S_{\disc}$ is abbreviated by the term localised holonomy-flux cross-product $^*$-algebra for surfaces.

\begin{theo}Let $\{\Gamma_i\}$ be an inductive family of graphs with inductive limit $\Gamma_\infty$, $\breve S$ be a set of surfaces and $\breve S_{\disc}$ a set of discretised surfaces associated to $\breve S$ such that the assumptions in definition \ref{defi locfluxop} are satisfied.

 Then there exists a $\Diff(\PD^{\breve S_{\disc}}_\Gamma)$- and $\Diff(\PD_{\bar\Gamma})$-invariant state on $C(\Ab_{\loc}) \otimes\ZD(\Ab^{\breve S_{\disc}}) \rtimes_L \mathfrak{E}_{\breve S_{\disc}}$, which is a product state of the state $\omega_{M}$ of $C(\Ab_{\loc})$ and the state $\hat\omega_{\LD}$ of $\ZD(\Ab^{\breve S_{\disc}}) \rtimes_L \mathfrak{E}_{\breve S_{\disc}}$. The state $\hat\omega_{\LD}^\Gamma$ is a KMS-state on $\ZD(\Ab^{\disc}_{\Gamma}) $ at inverse temperature $\beta\in\R$ w.r.t. the automorphism  $\tilde\alpha_{\ha_{\disc},\Gamma}$.
\end{theo}

Summarising a modified holonomy-flux algebra is constructed. The assumption of diffeomorphism invariance of the state space of the modified algebra is relaxed to a surface-preserving graph-diffeomorphism invariance for a finite set $\breve S$ of surfaces and an arbitrary fixed graph $\Gamma$. 

Finally, if different surface sets are considered, then the following is true.
There is a family $\{C(\Ab_{\bar\Gamma}) \otimes \ZD(\Ab^{\breve S_{\disc}}_{\Gamma}) \rtimes_L \mathfrak{E}_{\breve S_{\disc},\Gamma}\}_{\Gamma}$ of localised holonomy-flux cross-product $^*$-algebras associated to graphs and a surface set $\breve S$. Consider a subset $\breve S^1$ of $\breve S$ and $\breve S^1_{\disc}$ of $\breve S_{\disc}$ such that the assumptions in definition \ref{defi locfluxop} are satisfied. Then for every surface $S_1$ in $\breve S^1$ there is a surface $S$ in $\breve S$ with $S_1\subset S$ and $S_1^{\disc}\subset S^{\disc}$. 
Then it is true that, the algebra $\ZD(\Ab^{\breve S^1_{\disc}}_{\Gamma}) \rtimes_L \mathfrak{E}_{\breve S^1_{\disc},\Gamma}$ is a subalgebra of $\ZD(\Ab^{\breve S^2_{\disc}}_{\Gamma})\rtimes_L \mathfrak{E}_{\breve S_{\disc},\Gamma}$. But this is not true for the full localised holonomy-flux cross-product $^*$-algebras associated to a graph and a surface set $\breve S$.

For two disjoint surface sets $\breve S_{\disc}^1$ and $\breve S_{\disc}^2$ the elements of the localised holonomy-flux cross-product $^*$-algebras satisfies some relations. But there is no easy locality relation such that two algebra elements commute, i.e. $A\in C(\Ab_{\bar\Gamma}) \otimes \ZD(\Ab^{\breve S^2_{\disc}}_{\Gamma}) \rtimes_L \mathfrak{E}_{\breve S^2_{\disc},\Gamma}$ and $B\in C(\Ab_{\bar\Gamma}) \otimes \ZD(\Ab^{\breve S^1_{\disc}}_{\Gamma})\rtimes_L \mathfrak{E}_{\breve S^1_{\disc},\Gamma}$ it is not true that, $\bra A, B\ket=0$ yields. Notice that, the quantum flux operators $E_{S_{\disc}}^1(\Gamma)\in\mathfrak{E}_{\breve S^1_{\disc},\Gamma}$ and $E_{S_{\disc}}^2(\Gamma)\in \mathfrak{E}_{\breve S^2_{\disc},\Gamma}$ satisfy
$\bra E_{S_{\disc}}^1(\Gamma),E_{S_{\disc}}^2(\Gamma)\ket=0$.

\subsection{The quantum constraints and a suggestion for a physical $^*$-algebra}\label{subsec QHamilton}

Recall the quantum Hamilton constraint operator defined in \cite{Kaminski0,KaminskiPHD}, which is given for example in \cite{Thiem96} by the expression
\beqs\QD(C_T(N))=\sum_{\Delta\in T}\tr\Big(\left(\ho_A(l_\Delta)-\ho_A(l_\Delta)^{-1}\right)\ho_{A}(\gamma_\Delta)[\ho_{A}(\gamma_\Delta)^{-1},\QD(V)]\Big)
\eqs where $\ho_A(l_\Delta)$ denotes a holonomy along a loop $l_\Delta$ in a subset $\Delta$ of a triangulation $T$, $\gamma_\Delta$ denotes a path. Let $\breve S:=\{S_1,S_2,S_3\}$ be a set of surfaces associated to the triangulation $T$. Then the quantum volume operator $\QD(V)$ is defined by
\beqs \QD(V)=\sum_{\gamma_1,\gamma_2,\gamma_3}E_{S_1}(\gamma_1)E_{S_2}(\gamma_2)E_{S_3}(\gamma_3)
\eqs the sum is over all triples of paths, which are build from three paths that intersects in a common vertex $v$. Consequently, one can localise the quantum volume operator and the quantum Hamilton constraint operator on a set $\breve S_{\disc}$ of discretised surfaces and a graph system $\PD_\Gamma$. The resulting operators are denoted by $\QD(C_T(N))_{\disc,\Gamma}$ or $\QD(V)_{\disc,\Gamma}$ and are called the \textbf{modified (or discretised) quantum Hamilton constraint} or the \textbf{discretised quantum volume operator} associated to graphs. But the operator $\QD(C_T(N))_{\disc,\Gamma}$ is neither an element of $\ZD(\Ab^{\disc}_{\Gamma})$ nor $C(\Ab_{\loc}) \otimes\ZD(\Ab^{\breve S_{\disc}}) \rtimes_L \mathfrak{E}_{\breve S_{\disc}}$. 

Consequently, in the following the quantum Hamilton constraint is modified. First notice the \textbf{quantum Hamilton constraint $H_\Gamma$ restricted to a graph}, which is given by 
\beqs \exp(H_\Gamma):= \left(\ho_\Gamma(l)-\ho_A(l)^{-1}\right)\ho_\Gamma(\gamma)[\ho_\Gamma(\gamma)^{-1},\QD(V)_{\disc,\Gamma}]
\eqs whenever $l,\gamma\in\Gamma$. In particular, \textbf{the quantum Hamilton part $H_{\Gamma,P}$ restricted to a graph} is given by
\beqs \exp(H_{\Gamma,P}):=\left(\ho_\Gamma(l)-\ho_\Gamma(l)^{-1}\right)\ho_{\Gamma}(\gamma)
\eqs whenever $l,\gamma\in\Gamma$. The operator $[\ho_{\Gamma}(\gamma)^{-1},\QD(V)_{\disc,\Gamma}]$ is omitted first. Then the quantum Hamilton part is localised such that $H_{\Gamma,P}^+H_{\Gamma,P}$ is an element of the enveloping Lie algebra of the Lie algebra $\LAb_{\disc,\Gamma}$ associated to $\Ab^{\disc}_{\Gamma}$.

The quantum Hamilton part $H_{\Gamma,P}$ defines an action of $\R$ on $\ZD(\Ab^{\disc}_{\Gamma_i})$ by
\beqs(\alpha_{H_{\Gamma_i,P}^+H_{\Gamma_i,P}}(t) f_{\Gamma_i})(\ho_{\Gamma_i})
:=f_{\Gamma_i}(\exp(-t H_{\Gamma_i,P}^+H_{\Gamma_i,P})\ho_{\Gamma_i})
\eqs and 
\beqs(\alpha_{H_{\Gamma_i,P}^+H_{\Gamma_i,P}}(t) f_{\Gamma_i})^*(\ho_{\Gamma_i})
=\overline{(\alpha_{H_{\Gamma_i,P}H_{\Gamma_i,P}^+}(t) f_{\Gamma_i})(\ho_{\Gamma_i}^{-1})}
\eqs whenever $f_{\Gamma_i}\in \ZD(\Ab^{\disc}_{\Gamma_i})$ and $H_{\Gamma,P}^+H_{\Gamma,P}\in \LAb_{\breve S_{\disc},\Gamma}$ yields.
\begin{prop}
The triple $(\ZD(\Ab^{\disc}_{\Gamma_i}),\R,\alpha_{H_{\Gamma_i,P}^+H_{\Gamma_i,P}})$ is a $C^*$-dynamical system.

The state $\tilde\omega^{\Gamma_{i}}_{\LD}$ is a KMS-state at value $\beta\in\R$ on $\ZD(\Ab^{\disc}_{\Gamma_i})$ such that
\beqs \tilde\omega^{\Gamma_{i}}_{\LD}(A\alpha_{H_{\Gamma_i,P}^+H_{\Gamma_i,P}}(i\beta)(B))=\tilde\omega^{\Gamma_{i}}_{\LD}(BA)\eqs holds
for all $A,B\in \ZD(\Ab^{\disc}_{\Gamma_i})$.
\end{prop}

\begin{prop} \label{prop noninvstatehol}
Let the triples  $(\ZD(\Ab_{\Gamma_i}^{\disc}),\Diff(\PD_{\Gamma}^{\breve S_{\disc}}),\zeta)$ and $(\ZD(\Ab^{\disc}_{\Gamma_i}),\R,\alpha_{H_{\Gamma_i,P}^+H_{\Gamma_i,P}})$ be $C^*$-dynamical systems. 

Then the automorphisms $\zeta$ and $\alpha_{H_{\Gamma_i,P}^+H_{\Gamma_i,P}}$ on $\ZD(\Ab^{\disc}_{\Gamma_i})$ do not commute.  
\end{prop}

\begin{defi}
Denote the center of the compact Lie group $\Ab^{\disc}_{\Gamma}$ by $\bar\ZzD_{\breve S_{\disc},\Gamma}$. 
\end{defi}

The problem in proposition \ref{prop noninvstatehol} is solved, if it is assumed that, $\exp(tH_{\Gamma,P}^+H_{\Gamma,P})\in \bar\ZzD_{\breve S_{\disc},\Gamma}$ holds for all $t\in\R$. The state $\tilde\omega^{\Gamma}_{\LD}$ is $\Diff(\PD^{\breve S_{\disc}}_\Gamma)$-invariant. Furthermore this state is $\R$-invariant for all quantum Hamilton parts $H_{\Gamma_i,P}$ such that $\exp(tH_{\Gamma_i,P}^+H_{\Gamma_i,P})\in \bar\ZzD_{\breve S_{\disc},\Gamma}$ yields for all $t\in\R$ and $\Gamma_i$ being a subgraph of $\Gamma$.

Notice that, $(\alpha_{H_{\Gamma,P}^+H_{\Gamma,P}})(A_{\bar\Gamma})=A_{\bar\Gamma}$ holds for all $A_{\bar\Gamma}\in C(\Ab_{\bar\Gamma})$ and
\beqs (\alpha_{H_{\Gamma,P}^+H_{\Gamma,P}})(f_\Gamma\otimes E_{\breve S_{\disc}}^+(\Gamma)E_{\breve S_{\disc}}(\Gamma))
=(\alpha_{H_{\Gamma,P}^+H_{\Gamma,P}})(f_\Gamma)\otimes E_{\breve S_{\disc}}^+(\Gamma)E_{\breve S_{\disc}}(\Gamma)
\eqs for all $f_\Gamma\in \ZD(\Ab^{\disc}_{\Gamma})$ and $E_{\breve S_{\disc}}^+(\Gamma)E_{\breve S_{\disc}}(\Gamma)\in\mathfrak{E}_{\breve S_{\disc},\Gamma}$. Consequently, $\alpha_{H_{\Gamma,P}^+H_{\Gamma,P}}\in\Aut(\ZD(\Ab^{\disc}_\Gamma)\rtimes_L \mathfrak{E}_{\breve S_{\disc},\Gamma})$.

\begin{theo}Let $\{\Gamma_i\}$ be an inductive family of graphs, $\breve S$ be a set of surfaces and $\breve S_{\disc}$ a set of discretised surfaces associated to $\breve S$ such that the assumptions in definition \ref{defi locfluxop} are satisfied. For a fixed graph $\Gamma$ let $(C(\Ab_{\loc}),\Diff(\PD_{\bar\Gamma}),\zeta)$ and $(\ZD(\Ab^{\breve S_{\disc}}),\Diff(\PD_{\Gamma}^{\breve S_{\disc}}),\zeta)$ be two $C^*$-dynamical systems. 

Moreover, let $\{H_{\Gamma_i,P}\}$ be a family of quantum Hamilton parts restricted to graphs such that each element $\exp(tH_{\Gpi,P}^+H_{\Gpi,P})\in \bar\ZzD_{\breve S_{\disc},\Gamma_i}$ for all $t\in\R$ and all graphs $\Gpi$ being a subgraph of $\Gamma_i$.

Let $\{(\ZD(\Ab^{\disc}_{\Gamma_i}),\R,\alpha_{H_{\Gamma_i,P}^+H_{\Gamma_i}})\}$ be a family of $C^*$-dynamical systems. Finally, let $\hat\omega_{\LD}$ be the limit state on  the $^*$-algebra $\ZD(\Ab^{\disc})\rtimes_L \mathfrak{E}_{\breve S_{\disc}}$ of the states $\{\hat\omega^{\Gamma_i}_{\LD}\}$ of the families $\{\ZD(\Ab^{\disc}_{\Gamma_i})\rtimes_L \mathfrak{E}_{\breve S_{\disc},\Gamma_i}\}$ of $^*$-algebras. The state $\hat\omega^{\Gamma_i}_{\LD}$ is a KMS-state for $\ZD(\Ab^{\disc}_{\Gamma_i})$ at value $\beta\in\R$ for $\alpha_{H_{\Gamma_i,P}^+H_{\Gamma_i,P}}$ and such that
\beqs \hat\omega^{\Gamma_i}_{\LD}\circ\alpha_{H_{\Gamma_i,P}^+H_{\Gamma_i,P}}=\hat\omega^{\Gamma_i}_{\LD}
\eqs for a graph $\Gamma_i$ and all $1\leq i<\infty$. 

Then for a fixed graph $\Gamma$ there exists a $\Diff(\PD^{\breve S_{\disc}}_\Gamma)$- and $\Diff(\PD_{\bar\Gamma})$-invariant state on $C(\Ab_{\loc}) \otimes\ZD(\Ab^{\breve S_{\disc}}) \rtimes_L \mathfrak{E}_{\breve S_{\disc}}$, which is a product state on a state $\omega_M$ of $C(\Ab_{\loc})$ and a  state $\hat\omega_{\LD}$ of $\ZD(\Ab^{\breve S_{\disc}}) \rtimes_L \mathfrak{E}_{\breve S_{\disc}}$.
\end{theo}

The state on the localised holonomy-flus cross product $^*$-algebra depends on the family of KMS-states and the state $\omega_M$ of $C(\Ab_{\loc})$. Notice that, $\omega_M$ need not be $\bar G_{\breve S_{\disc},\bar\Gamma}$-invariant for any graph $\bar\Gamma$. This is indeed distinguishing from the results of the analytic holonomy $C^*$-algebra, where the state is required to be invariant. For example refer to \cite[Cor.: 3.60]{Kaminski1}\cite[Cor.: 6.4.3]{KaminskiPHD}. But since there is no action of the fluxes on this part of the localised holonomy-flux cross-product $^*$-algebra, this invariance is not required.

Contrary to corollary \ref{cor commdiscrflux} consider the following remark.
\begin{rem}
The limit
\beqs \tilde\delta_{S_{\disc},P}(f)&:
=i \lim_{j\rightarrow\infty}\bra H_{\Gpj,P}^+H_{\Gpj,P},f\ket
\eqs for every $f\in \DD(\tilde \delta_{S_{\disc},P})$ and every element $\exp(H_{\Gpj,P}^+H_{\Gpj,P})\in\bar \ZzD_{\breve S_{\disc},\Gamma_k}$ and $\Gpj<\Gamma_k$, $i,k\in\N$, is not well-defined in the norm topology. 
\end{rem}

Until now only quantum Hamilton parts restricted to certain subgraphs of a graph are considered. Hence, the full quantum Hamilton part for a family $\{\Gamma_i\}$ of graphs is given by
\beqs H_P^+H_P:=\lim_{\Gamma_i\rightarrow\Gamma_\infty}\sum_{\Gpi\in\PD_{\Gamma_i}}H_{\Gpi,P}^+H_{\Gpi,P}
\eqs

\begin{prop}\label{prop fulldiffeobroken}Let $\{\Gamma_i\}$ be an inductive family of graphs, $\breve S$ be a set of surfaces and $\breve S_{\disc}$ a set of discretised surfaces associated to $\breve S$ such that the assumptions in definition \ref{defi locfluxop} are satisfied. 

Moreover, let $\{H_{\Gamma_i,P}\}$ be a family of quantum Hamilton parts restricted to graphs such that each element $\exp(tH_{\Gpi,P}^+H_{\Gpi,P})\in \bar\ZzD_{\breve S_{\disc},\Gamma_i}$ for all $t\in\R$ and each subgraph $\Gpi$ of $\Gamma_i$.

Let $\{(\ZD(\Ab^{\disc}_{\Gamma_i}),\R,\alpha_{H_{\Gamma_i,P}^+H_{\Gamma_i,P}})\}$ be a family of $C^*$-dynamical systems. Finally, let $\{\hat\omega^{\Gamma_i}_{\LD,\mathfrak{B}}\}$ be a family of states such that $\hat\omega^{\Gamma_i}_{\LD,\mathfrak{B}}$ is a state on the $^*$-algebra  $\ZD(\Ab^{\disc}_{\Gamma_i})\rtimes_L \mathfrak{E}_{\breve S_{\disc},\Gamma_i}$. 

Then the limit state $\hat\omega_{\LD,P}$ on $\ZD(\Ab^{\breve S_{\disc}})\rtimes_L \mathfrak{E}_{\breve S_{\disc}}$, which is given by
 \beqs  \hat\omega_{\LD,P}(A)
:=\lim_{\Gamma_i\rightarrow\Gamma_\infty}\frac{1}{ \vert\PD_{\Gamma_i}\vert}\sum_{\Gpi\in\PD_{\Gamma_i}}\hat\omega^{\Gamma_i}_{\LD}\left(\alpha_{H_{\Gpi,P}^+H_{\Gpi,P}}(t)(A)\right)
\eqs and where $\vert\PD_{\Gamma_i}\vert$ denotes the number of subgraphs of a graph $\Gamma_i$, and which is $\R$-invariant w.r.t. the automorphism group $t\mapsto \alpha_{H_P^+H_P}(t)$, for $A\in\ZD(\Ab^{\breve S_{\disc}})\rtimes_L \mathfrak{E}_{\breve S_{\disc}}$, does not converge in weak $^*$-topology.
\end{prop}
This proposition implies that, the one-parameter group  $t\mapsto\alpha_{H_P^*H_P}(t)$ of $^*$-automorphisms is not strongly continuous. Consequently, the derivation $\delta_{P}$ on $\ZD(\Ab^{\breve S_{\disc}})$, which is given by
\beqs \delta_{P}(f):=\lim_{t\rightarrow 0}\frac{1}{t}\Big(\alpha_{H_P^*H_P}(t)(f)-f\Big)=\lim_{t\rightarrow 0}\text{  } \lim_{\Gamma_i\rightarrow\Gamma_\infty}\frac{1}{t \vert\PD_{\Gamma_i}\vert}\Big(\sum_{\Gpi\in\PD_{\Gamma_i}}\alpha_{H_{\Gp,P}^*H_{\Gp,P}}(t)(f)-f\Big)
\eqs
for $f\in\ZD(\Ab^{\breve S_{\disc}})$, is not converging in norm.

Now, recall the operator $[\ho_{A}(\gamma)^{-1},\QD(V)_{\disc,\Gamma}]$, whenever $\gamma\in\Gamma$ and where $\QD(V)_{\disc,\Gamma}$ is sum over finite products of discretised flux operators for a surface $S_{\disc}$ and a graph $\Gamma$. Then the quantum Hamilton restricted to a graph contains also elements of $\bar\go^{\loc}_{\breve S_{\disc},\Gamma}$. 

Consequently define the discretised and localised quantum flux operator associated to a surface $S_{\disc}$ and a family $\{\Gamma_i\}$ of graphs by
\beqs \tilde E_{S_{\disc}}^+\tilde E_{S_{\disc}}:=\lim_{\Gamma_i\rightarrow\Gamma_\infty}\sum_{\Gpi\in\PD_{\Gamma_i}}\tilde E_{S_{\disc},\Gpi}^+\tilde E_{S_{\disc},\Gpi}
\eqs where $E_{S_{\disc},\Gpi}:=E_{S_{\disc}}(\Gpi)\in\bar\go_{\breve S_{\disc},\Gamma_i}^{\loc}$ for every subgraph $\Gpi$ of $\Gamma_i$.

\begin{prop}Let $\{\Gamma_i\}$ be an inductive family of graphs, $\breve S$ be a set of surfaces and $\breve S_{\disc}$ a set of discretised surfaces associated to $\breve S$ such that the assumptions in definition \ref{defi locfluxop} are satisfied. 

Let $\Big\{(\ZD(\Ab^{\disc}_{\Gamma_i}),\R,\alpha_{\tilde E_{ S_{\disc},\Gamma_i}^+\tilde E_{S_{\disc},\Gamma_i}})\Big\}$ be a family of $C^*$-dynamical systems. Morever, let $\{\hat\omega^{\Gamma_i}_{\LD}\}$ be a family of states such that $\hat\omega^{\Gamma_i}_{\LD}$ is a state on the $^*$-algebra  $\ZD(\Ab^{\disc}_{\Gamma_i})\rtimes_L \mathfrak{E}_{\breve S_{\disc},\Gamma_i}$. 

Then the limit $\hat\omega_{\LD,E}$ on $\ZD(\Ab^{\breve S_{\disc}})\rtimes_L \mathfrak{E}_{\breve S_{\disc}}$, which is given by
 \beqs  \hat\omega_{\LD,E}(A)
:=\lim_{\Gamma_i\rightarrow\Gamma_\infty}\frac{1}{ \vert\PD_{\Gamma_i}\vert}\sum_{\Gpi\in\PD_{\Gamma_i}}\hat\omega^{\Gamma_i}_{\LD}\left(\alpha_{\tilde E_{S_{\disc},\Gpi}^+\tilde E_{S_{\disc},\Gpi}}(t)(A)\right)
\eqs whenever $A\in\ZD(\Ab^{\breve S_{\disc}})\rtimes_L \mathfrak{E}_{\breve S_{\disc}}$ and where $\vert\PD_{\Gamma_i}\vert$ denotes the number of subgraphs of a graph $\Gamma_i$. The state $\hat\omega_{\LD,E}$ is $\R$-invariant w.r.t. the automorphism group $t\mapsto \alpha_{\tilde E_{S_{\disc}}^+\tilde E_{S_{\disc}}}(t)$ and converges in weak $^*$-topology.
\end{prop}
\begin{proofs}
Derive
\beqs  &\lim_{\Gamma_i\longrightarrow\Gamma_\infty}\Big\vert
\hat\omega_{\LD}^{\Gamma_i}(A)-\frac{1}{ \vert\PD_{\Gamma_i}\vert}\sum_{\Gpi\in\PD_{\Gamma_i}}\hat\omega^{\Gamma_i}_{\LD}\left(\alpha_{\tilde E_{S_{\disc},\Gpi}^+\tilde E_{S_{\disc},\Gpi}}(t)(A)\right)\Big\vert
= \Big\vert \hat\omega_{\LD}^{\Gamma_0}(A)- \hat\omega_{\LD}^{\Gamma_0}\left(\alpha_{E_{S_{\disc},\Gamma_0}^+ E_{S_{\disc},\Gamma_0}}(t)(A)\right)\Big\vert\\
&=0
\eqs
\end{proofs}

Recall proposition \ref{defi derivonmatrix2}. Furthermore the last proposition implies that, the derivation $\delta_{E}$ on $\ZD(\Ab^{\breve S_{\disc}})$, which is given by
\beqs \delta_{E}(f):=\lim_{t\rightarrow 0}\frac{1}{t}\Big(\alpha_{E_{S_{\disc},\Gpi}^+E_{S_{\disc},\Gpi}}(t)(f)-f\Big)=\lim_{t\rightarrow 0}\text{  } \lim_{\Gamma_i\rightarrow\Gamma_\infty}\frac{1}{t \vert\PD_{\Gamma_i}\vert}\Big(\sum_{\Gpi\in\PD_{\Gamma_i}}\alpha_{E_{S_{\disc},\Gpi}^+E_{S_{\disc},\Gpi}}(t)(f)-f\Big)
\eqs
for $f\in\ZD(\Ab^{\breve S_{\disc}})$, converges in norm.

\begin{problem}Let $\{\Gamma_i\}$ be an inductive family of graphs, $\breve S$ be a set of surfaces and $\breve S_{\disc}$ a set of discretised surfaces associated to $\breve S$ such that the assumptions in definition \ref{defi locfluxop} are satisfied. For a fixed graph $\Gamma$ let $(C(\Ab),\mathfrak{B}(\PD_{\bar\Gamma}),\zeta)$ and $(\ZD(\Ab^{\breve S_{\disc}}),\mathfrak{B}_{\breve S_{\disc},\diff}(\PD_{\Gamma}^{\breve S_{\disc}}),\zeta)$ be two $C^*$-dynamical systems. 

The discetrised quantum volume operator is explicity defined by
\beqs \QD(V^*V)_{\disc,\Gamma}:=\sum_{\underset{\in\PD^v_\Gamma\times\PD^v_\Gamma\times\PD^v_\Gamma}{(\gamma_1,\gamma_2,\gamma_3)}}E_{S_3^{\disc}}(\gamma_3)^+E_{S_2^{\disc}}(\gamma_2)^+E_{S_1^{\disc}}(\gamma_1)^+E_{S_1^{\disc}}(\gamma_1)E_{S_2^{\disc}}(\gamma_2)
E_{S_3^{\disc}}(\gamma_3)
\eqs such that $\QD_{\disc,\Gamma}(V^*V)\in \mathfrak{E}_{\breve S_{\disc},\Gamma}$.
Recall the quantum Hamilton constraint $H_\Gamma$ restricted to a graph is presented by
\beqs
\exp(H_\Gamma):=\exp(H_{\Gamma,P})\bra \ho_{\Gamma}(\gamma), \QD(V)_{\disc,\Gamma}\ket
\eqs
Moreover let $\{H_{\Gamma_i}\}$ be a family of quantum Hamilton constraints restricted to graphs such that  each element $\exp(tH_{\Gpi}^+H_{\Gpi})\in C^*(\Ab^{\disc}_{\Gamma_i})\rtimes \mathfrak{E}_{\breve S_{\disc},\Gamma_i}$ for all $t\in\R$ and all graphs $\{\Gpi\}$ being subgraphs of $\Gamma_i$.

Recall the family $\{\hat\omega^{\Gamma_i}_{\LD}\}$ of states of the family $\{\ZD(\Ab^{\disc}_{\Gamma_i})\rtimes \mathfrak{E}_{\breve S_{\disc},\Gamma_i}\}$ of $^*$-algebras, which are KMS-states for $\ZD(\Ab^{\disc}_{\Gamma_i})$ at value $\beta\in\R$ and such that the states satisfy 
\beqs &\hat\omega^{\Gamma_i}_{\LD}\circ\alpha_{H_{\Gpi,P}^+H_{\Gpi,P}}=\hat\omega^{\Gamma_i}_{\LD}\\
&\hat\omega^{\Gamma_i}_{\LD}\circ \alpha_{H_{\Gamma_i,P}^+H_{\Gamma_i,P}}(t)\circ \zeta_\sigma= \hat\omega^{\Gamma_i}_{\LD} =\hat\omega^{\Gamma_i}_{\LD}\circ \zeta_\sigma\circ \alpha_{H_{\Gamma_i,P}^+H_{\Gamma_i,P}}(t)\\
&\hat\omega^{\Gamma_i}_{\LD}\circ \alpha_{E^+_{S_{\disc},\Gamma_i}E_{S_{\disc},\Gamma_i}}=\hat\omega^{\Gamma_i}_{\LD}\\
&\hat\omega^{\Gamma_i}_{\LD}\circ\zeta_\sigma\circ \alpha_{E^+_{S_{\disc},\Gamma_i}E_{S_{\disc},\Gamma_i}}=\hat\omega^{\Gamma_i}_{\LD}=\hat\omega^{\Gamma_i}_{\LD}\circ\alpha_{E^+_{S_{\disc},\Gamma_i}E_{S_{\disc},\Gamma_i}}\circ\zeta_\sigma\\
&\hat\omega^{\Gamma_i}_{\LD}\circ\alpha_{H_{\Gpi,P}^+H_{\Gpi,P}}(t)\circ \alpha_{E^+_{S_{\disc},\Gamma_i}E_{S_{\disc},\Gamma_i}}=\hat\omega^{\Gamma_i}_{\LD}=\hat\omega^{\Gamma_i}_{\LD}\circ\alpha_{E^+_{S_{\disc},\Gamma_i}E_{S_{\disc},\Gamma_i}}\circ\alpha_{H_{\Gpi,P}^+H_{\Gpi,P}}(t)
\eqs  for all $\sigma\in \Diff(\PD_{\Gamma_i}^{\breve S_{\disc}})$, $t\in\R$, a subgraph $\Gpi$ of $\Gamma_i$ and all $1\leq i<\infty$.

There is a problem of convergence of the limit state on the localised holonomy-flux cross-product $^*$-algebra presented in proposition \ref{prop fulldiffeobroken}. Consequently, the limit state $\hat\omega_{\LD}$ on $\ZD(\Ab^{\disc}_{\Gamma_i})\rtimes_L \mathfrak{E}_{\breve S_{\disc},\Gamma_i}$ has to be analysed further. The hope is that for a suitable modified (or localised) quantum Hamilton constraint derived from
\beqs \hat H^+\hat H:=\limN\sum_{i=1}^NH_{\Gamma_i}^+H_{\Gamma_i}
=\limG\sum_{\Gp\in\PD_{\Gamma_i}}H_{\Gp}^+H_{\Gp}
\eqs the state $\omega_{\LD}$ satisfies
\beq\label{eq fullHamiltoninv} \hat\omega_{\LD}\circ\alpha_{\hat H^+\hat H}=\limG\sum_{\Gp\in\PD_{\Gamma_i}}\hat\omega_{\LD}^{\Gamma_i}\circ\alpha_{H^+_{\Gp}H_{\Gp}}
=\hat\omega_{\LD}
\eq Summarising, in this situation the state $\hat\omega_{\LD}$ would be invariant under the automorphisms inherited by the modified quantum Hamilton $H$, but the state is only invariant under a finite set of exceptional graph-diffeomorphisms. 
Despite this fact a localised quantum diffeomorphism constraint is defined as follows. First recall the construction presented in \cite[Sec.: 5]{Kaminski2},\cite[Sec.: 7.3]{KaminskiPHD}. There some certain operators are developed in the situation of $C^*$-algebras. Apart from $C^*$-properties the following objects can be analysed. Similarly define an operator, which depends on a bisection in $\mathfrak{B}(\PD_{\bar\Gamma})$ and which is $C(\Ab_{\bar\Gamma})$-valued, and denote this operator by $D_{\bar \Gamma}^ \sigma$. The set of all these operators is denoted by $\Df_{\breve S_{\disc},\Gamma}$. Furthermore there is an operator, which depends on a bisection in $\mathfrak{B}_{\breve S_{\disc},\diff}(\PD_{\Gamma}^{\breve S_{\disc}})$ and which is $\ZD(\Ab^{\breve S_{\disc}}) \rtimes \mathfrak{E}_{\breve S_{\disc}}$-valued, and this operator is denoted by $D_{\breve S_{\disc},\Gamma}^\sigma$. The adjoint operator is denoted by $D_{\breve S_{\disc},\Gamma}^{\sigma,*}$. The set of all these operators is denoted by $\Df_{\bar\Gamma}$. For each graph $\Gamma_i$ of a family of graphs there exists a generating system $\mathfrak{B}^{\Gamma_i}_{\breve S,\diff}(\PD_{\Gamma_i}^{\breve S_{\disc}})$ of bisections for this graph.
Then set
\beqs D_{\breve S, \Gamma_i}^+D_{\breve S_{\disc}, \Gamma_i}:=\sum_{\sigma_l\in\mathfrak{B}^{\Gamma_i}_{\breve S_{\disc},\diff}(\PD_{\Gamma_i}^{\breve S_{\disc}})}D^{\sigma_l,*}_{\breve S_{\disc},\Gp}D^{\sigma_l}_{\breve S_{\disc},\Gp}
\eqs for every subgraph $\Gp$ of $\Gamma_i$. The sum over all graphs of a family of graphs defines the \textbf{localised quantum diffeomorphism constraint}.  
The linear hull over all graphs of a family of graphs of all elements of the set $\Df_{\breve S_{\disc},\Gamma}$, the set $\Df_{\bar\Gamma}$ and the set of all quantum Hamilton constraints restricted to a graph $\Gamma$ forms the $^*$-algebra $\Con$ of quantum constraints. Note that, this algebra is not a subalgebra of the localised holonomy-flux cross-product $^*$-algebra associated to a discretised surface set. Finally, the \textbf{modified quantum Master constraint} $\MM$ is defined by the sum of the modified quantum Hamilton constraint and the localised quantum diffeomorphism constraint.

The localised holonomy-flux cross-product $^*$-algebra can be enlarged such that this algebra will be a subalgebra. This algebra will be based on the cross-product construction once more and consequently will be called the \textbf{localised holonomy-flux-graph-diffeomorphism cross-product $^*$-algebra} associated to a discretised surface set. It will contain all finite graph-diffeomorphisms. Note that, the modified quantum Hamilton constraint is not contained in this algebra, but it will be in a suitable sense be affilliated with. Now, Dirac states and Dirac observables have to be analysed.

Assume that, $\Ss_D$ denotes a set of Dirac states on the localised holonomy-flux-graph-diffeomorphism cross-product $^*$-algebra $\Alg$. It is not obvious that Dirac observables can be easily defined, since the set generated by all quantum constraints in $\Con$ defines a closed left and right ideal in $\Alg$. Assume that $\OD_{D}$ is the algebra of Dirac observables, which is a subalgebra of  the localised holonomy-flux-graph-diffeomorphism cross-product $^*$-algebra.
Then 
\beqs\OD^{\alpha}_D:=\{A\in \OD_{D}:\alpha_{\MM}(t)(A)=A,\forall t\in\R\}
\eqs defines a \textbf{localised $^*$-algebra of complete quantum observables for surfaces}. Hence the localised holonomy-flux-graph-diffeomorphism cross-product $^*$-algebra associated to a discretised surface set is supposed to be a physical algebra in the context of \cite{Kaminski0}.
\end{problem}

Finally, a short remark with respect to $C^*$-algebras is stated. The \textbf{localised holonomy-flux cross-product $C^*$-algebra for surfaces} is constructable as the inductive limit $C^*$-algebra of the inductive family of $C^*$-algebras 
$\{ C(\Ab_{\bar\Gamma})\otimes C(\bar G_{\breve S_{\disc},\Gamma})\rtimes\Ab^{\disc}_{\Gamma}\}$ for a suitable set $\breve S_{\disc}$ of discretised surfaces associated to a surface set $\breve S$ with appropriate properties with respect to the inductive limit of the family of graphs. The ideas are derived from to the holonomy-flux cross-product $C^*$-algebra presented in \cite{Kaminski2,KaminskiPHD}.
There it has been also given an enlargement of the holonomy-flux cross-product $C^*$-algebra, which contains finite diffeomorphisms. The generators defined by the quantum diffeomorphisms are not contained in this algebra but affiliated with.   
This idea will be used in a future work for a similar enlargment of the localised holonomy-flux cross-product $C^*$-algebra for surfaces. Consequently a physical algebra, which is indeed a $C^*$-algebra, can be constructed in this way. 
 
\newpage
\section{Comparison table}\label{sec comparison}

 A comparison of the localised holonomy-flux cross-product $^*$-algebra and the holonomy-flux cross-product $^*$-algebra is presented in the next table. Summarising the construction is based on the algebra of continuous functions depending on holonomies along paths, which is a left (or right-) module for the enveloping flux algebra for surfaces. Consequently, certain algebras can be derived as abstract cross-product algebras. The differences appear by the choice of the set of paths, and hence the construction of the quantum configuration space. Therefore different holonomy algebras are considered. In particluar the algebras distinguish with respect to the multiplication operation of the elements of these algebras, and their localisation or non-localisation with respect to a set of discretised surfaces associated to surface sets.  

\begin{landscape}
\begin{longtable}[ht]{|l|l|l|l|}\caption{Comparison of $^*$-algebras}\label{tablecompalg}\\
\hline &&&\\
&holonomy-flux algebra &  holonomy-flux cross-product algebra & localised holonomy-flux cross-product algebra\\
&&&\\ \hline\hline &&&\\
ingredients& principal fibre bundle $P(\Sigma,G,\pi)$&principal fibre bundle $P(\Sigma,G,\pi)$&principal fibre bundle $P(\Sigma,G,\pi)$\\[3pt]
&surfaces with codim. $1$&set of finite set $\breve S$ of surfaces with codim. $1$&set of finite set $\breve S_{\disc}$ of discretised surfaces \\[3pt]
assumption&$G$ compact connected Lie group&$G$ compact connected Lie group&$G$ compact connected Lie group\\[3pt]
ingredients& &fin. path groupoid $\PD_\Gamma\Sigma$ over $V_\Gamma$ &fin. path groupoid $\PD_\Gamma\Sigma$ over $V_\Gamma$\\[3pt]
& path groupoid $\PD$ over $\Sigma$& path groupoid $\PD$ over $\Sigma$ & path groupoid $\PD$ over $\Sigma$\\[3pt]
ingredients& graph $\Gamma$ &fin. orient.-preserv. graph sys. $\PD_\Gamma^{\op}$ assoc. to $\Gamma$ &sets of paths starting or ending at disc. surfaces.\\[3pt]
&&&graphs not located at disc. surfaces\\[3pt]
inductive limit&inductive family of fin. path groupoids &ind. family of fin. orient.-preserv. graph sys. & inductive family of graphs \\[3pt]
holonomy map& groupoid morph. $\bar A$ from $\PD$ to $G$ & holonomy map $\ho_\Gamma$ from $\PD_\Gamma\Sigma$ to $G$&holonomy map $\ho_\Gamma$ from $\PD_\Gamma\Sigma$ to $G$\\[3pt]
&  & holonomy map $\ho_\Gamma$ from $\PD_\Gamma$ to $G^{\vert\Gamma\vert}$&holonomy map $\ho_\Gamma$ from $\PD_\Gamma$ to $G^{\vert\Gamma\vert}$\\[5pt]
config. space&$\Ab_\Gamma$ and proj. limit space $\Ab$& $\Ab_\Gamma$ and proj. limit space $\Ab$& $\Ab_{\disc,\Gamma}\times\Ab_{\bar\Gamma}$\\[3pt]
assumption&identification of sets of paths in $\PD_\Gamma\Sigma$&natural identif. of sets of indep. paths in $\PD_\Gamma\Sigma$&non-stand. identif. of sets of paths\\[3pt]
&&& located at disc. surfaces,\\[3pt]
&&& natural identif. of sets of indep. paths \\[3pt]
&&& not located at disc. surfaces\\[5pt]
&&&enveloping alg. $\LAb_{\disc,\Gamma}$ of Lie alg. assoc. to $\Ab^{\disc}_{\Gamma}$\\[5pt]
Hilbert space& $\HS_{\text{AL}}$&$\HS_\Gamma$ and ind. limit Hilbert space $\HS_\infty$& $\HS_\Gamma^{\disc}=L^2(\Ab_\Gamma^{\disc},\mu_\Gamma^{\disc})$ and $\HS_{\bar\Gamma}=L^2(\Ab_{\bar\Gamma},\mu_{\bar\Gamma})$ \\[3pt]
&&&and limit Hilbert spaces $\HS_{\disc}$ and $\HS_{\loc}$\\[5pt]
diffeomorphisms& $\varphi$ diffeomorphism on $\Sigma$, $(\Phi,\varphi)$& fin. path- or graph- diffeom. $(\Phi_\Gamma,\varphi_\Gamma)$&fin. path- or graph- diffeom. $(\Phi_\Gamma,\varphi_\Gamma)$\\[3pt]
&&group $\mathfrak{B}(\PD_\Gamma\Sigma)$ or group $\mathfrak{B}(\PD_\Gamma)$ of bisections &group $\mathfrak{B}(\PD_\Gamma\Sigma)$ or group $\mathfrak{B}(\PD_\Gamma)$ of bisections\\[5pt]
mom. space& expon. smearing vector field $E_{S,F}$& the Lie flux algebra $\bar \go_{\breve S,\Gamma}$ or $\bar \go_{\breve S,\Gamma_\infty}=:\bar \go_{\breve S}$& the localised Lie flux algebra $\bar \go^{\loc}_{\breve S_{\disc},\Gamma}$ or $\bar \go^{\loc}_{\breve S_{\disc}}$\\[5pt]
&  on a fibre of $P$ & the flux enveloping algebra $\bar \E_{\breve S,\Gamma}$ or $\bar \E_{\breve S,\Gamma_\infty}=:\bar \E_{\breve S}$& the localised flux enveloping algebra $\bar \E^{\loc}_{\breve S_{\disc},\Gamma}$ or $\bar \E^{\loc}_{\breve S_{\disc}}$\\[5pt]
\hline &&&\\
$C^*$-algebra & $C(\Ab_\Gamma)$ $\&$ sup-norm&$C(\Ab_\Gamma)$ $\&$ sup-norm&$C(\Ab_{\bar\Gamma})$ $\&$ $L^2$-norm\\[3pt]
& inductive limit $C^*$-algebra $C(\Ab)$& inductive limit $C^*$-algebra $C(\Ab)$&  inductive limit $C^*$-algebra $C(\Ab_{\loc})$\\[5pt]
\hline\newpage\hline &&&\\
$C^*$-algebra&&&$\bigotimes_i C^*(\Ab_{\disc,\gamma_i})$ $\&$ $L^2$-norm\\[3pt]
&&& infinite $C^*$-tensor algebra $C^*(\Ab^{\disc})$ \\[5pt]
\hline &&&\\
Hilbert space & $\pi(f)\psi=f\cdot \psi$ for $f \in C(\Ab)$& $\Phi_M(f)\psi=f\cdot\psi$ for $f \in C(\Ab)$ and for $\psi\in\HS_\infty$ & $\Phi_M(f)\psi=f\cdot\psi$ for $f \in C^*(\Ab^{\disc})$ and for $\psi\in\HS_{\disc}$ \\[3pt]
operators & & & $\Phi_M(f)\psi=f\cdot\psi$ for $f \in C(\Ab_{\loc})$ and for $\psi\in\HS_{\loc}$ \\[5pt]
& for $\psi\in\HS_{\text{AL}}$ & $\pi(\exp(tE_{S}(\gamma)))\psi=U_t(E_{S}(\gamma))\psi$ for $\psi\in\HS_\infty$& $\pi(\exp(tE_{S_{\disc}}^\gamma))\psi=U_t(E_{S_{\disc}}^\gamma)\psi$ for $\psi\in\HS_{\disc}$ \\[3pt]
& & & $\pi(\exp(tE_{S_{\disc}}^\gamma))\psi=\psi$ for $\psi\in\HS_{\loc}$ \\[5pt]
 & $\pi(E_{S,F})\psi=\frac{\dif}{\dif t}\Big\vert_{t=0}\psi\circ\theta_t(F)=:X_{S}\psi$& $\pi(E_{S}(\gamma)^+E_{S}(\gamma))\psi=-i \frac{\dif}{\dif t}\Big\vert_{t=0}U_t(E_{S}(\gamma)^+E_{S}(\gamma))\psi$ & where $E_{S_{\disc}}^\gamma:=E_{S_{\disc}}(\gamma)^+E_{S_{\disc}}(\gamma)$\\[3pt]
&&$\pi(E_{S}(\gamma)^+E_{S}(\gamma))=:\dif U(E_S(\gamma))$&$\pi(E_{S_{\disc}}^\gamma)\psi=-i \frac{\dif}{\dif t}\Big\vert_{t=0}U_t(E_{S_{\disc}}^\gamma)\psi=:\dif U(E_{S_{\disc}}^\gamma)\psi$\\[5pt]
& for $\psi\in D(E_{S,F})$ & for $\psi\in D(E_{S}(\gamma)^+E_{S}(\gamma))$&for $\psi\in D(\dif U(E_{S_{\disc}}^\gamma))$\\[5pt]
&&& self-adj. quantum Hamilton part $H_{\Gamma,P}^+H_{\Gamma,P}$ on $\HS_\Gamma^{\disc}$\\
\hline &&&\\
$^*$-algebras& & enveloping algebra $\bar\E_{\breve S,\Gamma}$ of $\bar\go_{\breve S,\Gamma}$ &enveloping algebra $\bar\E^{\loc}_{\breve S_{\disc},\Gamma}$ of $\bar\go^{\loc}_{\breve S_{\disc},\Gamma}$\\[3pt]
&&with involution $^+$ s.t. &with involution $^+$ s.t.\\[3pt]
&& $E_S(\Gamma)^+=-E_S(\Gamma)$ for all $E_S(\Gamma)\in\bar\go_{\breve S,\Gamma}$& $E_{S_{\disc}}(\Gamma)^+=-E_{S_{\disc}}(\Gamma)$ for all $E_{S_{\disc}}(\Gamma)\in\bar\go^{\loc}_{\breve S_{\disc},\Gamma}$\\[5pt]
&&& $^*$-algebra $\ZD(\Ab_\Gamma^{\disc})$ of central functions on $\Ab_\Gamma^{\disc}$\\[5pt]
&$^*$-algebra $\DD$  of differential op. on $\HS_{\text{AL}}$& $^*$-algebra $\DD(\bar G_{\breve S,\Gamma})$ of differential op. on $\HS_\Gamma$& $^*$-algebra $\DD(\bar G^{\loc}_{\breve S,\Gamma})$ of differential op. on $\HS_\Gamma^{\loc}$\\[5pt]
&enveloping algebra $\Alg_{\text{HF}}$ of $C(\Ab)\times \DD$ &$C^\infty(\Ab_\Gamma)\rtimes_X\bar\E_{\breve S,\Gamma}$  for $X=L,R$& $C(\Ab_{\bar\Gamma})\otimes_{\text{min}} C^\infty(\Ab_\Gamma^{\disc})\rtimes\bar\E^{\loc}_{\breve S_{\disc},\Gamma}$\\[3pt]
& $C^\infty(\Ab_\Gamma^{\disc})\rtimes\bar\E^{\loc}_{\breve S,\Gamma}$ with multiplication $\cdot$ &with multiplication $\cdot_X$&with multiplication $\cdot_L$\\[5pt]
&&$C^\infty(\Ab)\rtimes_X\bar\E_{\breve S,\Gamma_\infty}$  for $X=L,R$&$C(\Ab_{\loc})\otimes_{\text{min}} C^\infty(\Ab^{\disc})\rtimes\bar\E^{\loc}_{\breve S_{\disc}}$\\[5pt]
\hline &&&\\
$^*$-representation&$\pi(f,X_S)\psi=f\cdot X_S\psi$ & $\pi(f_\Gamma,E_S(\Gamma)^+E_S(\Gamma))\psi$& $\pi(f_\Gamma,E_{S_{\disc}}(\Gamma)^+E_{S_{\disc}}(\Gamma))\psi$\\[3pt]
&& $=-i \frac{\dif}{\dif t}\Big\vert_{t=0}U_t(E_{S}(\Gamma)^+E_{S}(\Gamma)) f_\Gamma \cdot \psi $& $=-i \frac{\dif}{\dif t}\Big\vert_{t=0}U_t(E_{S_{\disc}}(\Gamma)^+E_{S_{\disc}}(\Gamma)) f_\Gamma \cdot \psi $\\[5pt]
&&\hspace{15pt}$-i f_\Gamma\cdot \frac{\dif}{\dif t}\Big\vert_{t=0}U_t(E_{S}(\Gamma)^+E_{S}(\Gamma))\psi$&\hspace{15pt}$-i f_\Gamma\cdot \frac{\dif}{\dif t}\Big\vert_{t=0}U_t(E_{S_{\disc}}(\Gamma)^+E_{S_{\disc}}(\Gamma))\psi$\\[8pt]
&for $\psi\in D(E_{S,F})$ &for $\psi\in C^\infty(\Ab_\Gamma)$& for $\psi\in C^\infty(\Ab_\Gamma^{\disc})$\\[5pt]
\hline \newpage\hline &&&\\
automorphisms&&$\alpha_{E_{S}(\Gamma)^+E_{S}(\Gamma)}\in\Aut(C^\infty( \Ab))$&$\alpha_{E_{S_{\disc}}(\Gamma)^+E_{S_{\disc}}(\Gamma)}\in\Aut(C^*( \Ab))$ \\[3pt]
&&for $E_{S}(\Gamma)^+E_{S}(\Gamma)\in \ZD(\bar\E_{\breve S})$& for $E_{S_{\disc}}(\Gamma)^+E_{S_{\disc}}(\Gamma)\in  \mathfrak{z}_{\breve S_{\disc}}$\\[5pt]
&&& $\beta_{\ha_{\disc,\Gamma}}\in \Aut(C^*(\Ab_\Gamma^{\disc}))$ for $\exp(\ha_{\disc,\Gamma})$ in $\Ab_\Gamma^{\disc}$\\[5pt]
&&& $\tilde\alpha_{\ha_{\disc,\Gamma}}\in \Aut(\ZD(\Ab_\Gamma^{\disc}))$ for $\exp(\ha_{\disc,\Gamma})$ in $\Ab_\Gamma^{\disc}$ \\[3pt]
&&&$\alpha_{H_{\Gamma,P}^+H_{\Gamma,P}}\in\Aut(\ZD(\Ab^{\disc}_\Gamma)\rtimes \mathfrak{z}_{\breve S_{\disc},\Gamma})$ \\[5pt]
&&& for $\exp(H_{\Gamma,P}^+H_{\Gamma,P})$ in center of $\Ab^{\disc}_\Gamma$\\[5pt]
& $\zeta_{(\varphi,\Phi)}\in\Aut(\Alg_{\text{HF}})$&$\zeta_\sigma\in\Aut(C^\infty(\Ab)\rtimes_X\bar\E_{\breve S})$ for certain  $\sigma\in\mathfrak{B}(\PD_\Gamma^{\op})$&$\zeta_\sigma\in\Aut(C^\infty(\Ab^{\disc})\rtimes\mathfrak{z}_{\breve S_{\disc}})$ for certain $\sigma\in\mathfrak{B}(\PD_\Gamma^{\breve S_{\disc}})$\\[5pt]
\hline &&&\\
states&&&state $\omega_{M}$ on $C(\Ab_{\loc})$ s.t. \\[3pt]
&&&$\omega_{M}\circ\zeta_\sigma=\omega_{M}$ for $\sigma\in\mathfrak{B}(\PD_{\bar\Gamma})$\\[5pt]
&&&KMS-state $\tilde\omega_\LD^\Gamma$ on $\ZD(\Ab^{\disc}_\Gamma)$ w.r.t. automorph.  $\tilde\alpha_{\ha_{\disc},\Gamma}$\\[5pt]
&&&KMS-state $\tilde\omega_\LD^\Gamma$ on $\ZD(\Ab^{\disc}_\Gamma)$ w.r.t. automorph.  $\alpha_{H_{\Gamma,P}^+H_{\Gamma,P}}$\\[5pt]
&&& and s.t. $\tilde\omega^{\Gamma}_{\LD}\circ\alpha_{H_{\Gamma,P}^+H_{\Gamma,P}}=\tilde\omega^{\Gamma}_{\LD}$\\[8pt]
& unique state $\omega$ on $\Alg_{\text{HF}}$ s.t. &unique state  $\bar\omega_M$  on $C^\infty(\Ab)\rtimes_X\ZD(\bar\E_{\breve S})$ s.t. &state $\hat\omega_{\LD}$ on $\ZD(\Ab^{\disc}_{\Gamma}) \rtimes \mathfrak{z}_{\breve S_{\disc},\Gamma}$ s.t.\\[3pt]
&$\omega\circ\zeta_{(\varphi,\Phi)}=\omega$ &$\bar\omega_M\circ\zeta_\sigma=\bar\omega_M$ for certain $\sigma\in\mathfrak{B}(\PD_\Gamma^{\op})$& $\hat\omega_{\LD}\circ\zeta_\sigma=\hat\omega_{\LD}$ for certain $\sigma\in\mathfrak{B}(\PD_\Gamma^{\breve S_{\disc}})$\\[3pt]
&and $\omega(f,X_S)=0$ & and $\bar\omega_M(f_,E_S(\Gamma))=0$&and $\hat\omega_{\LD}(f_,E_{S_{\disc}}(\Gamma))=0$\\[5pt]
&&&\\[5pt]
&&&\\[5pt]
\hline
\end{longtable}
\end{landscape}

\section*{Acknowledgements}
The work has been supported by the Emmy-Noether-Programm (grant FL 622/1-1) of the Deutsche Forschungsgemeinschaft.

\end{document}